\pdfoutput=1
\RequirePackage{ifpdf}
\ifpdf 
\documentclass[pdftex]{sigma}
\else
\documentclass{sigma}
\fi

\usepackage{float,ytableau,tikz,extarrows,mathdots}
\usetikzlibrary{calc,arrows,math,snakes}

\newcommand{\n}{\noindent}
\newcommand{\Z}{\mathbb{Z}}
\newcommand{\N}{\mathbb{N}}
\newcommand{\R}{\mathbb{R}}
\newcommand{\tr}{\text{tr}}
\newcommand{\x}{\times}

\numberwithin{equation}{section}

\newtheorem{thm}{Theorem}[section]
\newtheorem{lem}[thm]{Lemma}
\newtheorem{cor}[thm]{Corollary}

\newtheorem{prop}[thm]{Proposition}

{\theoremstyle{definition}
\newtheorem{defn}[thm]{Definition}

\newtheorem{ex}[thm]{Example}
\newtheorem{rem}[thm]{Remark}
\newtheorem*{Note}{Note}
}

\begin{document}
\allowdisplaybreaks

\newcommand{\arXivNumber}{2106.07129}

\renewcommand{\PaperNumber}{063}

\FirstPageHeading

\ShortArticleName{A Path-Counting Analysis of Phase Shifts in Box-Ball Systems}

\ArticleName{A Path-Counting Analysis of Phase Shifts\\ in Box-Ball Systems}

\Author{Nicholas M. ERCOLANI and Jonathan RAMALHEIRA-TSU}

\AuthorNameForHeading{N.M.~Ercolani and J.~Ramalheira-Tsu}

\Address{Department of Mathematics, University of Arizona, USA}
\Email{\href{mailto:ercolani@math.arizona.edu}{ercolani@math.arizona.edu}, \href{mailto:jramalheiratsu@arizona.edu}{jramalheiratsu@arizona.edu}}
\URLaddress{\url{http://www.math.arizona.edu/~ercolani},\newline
\hspace*{10.5mm}\url{http://www.math.arizona.edu/~jramalheiratsu}}

\ArticleDates{Received April 08, 2022, in final form August 20, 2022; Published online August 25, 2022}

\Abstract{In this paper, we perform a detailed analysis of the phase shift phenomenon of the classical soliton cellular automaton known as the box-ball system, ultimately resulting in a statement and proof of a formula describing this phase shift. This phenomenon has been observed since the nineties, when the system was first introduced by Takahashi and Satsuma, but no explicit global description was made beyond its observation. By using the Gessel--Viennot--Lindstr\"om lemma and path-counting arguments, we present here a novel proof of the classical phase shift formula for the continuous-time Toda lattice, as discovered by Moser, and use this proof to derive a discrete-time Toda lattice analogue of the phase shift phenomenon. By carefully analysing the connection between the box-ball system and the discrete-time Toda lattice, through the mechanism of tropicalisation/dequantisation, we translate this discrete-time Toda lattice phase shift formula into our new formula for the box-ball system phase shift.}


\Classification{17B80; 37J70; 37K10}

\vspace{-2mm}

\section{Introduction}

In classical solitary wave theory there are three well-known signature features of the inherent nonlinearity of these waves that are inter-related. The first is that in long time, forward and backward, an exact\footnote{\textit{Exact} here means an $n$-soliton wave-form that propagates without ``radiating'' any dispersive oscillations.} solitary wave asymptotically separates into distinct localized wave-forms, referred to as {\it masses}, that maintain their form as they propagate further (both forward or backward). The second feature is that in long time these individual masses are ordered by their peak amplitudes, travelling with individual speeds asymptotically proportional to the square root of the amplitude. This is referred to as the {\it sorting property}. Finally, at intermediate times, when these individual masses interact, they pass through one another with their wave-forms unchanged but experiencing a relative {\it phase shift}. We illustrate this in the case of the KdV equation~\cite{bib:dj} with a space-time plot in Figure~\ref{kdvphase}.

\begin{figure}[h!] \centering
 \includegraphics[width=0.43\textwidth]{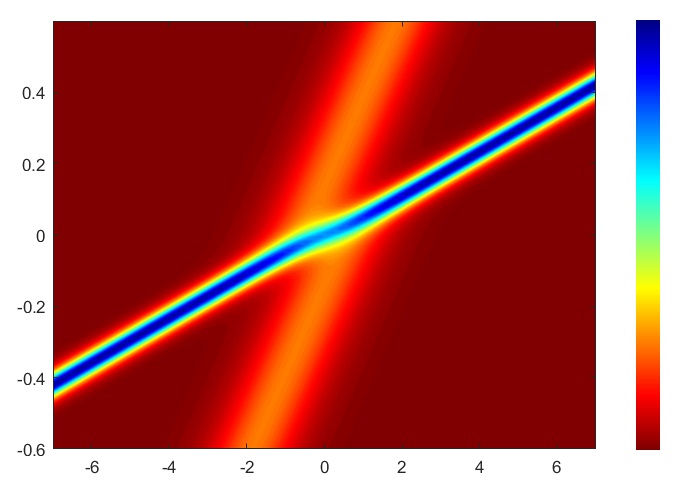}
 \caption{A 2-soliton solution to the KdV equation, $U(x,t)$, with $x$ along the horizontal axis and $t$ along the vertical. The gradient from red to blue represents small values to large. We see two coherent waves: one with a higher peak (blue) and another with a lower peak (orange). The larger wave, initially behind, collides with the smaller, and emerges after the collision in front. Note also the phase shift which corresponds to the ``breaks'' in the linear evolution of the peaks.} \label{kdvphase}
\end{figure}

The sorting and phase shift properties are characteristic of {\it integrable} evolution equations exhibiting coherent states of nonlinear wave type. In the next two sections of this paper, we will discuss two other examples of these phenomena: the finite Toda lattice and the box-ball system (BBS). They turn out to be related to one another through a type of semi-classical limit known as {\it tropicalization} or {\it ultra-discretization} as has been well-documented \cite{bib:l,bib:lmrs,bib:v}.

The precise sorting and phase shift formulas for the Toda lattice have been known for some time going back to Moser~\cite{bib:moser}. It is also the case that early work by Tokihiro et al.~\cite{bib:tns} derives results related to the explicit BBS phase shift formula that we will present. A comparison to the work of~\cite{bib:tns} is given later in Section~\ref{subsection:tokihiro} which discusses this and other related work. In this paper we give a new proof of the classical phase shift formula for Toda within a representation theoretic framework closely related to Kostant's generalized Toda systems~\cite{bib:kostant}.
However, rather than basing the proof on representation theoretic arguments, as Kostant does, our approach is based on a path counting construction. The potential advantages of this are, we feel, two-fold. First, Kostant's extension of the phase shift formula depends on somewhat complicated determinantal formulas, whereas what we present is completely geometric and combinatorial thus making generalizations more natural. Secondly, we found, subsequent to our initial derivations, that the directed graphs we employ for our path counting are, in fact, those used by Postnikov~\cite{bib:mr} to coordinatize the totally positive Grassmannians and their tropicalization. This reinforces our belief that this is the right geometric picture for our second main result which is to derive the explicit, general phase shift formula for BBS by direct tropical limit of the classical Toda formula.\looseness=-1

\subsection{Outline}
In the next two sections we will present the necessary background on the box-ball system and the Toda lattice to be able to explicitly describe the phase shift phenomena in these systems. Sections~\ref{crystal} and~\ref{repntheorysec} provide the necessary representation theoretic background that underlies our arguments. These are presented in such a way as to make it clear how our analysis would extend to the setting of generalized Toda lattices as introduced by Kostant~\cite{bib:kostant2}. In Section~\ref{todaphaseshiftsec} we present our novel derivation of the classical phase shift formula for both the classical (Theorem~\ref{classicaltodaphase}) and discrete-time (Theorem~\ref{discretetodaphase}) Toda systems in terms of the Gessel--Viennot--Lindstr\"om (GVL) lemma. Finally, in Section \ref{bbsphasesec}, we use this to deduce the phase shift formula for BBS (Theorem~\ref{maintheorembbsphaseshiftlim}). The appendices expand on details for the path-counting construction of our GVL formul{\ae}.\looseness=-1

\section{Phase shift in box-ball dynamics}
\subsection{Box-ball systems}
A cellular automaton is a special type of discrete dynamical system with both discrete time steps and a discrete (in fact finite) number of states. Of particular interest is the box-ball system (BBS) which was introduced in 1990 by Takahashi and Satsuma~\cite{bib:ts}.

\begin{defn}\label{bbsoneszeroesdefn}
The (basic) box-ball system consists of a one-dimensional infinite array of boxes with a finite number of the boxes filled with balls, and no more than one ball in each box (see, for example, Figure \ref{firstbbsexfordef}).

More formally, the phase space of this system, which we denote by $\text{BBS}$, can be identified with the space of binary sequences
$\{0,1\}^\Z$, with all but finitely many entries equal to zero, so that $1$'s correspond to filled boxes and $0$'s to empty boxes.
\end{defn}

\begin{figure}[h!]
\centering
\tikz[scale=0.6]{
\foreach \x in {0,1,2,3,4,5,6,7,8,9,10,11,12,13,14,15}
{\draw[fill=white] (\x,3) -- (\x+1,3) -- (\x+1,4) -- (\x,4) -- cycle;			
}
\foreach \x in {1,2,3,7,10,11,13}
{\draw[fill=white] (\x,3) -- (\x+1,3) -- (\x+1,4) -- (\x,4) -- cycle;			
\fill[cyan] (\x+0.5,3.5) circle (0.25);
}
\foreach \x in {}
{\draw[fill=white] (\x,3) -- (\x+1,3) -- (\x+1,4) -- (\x,4) -- cycle;			
\fill[red] (\x+0.5,3.5) circle (0.25);
}
\foreach \x in {16}
{\draw[fill=white,white] (\x,3) -- (\x+2,3) -- (\x+2,4) -- (\x,4) -- cycle;
\draw[-] (\x,3) -- (\x,4);
\draw[-] (\x,3) -- (\x+2,3);
\draw[-] (\x,4) -- (\x+2,4);
\draw[-] (\x+1,3) -- (\x+1,4);
\node at (\x+1.5,3.5) {$\cdots$};
}
\foreach \x in {0}
{\draw[fill=white,white] (\x,3) -- (\x-2,3) -- (\x-2,4) -- (\x,4) -- cycle;
\draw[-] (\x,3) -- (\x,4);
\draw[-] (\x,3) -- (\x-2,3);
\draw[-] (\x,4) -- (\x-2,4);
\draw[-] (\x-1,3) -- (\x-1,4);
\node at (\x-1.5,3.5) {$\cdots$};
}
}
\caption{A box-ball state.}\label{firstbbsexfordef}
\end{figure}

\subsection{The box-ball evolution}\label{bbesubsec}
A simple evolution rule is provided for the box-ball dynamics:\index{Basic Box-Ball Evolution}\index{Box-Ball System}
\begin{enumerate}\itemsep=0pt
\item[(1)] Take the left-most ball that has not been moved and move it to the left-most empty box to its right.
\item[(2)] Repeat (1) until all balls have been moved precisely once.
\end{enumerate}

Since the algorithm requires one to know which balls have been moved, we can, without technically changing the algorithm, introduce a colour-coding based on whether balls have moved or not. Balls will be blue until they have moved, after which they will become red. When all balls are red, the colours should be reset to blue, ready for the next time step. Or, equivalently, a $0$-th step of colouring all balls blue should be prescribed. We will use the latter for a minor benefit in brevity. Below is an example of the evolution with this colour-coding, with each ball move separated into a sub-step:

\begin{figure}[h!]
\centering
\tikz[scale=0.6]{
\foreach \x in {0,1,2,3,4,5,6,7,8,9,10,11,12,13,14,15}
{\draw[fill=white] (\x,3) -- (\x+1,3) -- (\x+1,4) -- (\x,4) -- cycle;			
}
\foreach \x in {1,2,3,7,10,11,13}
{\draw[fill=white] (\x,3) -- (\x+1,3) -- (\x+1,4) -- (\x,4) -- cycle;			
\fill[cyan] (\x+0.5,3.5) circle (0.25);
}
\foreach \x in {}
{\draw[fill=white] (\x,3) -- (\x+1,3) -- (\x+1,4) -- (\x,4) -- cycle;			
\fill[red] (\x+0.5,3.5) circle (0.25);
}
\foreach \x in {16}
{\draw[fill=white,white] (\x,3) -- (\x+2,3) -- (\x+2,4) -- (\x,4) -- cycle;
\draw[-] (\x,3) -- (\x,4);
\draw[-] (\x,3) -- (\x+2,3);
\draw[-] (\x,4) -- (\x+2,4);
\draw[-] (\x+1,3) -- (\x+1,4);
\node at (\x+1.5,3.5) {$\cdots$};
}
\foreach \x in {0}
{\draw[fill=white,white] (\x,3) -- (\x-2,3) -- (\x-2,4) -- (\x,4) -- cycle;
\draw[-] (\x,3) -- (\x,4);
\draw[-] (\x,3) -- (\x-2,3);
\draw[-] (\x,4) -- (\x-2,4);
\draw[-] (\x-1,3) -- (\x-1,4);
\node at (\x-1.5,3.5) {$\cdots$};
}
}
\tikz[scale=0.6]{
\foreach \x in {0,1,2,3,4,5,6,7,8,9,10,11,12,13,14,15}
{\draw[fill=white] (\x,3) -- (\x+1,3) -- (\x+1,4) -- (\x,4) -- cycle;			
}
\foreach \x in {2,3,7,10,11,13}
{\draw[fill=white] (\x,3) -- (\x+1,3) -- (\x+1,4) -- (\x,4) -- cycle;			
\fill[cyan] (\x+0.5,3.5) circle (0.25);
}
\foreach \x in {4}
{\draw[fill=white] (\x,3) -- (\x+1,3) -- (\x+1,4) -- (\x,4) -- cycle;			
\fill[red] (\x+0.5,3.5) circle (0.25);
}
\foreach \x in {16}
{\draw[fill=white,white] (\x,3) -- (\x+2,3) -- (\x+2,4) -- (\x,4) -- cycle;
\draw[-] (\x,3) -- (\x,4);
\draw[-] (\x,3) -- (\x+2,3);
\draw[-] (\x,4) -- (\x+2,4);
\draw[-] (\x+1,3) -- (\x+1,4);
\node at (\x+1.5,3.5) {$\cdots$};
}
\foreach \x in {0}
{\draw[fill=white,white] (\x,3) -- (\x-2,3) -- (\x-2,4) -- (\x,4) -- cycle;
\draw[-] (\x,3) -- (\x,4);
\draw[-] (\x,3) -- (\x-2,3);
\draw[-] (\x,4) -- (\x-2,4);
\draw[-] (\x-1,3) -- (\x-1,4);
\node at (\x-1.5,3.5) {$\cdots$};
}
}
\tikz[scale=0.6]{
\foreach \x in {0,1,2,3,4,5,6,7,8,9,10,11,12,13,14,15}
{\draw[fill=white] (\x,3) -- (\x+1,3) -- (\x+1,4) -- (\x,4) -- cycle;			
}
\foreach \x in {3,7,10,11,13}
{\draw[fill=white] (\x,3) -- (\x+1,3) -- (\x+1,4) -- (\x,4) -- cycle;			
\fill[cyan] (\x+0.5,3.5) circle (0.25);
}
\foreach \x in {4,5}
{\draw[fill=white] (\x,3) -- (\x+1,3) -- (\x+1,4) -- (\x,4) -- cycle;			
\fill[red] (\x+0.5,3.5) circle (0.25);
}
\foreach \x in {16}
{\draw[fill=white,white] (\x,3) -- (\x+2,3) -- (\x+2,4) -- (\x,4) -- cycle;
\draw[-] (\x,3) -- (\x,4);
\draw[-] (\x,3) -- (\x+2,3);
\draw[-] (\x,4) -- (\x+2,4);
\draw[-] (\x+1,3) -- (\x+1,4);
\node at (\x+1.5,3.5) {$\cdots$};
}
\foreach \x in {0}
{\draw[fill=white,white] (\x,3) -- (\x-2,3) -- (\x-2,4) -- (\x,4) -- cycle;
\draw[-] (\x,3) -- (\x,4);
\draw[-] (\x,3) -- (\x-2,3);
\draw[-] (\x,4) -- (\x-2,4);
\draw[-] (\x-1,3) -- (\x-1,4);
\node at (\x-1.5,3.5) {$\cdots$};
}
}
\tikz[scale=0.6]{
\foreach \x in {0,1,2,3,4,5,6,7,8,9,10,11,12,13,14,15}
{\draw[fill=white] (\x,3) -- (\x+1,3) -- (\x+1,4) -- (\x,4) -- cycle;			
}
\foreach \x in {7,10,11,13}
{\draw[fill=white] (\x,3) -- (\x+1,3) -- (\x+1,4) -- (\x,4) -- cycle;			
\fill[cyan] (\x+0.5,3.5) circle (0.25);
}
\foreach \x in {4,5,6}
{\draw[fill=white] (\x,3) -- (\x+1,3) -- (\x+1,4) -- (\x,4) -- cycle;			
\fill[red] (\x+0.5,3.5) circle (0.25);
}
\foreach \x in {16}
{\draw[fill=white,white] (\x,3) -- (\x+2,3) -- (\x+2,4) -- (\x,4) -- cycle;
\draw[-] (\x,3) -- (\x,4);
\draw[-] (\x,3) -- (\x+2,3);
\draw[-] (\x,4) -- (\x+2,4);
\draw[-] (\x+1,3) -- (\x+1,4);
\node at (\x+1.5,3.5) {$\cdots$};
}
\foreach \x in {0}
{\draw[fill=white,white] (\x,3) -- (\x-2,3) -- (\x-2,4) -- (\x,4) -- cycle;
\draw[-] (\x,3) -- (\x,4);
\draw[-] (\x,3) -- (\x-2,3);
\draw[-] (\x,4) -- (\x-2,4);
\draw[-] (\x-1,3) -- (\x-1,4);
\node at (\x-1.5,3.5) {$\cdots$};
}
}
\tikz[scale=0.6]{
\foreach \x in {0,1,2,3,4,5,6,7,8,9,10,11,12,13,14,15}
{\draw[fill=white] (\x,3) -- (\x+1,3) -- (\x+1,4) -- (\x,4) -- cycle;			
}
\foreach \x in {10,11,13}
{\draw[fill=white] (\x,3) -- (\x+1,3) -- (\x+1,4) -- (\x,4) -- cycle;			
\fill[cyan] (\x+0.5,3.5) circle (0.25);
}
\foreach \x in {4,5,6,8}
{\draw[fill=white] (\x,3) -- (\x+1,3) -- (\x+1,4) -- (\x,4) -- cycle;			
\fill[red] (\x+0.5,3.5) circle (0.25);
}
\foreach \x in {16}
{\draw[fill=white,white] (\x,3) -- (\x+2,3) -- (\x+2,4) -- (\x,4) -- cycle;
\draw[-] (\x,3) -- (\x,4);
\draw[-] (\x,3) -- (\x+2,3);
\draw[-] (\x,4) -- (\x+2,4);
\draw[-] (\x+1,3) -- (\x+1,4);
\node at (\x+1.5,3.5) {$\cdots$};
}
\foreach \x in {0}
{\draw[fill=white,white] (\x,3) -- (\x-2,3) -- (\x-2,4) -- (\x,4) -- cycle;
\draw[-] (\x,3) -- (\x,4);
\draw[-] (\x,3) -- (\x-2,3);
\draw[-] (\x,4) -- (\x-2,4);
\draw[-] (\x-1,3) -- (\x-1,4);
\node at (\x-1.5,3.5) {$\cdots$};
}
}
\tikz[scale=0.6]{
\foreach \x in {0,1,2,3,4,5,6,7,8,9,10,11,12,13,14,15}
{\draw[fill=white] (\x,3) -- (\x+1,3) -- (\x+1,4) -- (\x,4) -- cycle;			
}
\foreach \x in {11,13}
{\draw[fill=white] (\x,3) -- (\x+1,3) -- (\x+1,4) -- (\x,4) -- cycle;			
\fill[cyan] (\x+0.5,3.5) circle (0.25);
}
\foreach \x in {4,5,6,8,12}
{\draw[fill=white] (\x,3) -- (\x+1,3) -- (\x+1,4) -- (\x,4) -- cycle;			
\fill[red] (\x+0.5,3.5) circle (0.25);
}
\foreach \x in {16}
{\draw[fill=white,white] (\x,3) -- (\x+2,3) -- (\x+2,4) -- (\x,4) -- cycle;
\draw[-] (\x,3) -- (\x,4);
\draw[-] (\x,3) -- (\x+2,3);
\draw[-] (\x,4) -- (\x+2,4);
\draw[-] (\x+1,3) -- (\x+1,4);
\node at (\x+1.5,3.5) {$\cdots$};
}
\foreach \x in {0}
{\draw[fill=white,white] (\x,3) -- (\x-2,3) -- (\x-2,4) -- (\x,4) -- cycle;
\draw[-] (\x,3) -- (\x,4);
\draw[-] (\x,3) -- (\x-2,3);
\draw[-] (\x,4) -- (\x-2,4);
\draw[-] (\x-1,3) -- (\x-1,4);
\node at (\x-1.5,3.5) {$\cdots$};
}
}
\tikz[scale=0.6]{
\foreach \x in {0,1,2,3,4,5,6,7,8,9,10,11,12,13,14,15}
{\draw[fill=white] (\x,3) -- (\x+1,3) -- (\x+1,4) -- (\x,4) -- cycle;			
}
\foreach \x in {13}
{\draw[fill=white] (\x,3) -- (\x+1,3) -- (\x+1,4) -- (\x,4) -- cycle;			
\fill[cyan] (\x+0.5,3.5) circle (0.25);
}
\foreach \x in {4,5,6,8,12,14}
{\draw[fill=white] (\x,3) -- (\x+1,3) -- (\x+1,4) -- (\x,4) -- cycle;			
\fill[red] (\x+0.5,3.5) circle (0.25);
}
\foreach \x in {16}
{\draw[fill=white,white] (\x,3) -- (\x+2,3) -- (\x+2,4) -- (\x,4) -- cycle;
\draw[-] (\x,3) -- (\x,4);
\draw[-] (\x,3) -- (\x+2,3);
\draw[-] (\x,4) -- (\x+2,4);
\draw[-] (\x+1,3) -- (\x+1,4);
\node at (\x+1.5,3.5) {$\cdots$};
}
\foreach \x in {0}
{\draw[fill=white,white] (\x,3) -- (\x-2,3) -- (\x-2,4) -- (\x,4) -- cycle;
\draw[-] (\x,3) -- (\x,4);
\draw[-] (\x,3) -- (\x-2,3);
\draw[-] (\x,4) -- (\x-2,4);
\draw[-] (\x-1,3) -- (\x-1,4);
\node at (\x-1.5,3.5) {$\cdots$};
}
}
\tikz[scale=0.6]{
\foreach \x in {0,1,2,3,4,5,6,7,8,9,10,11,12,13,14,15}
{\draw[fill=white] (\x,3) -- (\x+1,3) -- (\x+1,4) -- (\x,4) -- cycle;			
}
\foreach \x in {}
{\draw[fill=white] (\x,3) -- (\x+1,3) -- (\x+1,4) -- (\x,4) -- cycle;			
\fill[cyan] (\x+0.5,3.5) circle (0.25);
}
\foreach \x in {4,5,6,8,12,14,15}
{\draw[fill=white] (\x,3) -- (\x+1,3) -- (\x+1,4) -- (\x,4) -- cycle;			
\fill[red] (\x+0.5,3.5) circle (0.25);
}
\foreach \x in {16}
{\draw[fill=white,white] (\x,3) -- (\x+2,3) -- (\x+2,4) -- (\x,4) -- cycle;
\draw[-] (\x,3) -- (\x,4);
\draw[-] (\x,3) -- (\x+2,3);
\draw[-] (\x,4) -- (\x+2,4);
\draw[-] (\x+1,3) -- (\x+1,4);
\node at (\x+1.5,3.5) {$\cdots$};
}
\foreach \x in {0}
{\draw[fill=white,white] (\x,3) -- (\x-2,3) -- (\x-2,4) -- (\x,4) -- cycle;
\draw[-] (\x,3) -- (\x,4);
\draw[-] (\x,3) -- (\x-2,3);
\draw[-] (\x,4) -- (\x-2,4);
\draw[-] (\x-1,3) -- (\x-1,4);
\node at (\x-1.5,3.5) {$\cdots$};
}
}

\caption{A box-ball system time evolution (one time step).}\label{firstbbsexample}
\end{figure}	

\subsection{Soliton behaviour and the sorting property}

The box-ball system is sometimes referred to as a soliton cellular automaton. To appreciate this reference, we think of an entire box-ball configuration as being the soliton. In the classical setting, a soliton is thought of as being composed of masses (or pulses) that are nonlinearly related. In the box-ball setting, a ``mass'' corresponds to a consecutive sequence of balls. One may observe (see below) that such a block travels with velocity equal to the number of balls in it, so that larger blocks have velocity greater than smaller blocks. As with classical soliton masses, during the course of its evolution, the blocks may collide, and temporarily change their sizes. However, asymptotically in both forward and backward (discrete) time ($t$), the sizes of blocks comprising the soliton are the same. We will therefore refer to such a configuration as an \textit{$n$-soliton} if the total number of blocks asymptotically is $n$.

After blocks collide, they come out of the collision ordered with the longer blocks ahead of smaller blocks but having a \textit{phase shift} due to the nonlinearity. By phase shift here, we mean the difference between where the block ends up after the collision and where the block would have been if it were not for the collision.

In the following figure, we illustrate how the blocks become ordered after sufficiently many time evolutions. Once sorted, they travel with their respective velocities, never to collide again.

\begin{figure}[h!]
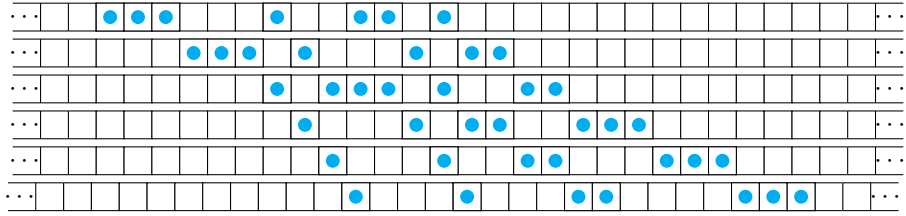
\centering
\tikz[scale=0.37]{
\foreach \x in {0,1,2,3,4,5,6,7,8,9,10,11,12,13,14,15,16,17,18,19,20,21,22,23,24,25,26,27}
{\draw[fill=white] (\x,3) -- (\x+1,3) -- (\x+1,4) -- (\x,4) -- cycle;			
}
\foreach \x in {1,2,3,7,10,11,13}
{\draw[fill=white] (\x,3) -- (\x+1,3) -- (\x+1,4) -- (\x,4) -- cycle;			
\fill[cyan] (\x+0.5,3.5) circle (0.25);
}
\foreach \x in {}
{\draw[fill=white] (\x,3) -- (\x+1,3) -- (\x+1,4) -- (\x,4) -- cycle;			
\fill[red] (\x+0.5,3.5) circle (0.25);
}
\foreach \x in {28}
{\draw[fill=white,white] (\x,3) -- (\x+2,3) -- (\x+2,4) -- (\x,4) -- cycle;
\draw[-] (\x,3) -- (\x,4);
\draw[-] (\x,3) -- (\x+2,3);
\draw[-] (\x,4) -- (\x+2,4);
\draw[-] (\x+1,3) -- (\x+1,4);
\node at (\x+1.6,3.5) {\small$\cdots$};
}
\foreach \x in {0}
{\draw[fill=white,white] (\x,3) -- (\x-2,3) -- (\x-2,4) -- (\x,4) -- cycle;
\draw[-] (\x,3) -- (\x,4);
\draw[-] (\x,3) -- (\x-2,3);
\draw[-] (\x,4) -- (\x-2,4);
\draw[-] (\x-1,3) -- (\x-1,4);
\node at (\x-1.5,3.5) {\small$\cdots$};
}
}
\tikz[scale=0.37]{
\foreach \x in {0,1,2,3,4,5,6,7,8,9,10,11,12,13,14,15,16,17,18,19,20,21,22,23,24,25,26,27}
{\draw[fill=white] (\x,3) -- (\x+1,3) -- (\x+1,4) -- (\x,4) -- cycle;			
}
\foreach \x in {4,5,6,8,12,14,15}
{\draw[fill=white] (\x,3) -- (\x+1,3) -- (\x+1,4) -- (\x,4) -- cycle;			
\fill[cyan] (\x+0.5,3.5) circle (0.25);
}
\foreach \x in {}
{\draw[fill=white] (\x,3) -- (\x+1,3) -- (\x+1,4) -- (\x,4) -- cycle;			
\fill[red] (\x+0.5,3.5) circle (0.25);
}
\foreach \x in {28}
{\draw[fill=white,white] (\x,3) -- (\x+2,3) -- (\x+2,4) -- (\x,4) -- cycle;
\draw[-] (\x,3) -- (\x,4);
\draw[-] (\x,3) -- (\x+2,3);
\draw[-] (\x,4) -- (\x+2,4);
\draw[-] (\x+1,3) -- (\x+1,4);
\node at (\x+1.6,3.5) {\small$\cdots$};
}
\foreach \x in {0}
{\draw[fill=white,white] (\x,3) -- (\x-2,3) -- (\x-2,4) -- (\x,4) -- cycle;
\draw[-] (\x,3) -- (\x,4);
\draw[-] (\x,3) -- (\x-2,3);
\draw[-] (\x,4) -- (\x-2,4);
\draw[-] (\x-1,3) -- (\x-1,4);
\node at (\x-1.5,3.5) {\small$\cdots$};
}
}
\tikz[scale=0.37]{
\foreach \x in {0,1,2,3,4,5,6,7,8,9,10,11,12,13,14,15,16,17,18,19,20,21,22,23,24,25,26,27}
{\draw[fill=white] (\x,3) -- (\x+1,3) -- (\x+1,4) -- (\x,4) -- cycle;			
}
\foreach \x in {7,9,10,11,13,16,17}
{\draw[fill=white] (\x,3) -- (\x+1,3) -- (\x+1,4) -- (\x,4) -- cycle;			
\fill[cyan] (\x+0.5,3.5) circle (0.25);
}
\foreach \x in {}
{\draw[fill=white] (\x,3) -- (\x+1,3) -- (\x+1,4) -- (\x,4) -- cycle;			
\fill[red] (\x+0.5,3.5) circle (0.25);
}
\foreach \x in {28}
{\draw[fill=white,white] (\x,3) -- (\x+2,3) -- (\x+2,4) -- (\x,4) -- cycle;
\draw[-] (\x,3) -- (\x,4);
\draw[-] (\x,3) -- (\x+2,3);
\draw[-] (\x,4) -- (\x+2,4);
\draw[-] (\x+1,3) -- (\x+1,4);
\node at (\x+1.6,3.5) {\small$\cdots$};
}
\foreach \x in {0}
{\draw[fill=white,white] (\x,3) -- (\x-2,3) -- (\x-2,4) -- (\x,4) -- cycle;
\draw[-] (\x,3) -- (\x,4);
\draw[-] (\x,3) -- (\x-2,3);
\draw[-] (\x,4) -- (\x-2,4);
\draw[-] (\x-1,3) -- (\x-1,4);
\node at (\x-1.5,3.5) {\small$\cdots$};
}
}
\tikz[scale=0.37]{
\foreach \x in {0,1,2,3,4,5,6,7,8,9,10,11,12,13,14,15,16,17,18,19,20,21,22,23,24,25,26,27}
{\draw[fill=white] (\x,3) -- (\x+1,3) -- (\x+1,4) -- (\x,4) -- cycle;			
}
\foreach \x in {8,12,14,15,18,19,20}
{\draw[fill=white] (\x,3) -- (\x+1,3) -- (\x+1,4) -- (\x,4) -- cycle;			
\fill[cyan] (\x+0.5,3.5) circle (0.25);
}
\foreach \x in {}
{\draw[fill=white] (\x,3) -- (\x+1,3) -- (\x+1,4) -- (\x,4) -- cycle;			
\fill[red] (\x+0.5,3.5) circle (0.25);
}
\foreach \x in {28}
{\draw[fill=white,white] (\x,3) -- (\x+2,3) -- (\x+2,4) -- (\x,4) -- cycle;
\draw[-] (\x,3) -- (\x,4);
\draw[-] (\x,3) -- (\x+2,3);
\draw[-] (\x,4) -- (\x+2,4);
\draw[-] (\x+1,3) -- (\x+1,4);
\node at (\x+1.6,3.5) {\small$\cdots$};
}
\foreach \x in {0}
{\draw[fill=white,white] (\x,3) -- (\x-2,3) -- (\x-2,4) -- (\x,4) -- cycle;
\draw[-] (\x,3) -- (\x,4);
\draw[-] (\x,3) -- (\x-2,3);
\draw[-] (\x,4) -- (\x-2,4);
\draw[-] (\x-1,3) -- (\x-1,4);
\node at (\x-1.5,3.5) {\small$\cdots$};
}
}
\tikz[scale=0.37]{
\foreach \x in {0,1,2,3,4,5,6,7,8,9,10,11,12,13,14,15,16,17,18,19,20,21,22,23,24,25,26,27}
{\draw[fill=white] (\x,3) -- (\x+1,3) -- (\x+1,4) -- (\x,4) -- cycle;			
}
\foreach \x in {9,13,16,17,21,22,23}
{\draw[fill=white] (\x,3) -- (\x+1,3) -- (\x+1,4) -- (\x,4) -- cycle;			
\fill[cyan] (\x+0.5,3.5) circle (0.25);
}
\foreach \x in {}
{\draw[fill=white] (\x,3) -- (\x+1,3) -- (\x+1,4) -- (\x,4) -- cycle;			
\fill[red] (\x+0.5,3.5) circle (0.25);
}
\foreach \x in {28}
{\draw[fill=white,white] (\x,3) -- (\x+2,3) -- (\x+2,4) -- (\x,4) -- cycle;
\draw[-] (\x,3) -- (\x,4);
\draw[-] (\x,3) -- (\x+2,3);
\draw[-] (\x,4) -- (\x+2,4);
\draw[-] (\x+1,3) -- (\x+1,4);
\node at (\x+1.6,3.5) {\small$\cdots$};
}
\foreach \x in {0}
{\draw[fill=white,white] (\x,3) -- (\x-2,3) -- (\x-2,4) -- (\x,4) -- cycle;
\draw[-] (\x,3) -- (\x,4);
\draw[-] (\x,3) -- (\x-2,3);
\draw[-] (\x,4) -- (\x-2,4);
\draw[-] (\x-1,3) -- (\x-1,4);
\node at (\x-1.5,3.5) {\small$\cdots$};
}
}
\tikz[scale=0.37]{
\foreach \x in {0,1,2,3,4,5,6,7,8,9,10,11,12,13,14,15,16,17,18,19,20,21,22,23,24,25,26,27}
{\draw[fill=white] (\x,3) -- (\x+1,3) -- (\x+1,4) -- (\x,4) -- cycle;			
}
\foreach \x in {10,14,18,19,24,25,26}
{\draw[fill=white] (\x,3) -- (\x+1,3) -- (\x+1,4) -- (\x,4) -- cycle;			
\fill[cyan] (\x+0.5,3.5) circle (0.25);
}
\foreach \x in {}
{\draw[fill=white] (\x,3) -- (\x+1,3) -- (\x+1,4) -- (\x,4) -- cycle;			
\fill[red] (\x+0.5,3.5) circle (0.25);
}
\foreach \x in {28}
{\draw[fill=white,white] (\x,3) -- (\x+2,3) -- (\x+2,4) -- (\x,4) -- cycle;
\draw[-] (\x,3) -- (\x,4);
\draw[-] (\x,3) -- (\x+2,3);
\draw[-] (\x,4) -- (\x+2,4);
\draw[-] (\x+1,3) -- (\x+1,4);
\node at (\x+1.6,3.5) {\small$\cdots$};
}
\foreach \x in {0}
{\draw[fill=white,white] (\x,3) -- (\x-2,3) -- (\x-2,4) -- (\x,4) -- cycle;
\draw[-] (\x,3) -- (\x,4);
\draw[-] (\x,3) -- (\x-2,3);
\draw[-] (\x,4) -- (\x-2,4);
\draw[-] (\x-1,3) -- (\x-1,4);
\node at (\x-1.5,3.5) {\small$\cdots$};
}
}
\caption{The sorting property of the box-ball system.}\label{sortingpropbbs}
\end{figure}

\subsubsection{Coordinates on the box-ball system}
Suppose at time $t$, one has $n$ blocks in the soliton. Let $Q_1^t, Q_2^t, \dots, Q_n^t$ denote the lengths of these blocks, taken from left to right. Let $W_1^t, W_2^t, \dots, W_{n-1}^t$ denote the lengths of the sets of empty boxes between the $n$ sets of filled boxes, again taken from left to right. Lastly, let $W_0^t$ and $W_n^t$ be formally defined to be $\infty$, reflecting the fact that the empty boxes continue infinitely in both directions.

The following theorem gives evolution equations for these coordinates. They can be found, for example, in~\cite{bib:tokihiro}.

\begin{thm}[\cite{bib:tokihiro}]\label{thmbbscoords}
The coordinates $\big(W_0^t,Q_1^t,W_1^t,\dots,Q_n^t,W_n^t\big)$ evolve under the box ball dynamics according to
\begin{gather}
W_0^{t+1}=W_n^{t+1}=\infty,\qquad
W_i^{t+1}=Q_{i+1}^t+W_i^t-Q_i^{t+1},\qquad i=1,\dots,n-1,\label{Witplusoneeqnbbs}\\
Q_i^{t+1}=\min\Bigg(W_i^{t},\sum_{j=1}^i Q_j^t-\sum_{j=1}^{i-1}
Q_j^{t+1}\Bigg),\qquad i=1,\dots,n.\nonumber
\end{gather}
\end{thm}

\begin{rem}It follows from Remark 2.6 in \cite{bib:er} that each $W_i^t>0$ for each $i$ and for all time. Furthermore, since there are always $n$ blocks, each $Q_i^t>0$ (by definition of a block).
\end{rem}

\begin{ex}Take the initial state in Figure \ref{firstbbsexample}:
\begin{figure}[h!]
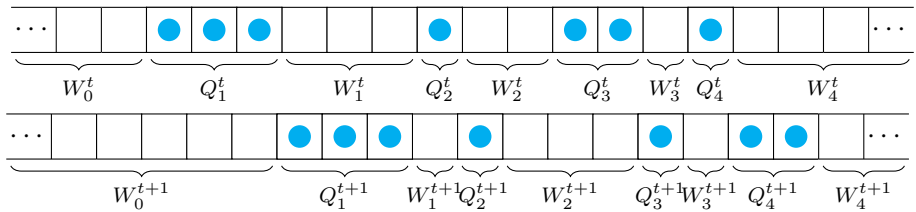
\centering
\tikz[scale=0.6]{
\foreach \x in {0,1,2,3,4,5,6,7,8,9,10,11,12,13,14,15}
{\draw[fill=white] (\x,3) -- (\x+1,3) -- (\x+1,4) -- (\x,4) -- cycle;			
}
\foreach \x in {1,2,3,7,10,11,13}
{\draw[fill=white] (\x,3) -- (\x+1,3) -- (\x+1,4) -- (\x,4) -- cycle;			
\fill[cyan] (\x+0.5,3.5) circle (0.25);
}
\foreach \x in {}
{\draw[fill=white] (\x,3) -- (\x+1,3) -- (\x+1,4) -- (\x,4) -- cycle;			
\fill[red] (\x+0.5,3.5) circle (0.25);
}
\foreach \x in {16}
{\draw[fill=white,white] (\x,3) -- (\x+2,3) -- (\x+2,4) -- (\x,4) -- cycle;
\draw[-] (\x,3) -- (\x,4);
\draw[-] (\x,3) -- (\x+2,3);
\draw[-] (\x,4) -- (\x+2,4);
\draw[-] (\x+1,3) -- (\x+1,4);
\node at (\x+1.5,3.5) {$\cdots$};
}
\foreach \x in {0}
{\draw[fill=white,white] (\x,3) -- (\x-2,3) -- (\x-2,4) -- (\x,4) -- cycle;
\draw[-] (\x,3) -- (\x,4);
\draw[-] (\x,3) -- (\x-2,3);
\draw[-] (\x,4) -- (\x-2,4);
\draw[-] (\x-1,3) -- (\x-1,4);
\node at (\x-1.5,3.5) {$\cdots$};
}
\draw [decorate,decoration={brace,amplitude=4pt}] (0.9,2.85) -- (-1.9,2.85) node [black,midway,yshift=-0.4cm] {\scriptsize{$W_0^t$}};
\draw [decorate,decoration={brace,amplitude=4pt}] (3.9,2.85) -- (1.1,2.85) node [black,midway,yshift=-0.4cm] {\scriptsize{$Q_1^t$}};
\draw [decorate,decoration={brace,amplitude=4pt}] (6.9,2.85) -- (4.1,2.85) node [black,midway,yshift=-0.4cm] {\scriptsize{$W_1^t$}};
\draw [decorate,decoration={brace,amplitude=4pt}] (7.9,2.85) -- (7.1,2.85) node [black,midway,yshift=-0.4cm] {\scriptsize{$Q_2^t$}};
\draw [decorate,decoration={brace,amplitude=4pt}] (9.9,2.85) -- (8.1,2.85) node [black,midway,yshift=-0.4cm] {\scriptsize{$W_2^t$}};
\draw [decorate,decoration={brace,amplitude=4pt}] (11.9,2.85) -- (10.1,2.85) node [black,midway,yshift=-0.4cm] {\scriptsize{$Q_3^t$}};
\draw [decorate,decoration={brace,amplitude=4pt}] (12.9,2.85) -- (12.1,2.85) node [black,midway,yshift=-0.4cm] {\scriptsize{$W_3^t$}};
\draw [decorate,decoration={brace,amplitude=4pt}] (13.9,2.85) -- (13.1,2.85) node [black,midway,yshift=-0.4cm] {\scriptsize{$Q_4^t$}};
\draw [decorate,decoration={brace,amplitude=4pt}] (17.9,2.85) -- (14.1,2.85) node [black,midway,yshift=-0.4cm] {\scriptsize{$W_4^t$}};
}
\tikz[scale=0.6]{
\foreach \x in {0,1,2,3,4,5,6,7,8,9,10,11,12,13,14,15}
{\draw[fill=white] (\x,3) -- (\x+1,3) -- (\x+1,4) -- (\x,4) -- cycle;			
}
\foreach \x in {4,5,6,8,12,14,15}
{\draw[fill=white] (\x,3) -- (\x+1,3) -- (\x+1,4) -- (\x,4) -- cycle;			
\fill[cyan] (\x+0.5,3.5) circle (0.25);
}
\foreach \x in {}
{\draw[fill=white] (\x,3) -- (\x+1,3) -- (\x+1,4) -- (\x,4) -- cycle;			
\fill[red] (\x+0.5,3.5) circle (0.25);
}
\foreach \x in {16}
{\draw[fill=white,white] (\x,3) -- (\x+2,3) -- (\x+2,4) -- (\x,4) -- cycle;
\draw[-] (\x,3) -- (\x,4);
\draw[-] (\x,3) -- (\x+2,3);
\draw[-] (\x,4) -- (\x+2,4);
\draw[-] (\x+1,3) -- (\x+1,4);
\node at (\x+1.5,3.5) {$\cdots$};
}
\foreach \x in {0}
{\draw[fill=white,white] (\x,3) -- (\x-2,3) -- (\x-2,4) -- (\x,4) -- cycle;
\draw[-] (\x,3) -- (\x,4);
\draw[-] (\x,3) -- (\x-2,3);
\draw[-] (\x,4) -- (\x-2,4);
\draw[-] (\x-1,3) -- (\x-1,4);
\node at (\x-1.5,3.5) {$\cdots$};
}
\draw [decorate,decoration={brace,amplitude=4pt}] (3.9,2.85) -- (-1.9,2.85) node [black,midway,yshift=-0.4cm] {\scriptsize{$W_0^{t+1}$}};
\draw [decorate,decoration={brace,amplitude=4pt}] (6.9,2.85) -- (4.1,2.85) node [black,midway,yshift=-0.4cm] {\scriptsize{$Q_1^{t+1}$}};
\draw [decorate,decoration={brace,amplitude=4pt}] (7.9,2.85) -- (7.1,2.85) node [black,midway,yshift=-0.4cm] {\scriptsize{$W_1^{t+1}$}};
\draw [decorate,decoration={brace,amplitude=4pt}] (8.9,2.85) -- (8.1,2.85) node [black,midway,yshift=-0.4cm] {\scriptsize{$\,\,Q_2^{t+1}$}};
\draw [decorate,decoration={brace,amplitude=4pt}] (11.9,2.85) -- (9.1,2.85) node [black,midway,yshift=-0.4cm] {\scriptsize{$W_2^{t+1}$}};
\draw [decorate,decoration={brace,amplitude=4pt}] (12.9,2.85) -- (12.1,2.85) node [black,midway,yshift=-0.4cm] {\scriptsize{$Q_3^{t+1}$}};
\draw [decorate,decoration={brace,amplitude=4pt}] (13.9,2.85) -- (13.1,2.85) node [black,midway,yshift=-0.4cm] {\scriptsize{$\,\,W_3^{t+1}$}};
\draw [decorate,decoration={brace,amplitude=4pt}] (15.9,2.85) -- (14.1,2.85) node [black,midway,yshift=-0.4cm] {\scriptsize{$Q_4^{t+1}$}};
\draw [decorate,decoration={brace,amplitude=4pt}] (17.9,2.85) -- (16.1,2.85) node [black,midway,yshift=-0.4cm] {\scriptsize{$W_4^{t+1}$}};
}
\caption{The box-ball coordinates on a box-ball system and its time evolution.}
\end{figure}

Under the time evolution, the coordinates evolve as
\begin{equation*}
(\infty,3,3,1,2,2,1,1,\infty)\mapsto (\infty,3,1,1,3,1,1,2,\infty).
\end{equation*}
\end{ex}

\subsection{Phase shift phenomenon for the BBS}
We have seen how, as $t\to +\infty$, the blocks sort themselves by increasing length. The same holds in reverse time: as $t\to-\infty$, the blocks are ordered by decreasing lengths. The following example demonstrates why one cannot simply just count block lengths (the third state does not show the soliton structure of blocks of length~1 and~3):

\begin{figure}[h!]
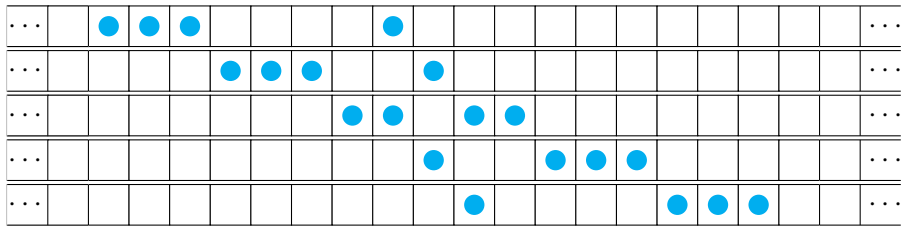

\centering
\tikz[scale=0.54]{
\foreach \y in {3}{
\foreach \x in {-1,0,1,2,3,4,5,6,7,8,9,10,11,12,13,14,15,16,17,18}
{\draw[fill=white] (\x,\y) -- (\x+1,\y) -- (\x+1,\y+1) -- (\x,\y+1) -- cycle;			
}
\foreach \x in {1,2,3,8}
{
\fill[cyan] (\x+0.5,\y+0.5) circle (0.25);
}
\foreach \x in {19}
{\draw[fill=white,white] (\x,\y) -- (\x+2,\y) -- (\x+2,\y+1) -- (\x,\y+1) -- cycle;
\draw[-] (\x,\y) -- (\x,\y+1);
\draw[-] (\x,\y) -- (\x+2,\y);
\draw[-] (\x,\y+1) -- (\x+2,\y+1);
\draw[-] (\x+1,\y) -- (\x+1,\y+1);
\node at (\x+1.5,\y+0.5) {$~\cdots$};
}
\foreach \x in {1}
{\draw[fill=white,white] (\x,\y) -- (\x-2,\y) -- (\x-2,\y+1) -- (\x,\y+1) -- cycle;
\draw[-] (\x,\y) -- (\x,\y+1);
\draw[-] (\x,\y) -- (\x-2,\y);
\draw[-] (\x,\y+1) -- (\x-2,\y+1);
\draw[-] (\x-1,\y+1) -- (\x-1,\y+1);
\draw[-] (\x-1,\y+1) -- (\x-1,\y);
\node at (\x-1.5,\y+0.5) {$\cdots$};
}}
\foreach \y in {1.9}{
\foreach \x in {-1,0,1,2,3,4,5,6,7,8,9,10,11,12,13,14,15,16,17,18}
{\draw[fill=white] (\x,\y) -- (\x+1,\y) -- (\x+1,\y+1) -- (\x,\y+1) -- cycle;			
}
\foreach \x in {4,5,6,9}
{
\fill[cyan] (\x+0.5,\y+0.5) circle (0.25);
}
\foreach \x in {19}
{\draw[fill=white,white] (\x,\y) -- (\x+2,\y) -- (\x+2,\y+1) -- (\x,\y+1) -- cycle;
\draw[-] (\x,\y) -- (\x,\y+1);
\draw[-] (\x,\y) -- (\x+2,\y);
\draw[-] (\x,\y+1) -- (\x+2,\y+1);
\draw[-] (\x+1,\y) -- (\x+1,\y+1);
\node at (\x+1.5,\y+0.5) {$~\cdots$};
}
\foreach \x in {1}
{\draw[fill=white,white] (\x,\y) -- (\x-2,\y) -- (\x-2,\y+1) -- (\x,\y+1) -- cycle;
\draw[-] (\x,\y) -- (\x,\y+1);
\draw[-] (\x,\y) -- (\x-2,\y);
\draw[-] (\x,\y+1) -- (\x-2,\y+1);
\draw[-] (\x-1,\y+1) -- (\x-1,\y+1);
\draw[-] (\x-1,\y+1) -- (\x-1,\y);
\node at (\x-1.5,\y+0.5) {$\cdots$};
}}
\foreach \y in {0.8}{
\foreach \x in {-1,0,1,2,3,4,5,6,7,8,9,10,11,12,13,14,15,16,17,18}
{\draw[fill=white] (\x,\y) -- (\x+1,\y) -- (\x+1,\y+1) -- (\x,\y+1) -- cycle;			
}
\foreach \x in {7,8,10,11}
{
\fill[cyan] (\x+0.5,\y+0.5) circle (0.25);
}
\foreach \x in {19}
{\draw[fill=white,white] (\x,\y) -- (\x+2,\y) -- (\x+2,\y+1) -- (\x,\y+1) -- cycle;
\draw[-] (\x,\y) -- (\x,\y+1);
\draw[-] (\x,\y) -- (\x+2,\y);
\draw[-] (\x,\y+1) -- (\x+2,\y+1);
\draw[-] (\x+1,\y) -- (\x+1,\y+1);
\node at (\x+1.5,\y+0.5) {$~\cdots$};
}
\foreach \x in {1}
{\draw[fill=white,white] (\x,\y) -- (\x-2,\y) -- (\x-2,\y+1) -- (\x,\y+1) -- cycle;
\draw[-] (\x,\y) -- (\x,\y+1);
\draw[-] (\x,\y) -- (\x-2,\y);
\draw[-] (\x,\y+1) -- (\x-2,\y+1);
\draw[-] (\x-1,\y+1) -- (\x-1,\y+1);
\draw[-] (\x-1,\y+1) -- (\x-1,\y);
\node at (\x-1.5,\y+0.5) {$\cdots$};
}}
\foreach \y in {-0.3}{
\foreach \x in {-1,0,1,2,3,4,5,6,7,8,9,10,11,12,13,14,15,16,17,18}
{\draw[fill=white] (\x,\y) -- (\x+1,\y) -- (\x+1,\y+1) -- (\x,\y+1) -- cycle;			
}
\foreach \x in {9,12,13,14}
{
\fill[cyan] (\x+0.5,\y+0.5) circle (0.25);
}
\foreach \x in {19}
{\draw[fill=white,white] (\x,\y) -- (\x+2,\y) -- (\x+2,\y+1) -- (\x,\y+1) -- cycle;
\draw[-] (\x,\y) -- (\x,\y+1);
\draw[-] (\x,\y) -- (\x+2,\y);
\draw[-] (\x,\y+1) -- (\x+2,\y+1);
\draw[-] (\x+1,\y) -- (\x+1,\y+1);
\node at (\x+1.5,\y+0.5) {$~\cdots$};
}
\foreach \x in {1}
{\draw[fill=white,white] (\x,\y) -- (\x-2,\y) -- (\x-2,\y+1) -- (\x,\y+1) -- cycle;
\draw[-] (\x,\y) -- (\x,\y+1);
\draw[-] (\x,\y) -- (\x-2,\y);
\draw[-] (\x,\y+1) -- (\x-2,\y+1);
\draw[-] (\x-1,\y+1) -- (\x-1,\y+1);
\draw[-] (\x-1,\y+1) -- (\x-1,\y);
\node at (\x-1.5,\y+0.5) {$\cdots$};
}}
\foreach \y in {-1.4}{
\foreach \x in {-1,0,1,2,3,4,5,6,7,8,9,10,11,12,13,14,15,16,17,18}
{\draw[fill=white] (\x,\y) -- (\x+1,\y) -- (\x+1,\y+1) -- (\x,\y+1) -- cycle;			
}
\foreach \x in {10,15,16,17}
{
\fill[cyan] (\x+0.5,\y+0.5) circle (0.25);
}
\foreach \x in {19}
{\draw[fill=white,white] (\x,\y) -- (\x+2,\y) -- (\x+2,\y+1) -- (\x,\y+1) -- cycle;
\draw[-] (\x,\y) -- (\x,\y+1);
\draw[-] (\x,\y) -- (\x+2,\y);
\draw[-] (\x,\y+1) -- (\x+2,\y+1);
\draw[-] (\x+1,\y) -- (\x+1,\y+1);
\node at (\x+1.5,\y+0.5) {$~\cdots$};
}
\foreach \x in {1}
{\draw[fill=white,white] (\x,\y) -- (\x-2,\y) -- (\x-2,\y+1) -- (\x,\y+1) -- cycle;
\draw[-] (\x,\y) -- (\x,\y+1);
\draw[-] (\x,\y) -- (\x-2,\y);
\draw[-] (\x,\y+1) -- (\x-2,\y+1);
\draw[-] (\x-1,\y+1) -- (\x-1,\y+1);
\draw[-] (\x-1,\y+1) -- (\x-1,\y);
\node at (\x-1.5,\y+0.5) {$\cdots$};
}}
}
\caption{A phase shift interaction between two colliding blocks.}\label{figcolsdhufhdfa}
\end{figure}	

\begin{rem}\label{remarkaboutspacingforsoliton}
In the above example, we can discern the asymptotic ordering in the first, second, fourth and fifth rows, simply by counting the numbers of balls in each block of adjacent balls. The middle row (the third) could be misleading, since it reveals a $(2,2)$ structure for the blocks. If two blocks are spaced far enough apart, then no such obfuscation occurs.
\end{rem}

Barring this intricacy (i.e., when there is enough space between consecutive blocks), one can take two blocks, evolve sufficiently many times according to the box-ball evolution, and compare the position of the blocks to where they would have been if it had not have been for the collision.

In the figure below, we replicate Figure~\ref{figcolsdhufhdfa}. However, we use green balls to keep track of where the block of three balls would have been without the collision, and magenta balls to keep track of where the block of one ball would have been.

\begin{figure}[h!]
\centering
\tikz[scale=0.51]{
\foreach \x in {0,1,2,3,4,5,6,7,8,9,10,11,12,13,14,15,16,17,18}
{\draw[fill=white] (\x,3) -- (\x+1,3) -- (\x+1,4) -- (\x,4) -- cycle;			
}
\foreach \x in {1,2,3}
{			
\fill[green] (\x+0.5,3.5) circle (0.35);
}
\foreach \x in {8}
{			
\fill[magenta] (\x+0.5,3.5) circle (0.25);
}
\foreach \x in {1,2,3,8}
{			
\fill[cyan] (\x+0.5,3.5) circle (0.15);
}
\foreach \x in {19}
{\draw[fill=white,white] (\x,3) -- (\x+2,3) -- (\x+2,4) -- (\x,4) -- cycle;
\draw[-] (\x,3) -- (\x,4);
\draw[-] (\x,3) -- (\x+2,3);
\draw[-] (\x,4) -- (\x+2,4);
\draw[-] (\x+1,3) -- (\x+1,4);
\node at (\x+1.5,3.5) {$\cdots$};
}
\foreach \x in {0}
{\draw[fill=white,white] (\x,3) -- (\x-2,3) -- (\x-2,4) -- (\x,4) -- cycle;
\draw[-] (\x,3) -- (\x,4);
\draw[-] (\x,3) -- (\x-2,3);
\draw[-] (\x,4) -- (\x-2,4);
\draw[-] (\x-1,3) -- (\x-1,4);
\node at (\x-1.5,3.5) {$\cdots$};
}
}
\tikz[scale=0.51]{
\foreach \x in {0,1,2,3,4,5,6,7,8,9,10,11,12,13,14,15,16,17,18}
{\draw[fill=white] (\x,3) -- (\x+1,3) -- (\x+1,4) -- (\x,4) -- cycle;			
}
\foreach \x in {4,5,6}
{			
\fill[green] (\x+0.5,3.5) circle (0.35);
}
\foreach \x in {9}
{			
\fill[magenta] (\x+0.5,3.5) circle (0.25);
}
\foreach \x in {4,5,6,9}
{			
\fill[cyan] (\x+0.5,3.5) circle (0.15);
}
\foreach \x in {19}
{\draw[fill=white,white] (\x,3) -- (\x+2,3) -- (\x+2,4) -- (\x,4) -- cycle;
\draw[-] (\x,3) -- (\x,4);
\draw[-] (\x,3) -- (\x+2,3);
\draw[-] (\x,4) -- (\x+2,4);
\draw[-] (\x+1,3) -- (\x+1,4);
\node at (\x+1.5,3.5) {$\cdots$};
}
\foreach \x in {0}
{\draw[fill=white,white] (\x,3) -- (\x-2,3) -- (\x-2,4) -- (\x,4) -- cycle;
\draw[-] (\x,3) -- (\x,4);
\draw[-] (\x,3) -- (\x-2,3);
\draw[-] (\x,4) -- (\x-2,4);
\draw[-] (\x-1,3) -- (\x-1,4);
\node at (\x-1.5,3.5) {$\cdots$};
}
}
\tikz[scale=0.51]{
\foreach \x in {0,1,2,3,4,5,6,7,8,9,10,11,12,13,14,15,16,17,18}
{\draw[fill=white] (\x,3) -- (\x+1,3) -- (\x+1,4) -- (\x,4) -- cycle;			
}
\foreach \x in {7,8,9}
{			
\fill[green] (\x+0.5,3.5) circle (0.35);
}
\foreach \x in {10}
{			
\fill[magenta] (\x+0.5,3.5) circle (0.25);
}
\foreach \x in {7,8,10,11}
{			
\fill[cyan] (\x+0.5,3.5) circle (0.15);
}
\foreach \x in {19}
{\draw[fill=white,white] (\x,3) -- (\x+2,3) -- (\x+2,4) -- (\x,4) -- cycle;
\draw[-] (\x,3) -- (\x,4);
\draw[-] (\x,3) -- (\x+2,3);
\draw[-] (\x,4) -- (\x+2,4);
\draw[-] (\x+1,3) -- (\x+1,4);
\node at (\x+1.5,3.5) {$\cdots$};
}
\foreach \x in {0}
{\draw[fill=white,white] (\x,3) -- (\x-2,3) -- (\x-2,4) -- (\x,4) -- cycle;
\draw[-] (\x,3) -- (\x,4);
\draw[-] (\x,3) -- (\x-2,3);
\draw[-] (\x,4) -- (\x-2,4);
\draw[-] (\x-1,3) -- (\x-1,4);
\node at (\x-1.5,3.5) {$\cdots$};
}
}
\tikz[scale=0.51]{
\foreach \x in {0,1,2,3,4,5,6,7,8,9,10,11,12,13,14,15,16,17,18}
{\draw[fill=white] (\x,3) -- (\x+1,3) -- (\x+1,4) -- (\x,4) -- cycle;			
}
\foreach \x in {10,11,12}
{			
\fill[green] (\x+0.5,3.5) circle (0.35);
}
\foreach \x in {11}
{			
\fill[magenta] (\x+0.5,3.5) circle (0.25);
}
\foreach \x in {9,12,13,14}
{			
\fill[cyan] (\x+0.5,3.5) circle (0.15);
}
\foreach \x in {19}
{\draw[fill=white,white] (\x,3) -- (\x+2,3) -- (\x+2,4) -- (\x,4) -- cycle;
\draw[-] (\x,3) -- (\x,4);
\draw[-] (\x,3) -- (\x+2,3);
\draw[-] (\x,4) -- (\x+2,4);
\draw[-] (\x+1,3) -- (\x+1,4);
\node at (\x+1.5,3.5) {$\cdots$};
}
\foreach \x in {0}
{\draw[fill=white,white] (\x,3) -- (\x-2,3) -- (\x-2,4) -- (\x,4) -- cycle;
\draw[-] (\x,3) -- (\x,4);
\draw[-] (\x,3) -- (\x-2,3);
\draw[-] (\x,4) -- (\x-2,4);
\draw[-] (\x-1,3) -- (\x-1,4);
\node at (\x-1.5,3.5) {$\cdots$};
}
}
\tikz[scale=0.51]{
\foreach \x in {0,1,2,3,4,5,6,7,8,9,10,11,12,13,14,15,16,17,18}
{\draw[fill=white] (\x,3) -- (\x+1,3) -- (\x+1,4) -- (\x,4) -- cycle;			
}
\foreach \x in {13,14,15}
{			
\fill[green] (\x+0.5,3.5) circle (0.35);
}
\foreach \x in {12}
{			
\fill[magenta] (\x+0.5,3.5) circle (0.25);
}
\foreach \x in {10,15,16,17}
{			
\fill[cyan] (\x+0.5,3.5) circle (0.15);
}
\foreach \x in {19}
{\draw[fill=white,white] (\x,3) -- (\x+2,3) -- (\x+2,4) -- (\x,4) -- cycle;
\draw[-] (\x,3) -- (\x,4);
\draw[-] (\x,3) -- (\x+2,3);
\draw[-] (\x,4) -- (\x+2,4);
\draw[-] (\x+1,3) -- (\x+1,4);
\node at (\x+1.5,3.5) {$\cdots$};
}
\foreach \x in {0}
{\draw[fill=white,white] (\x,3) -- (\x-2,3) -- (\x-2,4) -- (\x,4) -- cycle;
\draw[-] (\x,3) -- (\x,4);
\draw[-] (\x,3) -- (\x-2,3);
\draw[-] (\x,4) -- (\x-2,4);
\draw[-] (\x-1,3) -- (\x-1,4);
\node at (\x-1.5,3.5) {$\cdots$};
}
}
\end{figure}	

In the above, we look at the position of the block of three blue balls and notice that it is shifted two spaces to the right of where it would have been if it were not for the collision. Similarly, we look at the isolated blue ball and see that it is shifted left two positions of where it would have been without the collision. This shifting of blocks from their \textit{would-be} positions is the phenomenon of phase shifting described above.

\section{The Toda lattice} \label{TODA}

The Toda lattice \cite{bib:toda} is a dynamical system on $\mathbb{R}^{2n}$, with coordinates $(p_1,\dots,p_n,q_1,\dots,q_n)$. The system is Hamiltonian with respect to the standard symplectic structure on~$\mathbb{R}^{2n}$ with Hamiltonian
\begin{equation*}
H(p_1,\dots,p_n,q_1,\dots,q_n)=\dfrac{1}{2}\sum_{j=1}^n p_j^2 + \sum_{j=1}^{n-1}{\rm e}^{q_j-q_{j+1}}.
\end{equation*}

The Toda lattice equations are then given by
\begin{gather}
\dot{q}_j =p_j,\qquad j=1,\dots,n,\label{eqoneofhamiltod}\\
\dot{p}_j = \begin{cases}
-{\rm e}^{q_1-q_2} & \text{if }j=1,\\
{\rm e}^{q_{j-1}-q_j}-{\rm e}^{q_{j}-q_{j+1}} & \text{if }1<j<n,\\
{\rm e}^{q_{n-1}-q_n} & \text{if }j=n.
\end{cases}\label{eqtwoofhamiltod}
\end{gather}

In this representation, boundary conditions of $q_0=-\infty$ and $q_{n+1}=\infty$ have been imposed, which, formally, result in ${\rm e}^{q_0-q_1}={\rm e}^{q_n-q_{n+1}}=0$. These boundary conditions are chosen to truncate the lattice to a finite system.

\subsection{Flaschka's transformation and isospectrality}
Flaschka's transformation \cite{bib:fl} makes the variable replacement
\[
(p_1,\dots,p_n,q_1,\dots,q_n)\mapsto (a_1,\dots,a_n,b_1,\dots,b_{n-1})
\] given by setting $a_j=-p_j$ for $j=1,\dots,n$, and $b_j={\rm e}^{q_j-q_{j+1}}$ for $j=1,\dots,n-1$.
This is clearly a surjection of~$\mathbb{R}^{2n}$ onto the open subset of
$\mathbb{R}^{2n-1}$ with $b_i > 0$. It is evident that uniformly translating the position variables, $(q_1, \dots, q_n) \to (q_1 + c, \dots, q_n + c)$ leaves the equations~\mbox{(\ref{eqoneofhamiltod})--(\ref{eqtwoofhamiltod})} invariant. Fixing a center of mass for this particle system amounts to a particular cross-section of the fibration that the Flaschka transformation presents.

In Flaschka variables, the Toda lattice assumes the following simple form:
\begin{gather}
\dot{b}_j=(a_{j+1}-a_j)b_j,\qquad j=1,\dots,n-1,\label{flaschkatodaaeqn}\\
\dot{a}_j= \begin{cases}
b_1 & \text{if }j=1,\\
b_j-b_{j-1} & \text{if }1<j<n,\\
-b_{n-1} & \text{if }j=n.
\end{cases}\label{flaschkatodabeqn}
\end{gather}
One can arrange the variables neatly into a tridiagonal Hessenberg matrix
\begin{equation*} 
X \doteq \left[\begin{array}{cccc} a_1 & 1\\ b_1 & a_2 & \ddots\\ &\ddots & \ddots & 1\\ &&b_{n-1}&a_n\end{array}\right].
\end{equation*}

It is immediate from (\ref{flaschkatodabeqn}) that $\tr X$ is a constant of motion. This is conservation of momentum corresponding to the translation symmetry mentioned above.

Equations \eqref{flaschkatodaaeqn} and \eqref{flaschkatodabeqn} amount to the following matrix differential equation:
\begin{equation*}
\left[\begin{matrix} a_1 & 1\\ b_1 & a_2 & \ddots\\ &\ddots & \ddots & 1\\ &&b_{n-1}&a_n\end{matrix}\right]^\bullet
 =
\left[\def\arraystretch{2.0}\begin{matrix} b_1 &0 \\ (a_2-a_1)b_1 & b_2-b_1 & \ddots\\ &\ddots & \ddots & 0\\ &&(a_n-a_{n-1})b_{n-1}&-b_{n-1}\end{matrix}\right].
\end{equation*}

This matrix form of the Toda equations has the form of a {\it Lax equation},
\begin{equation} \label{Lax}
\dfrac{{\rm d}}{{\rm d}t}X=[X,\pi_-(X)],\qquad X(0)=X_0,
\end{equation}
where $\pi_-(X)$ denotes the projection of $X$ into its strictly lower part:
\begin{equation*} 
\pi_-(X)= \left[\begin{matrix} 0 & \\ b_1 & 0 & \\ &\ddots & \ddots & \\ &&b_{n-1}&0\end{matrix}\right].
\end{equation*}

It is a straightforward application of the product rule using (\ref{Lax}) to see that
\begin{equation*} 
\frac{{\rm d}}{{\rm d}t} \tr X^k = \tr \frac{{\rm d}}{{\rm d}t} X^k = \tr \big[ X^k , \pi_- (X) \big] = 0.
\end{equation*}
This implies the so-called {\it isospectrality} of the Toda lattice: the eigenvalues of $X$ remain invariant under the Toda flow.

\subsection{Sorting for Toda} \label{Sort}

In Section~\ref{crystal}, we will see that for $t \to \infty$, the Toda solution $X(t)$ limits to a matrix of the form
\begin{equation} \label{epslam}
\epsilon_\lambda = \left[\begin{matrix}
\lambda_1 & 1\\
&\lambda_2 & \ddots\\
&&\ddots & 1\\
&&&\lambda_n
\end{matrix}\right],
\end{equation}
where $ \lambda_1 > \dots > \lambda_n $ are the eigenvalues of $X_0$ which we will assume to be distinct in the remainder of this paper (by isospectrality these are the eigenvalues of $X(t)$ for all $t$). It is immediate from (\ref{Lax}) that $\epsilon_\lambda$ is a fixed point of the Toda flow as are all the other $n!$ matrices of this form with the eigenvalues permuted along the diagonal (and these are the only fixed points). We will also see in Section~\ref{crystal} that as $t \to - \infty$, $X(t)$ limits to the fixed point associated to the longest permutation (which reverses the order of the eigenvalues along the diagonal). Finally, from the Flaschka representation, one sees that the asymptotic velocities of the original Toda variables, in forward and backward time, are the eigenvalues of $X(t)$. So, what we have just described is just the {\it sorting property} of the Toda lattice. Precisely, the asymptotic dynamics is
\begin{gather} \label{asymptotesf}
q_k(t) = \alpha_k^+ t + \beta_k^+ + O\big({\rm e}^{-\delta t}\big),\\ \label{asymptotesb}
q_k(-t) = -\alpha_k^- t + \beta_k^- + O\big({\rm e}^{-\delta t}\big)
\end{gather}
as $t\to +\infty$, where
\[
\alpha_k^+ = \lim_{t\to +\infty} p_k(t) .
\]
To relate this to the Hessenberg form of the Toda lattice, recall that
\[
a_k(t) = -p_k(t) ,\qquad b_k(t) = {\rm e}^{q_k(t)-q_{k+1}(t)}.
\]

So, since the limiting form in forward time of $X(t)$ is given by (\ref{epslam}) (and by its reverse permutation in backward time), one has that
 \begin{gather*}
 \alpha_k^+ = - \lim\limits_{t\to +\infty}a_k(t) = -\lambda_k,\\
 \alpha_k^- = - \lim\limits_{t\to -\infty}a_k(t) = -\lambda_{n-k+1}
 \end{gather*}
 from which we note the symmetry $\alpha_k^-=\alpha_{n-k+1}^+$. Since
 $\lambda_1>\lambda_2>\cdots>\lambda_n$,
 $q_k(t) - q_{k+1}(t) \sim -(\lambda_k - \lambda_{k+1}) t$ as $t \to + \infty$. So, it is indeed the case that the particles spread apart linearly and sort, from which it follows that $b_k(t) \to 0$ exponentially, consistent with the stated limiting form~(\ref{epslam}) of $X(t)$. A similar statement obtains asymptotically in reverse time.

\subsection[Symes's discrete-time matrix dynamics and the discrete-time Toda lattice]{Symes's discrete-time matrix dynamics\\ and the discrete-time Toda lattice}\label{subsecsymesdnt}

In 1980, Symes \cite{bib:symes} proposed a discrete time flow based on matrix factorization that commutes with the continuous time Toda flow described previously. This factorization mirrors that described in Theorem~\ref{factorisationthmbkg}, more commonly referred to as Gaussian elimination or the $LU$-decomposition.

\begin{Note}
In this section and what follows, we continue to use $t$ to denote discrete time, i.e., $t\in \Z$ (as well as for continuous time, $t\in\R$). The temporal domain intended should be clear from context.
\end{Note}

We inductively define Symes's dynamics as a two-step discrete evolution on Hessenberg matrices. If at (discrete) time $t$ we have a matrix $X(t)$, we obtain $X(t+1)$ as follows:
\begin{enumerate}\itemsep=0pt
\item Perform Gaussian elimination to factor $X(t)=L(t)R(t)$, with $L(t)$ lower unipotent and~$R(t)$ upper triangular.
\item Permute the factors to define $X(t+1) = R(t)L(t)$.
\end{enumerate}

\begin{rem}
By construction, one has
\begin{equation} \label{symeseqn}
X(t+1)=R(t)L(t) = \big(L(t)^{-1}X(t)\big)L(t) = L(t)^{-1}X(t)L(t).
\end{equation}
Thus, this discrete evolution is given by conjugating a matrix by its lower unipotent factor. Since the spectrum of a matrix is invariant under conjugation, it follows that the eigenvalues are constants of motion for this discrete evolution; i.e., this discrete flow is isospectral.

Furthermore, if $X(t)$ is tridiagonal, then one can show that $L(t)$ is lower bi-diagonal with ones on its diagonal and~$R(t)$ is upper bi-diagonal with ones on its superdiagonal. Therefore, the product $X(t+1)=R(t)L(t)$ is itself once again a tridiagonal Hessenberg matrix; i.e., this flow preserves the Toda lattice phase space.
\end{rem}

It is of course not always the case that a given matrix has an $LU$-factorization in which the coefficients of the factors do not become singular. It is possible to continue the dynamics through these singularities, as we will see in Section~\ref{crystal}. For the phase shift behavior that we study, we are only interested in the long-time dynamics; for~$t$ of sufficiently large magnitude, $LU$-factorization always works.

Symes \cite{bib:symes} observed that this discrete evolution extends to a continuous evolution with Lax equation of the same form as the Toda lattice Lax equation (equation \eqref{Lax}), but with $\pi_-(X)$ replaced by $\pi_-(\log X)$:
\begin{equation}
\dfrac{{\rm d}}{{\rm d}t}X=[X,\pi_-(\log X)].\label{luloglaxen}
\end{equation}

\begin{rem} \label{rem:symes} In our applications here we will always take $X$ to have positive eigenvalues so that $\log X$ may be uniquely defined in terms of the principal branch of the logarithm along the positive real axis. From this it follows that this continuous flow is well-defined and commutes with the original Toda flow~\cite{bib:dlt, bib:watkins}. The Hamiltonian for~(\ref{luloglaxen}) is $H_{LU} = \operatorname{Tr}(X \log X - X)$. This is discussed further in Section~\ref{crystal}.
\end{rem}

To write the Symes discrete-time evolution out explicitly, let
\begin{equation}\label{ltrtdtoda}
L(t)=\arraycolsep=3.1pt\def\arraystretch{1.5}\left[\begin{matrix}
1\\
V_1^t&1\\&\ddots&\ddots\\
&&V_{n-1}^t&1
\end{matrix}\right],\qquad \text{and}\qquad R(t)=\arraycolsep=4.4pt\def\arraystretch{1.3}\left[\begin{matrix}
I_1^t & 1\\
&I_2^t&\ddots\\
&&\ddots&1\\
&&&I_n^t
\end{matrix}\right],
\end{equation}
then Symes's discrete-time evolution produces what has come to be known as the finite discrete-time Toda lattice:

\begin{defn}\label{defnbiinfdisctodafinite}
The finite discrete-time Toda lattice is the system
\begin{equation*}
\begin{cases}
I_i^{t+1}=I_i^t+V_i^t-V_{i-1}^{t+1},& i=1,\dots,n,\\
V_i^{t+1}=\dfrac{I_{i+1}^tV_i^t}{I_i^{t+1}},& i=1,\dots,n-1,\\
V_0^t=V_n^t = 0,&
\end{cases}
\end{equation*}
which is expressible as
\begin{equation*}
L(t+1)R(t+1)=R(t)L(t).
\end{equation*}
\end{defn}

\subsection{Maslov tropicalisation}
Tropical mathematics (see, for example, \cite{bib:l,bib:lmrs,bib:v}) is the study of the min-plus (or the max-plus) semiring, which we will now define. In this section, we follow the presentation given by Maslov~\cite{bib:lmrs}. The structure of the semiring $(\R_{\geq 0},+,\x)$ is carried over to the set $\R\cup\{\infty\}$ by a~family of bijections $D_\hbar$, for $\hbar>0$, given by
\begin{equation*}
D_\hbar(x)= \begin{cases}
-\hbar\ln x & \text{if }x\neq 0,\\
-\infty & \text{if }x= 0.
\end{cases}
\end{equation*}

This induces a family of semirings, parametrised by $\hbar>0$, $(\R\cup\{\infty\},\oplus_\hbar,\otimes_\hbar)$ with operations given by
\begin{gather*}
a\oplus_\hbar b =D_\hbar\big(D_\hbar^{-1}(a)+D_\hbar^{-1}(b)\big)
= \begin{cases}
-\hbar\ln\big({\rm e}^{-a/\hbar}+{\rm e}^{-b/\hbar}\big) & \text{if }a,b\neq \infty,\\
\min(a,b) & \text{otherwise},
\end{cases}\\
a\otimes_\hbar b =D_\hbar\big(D_\hbar^{-1}(a)D_\hbar^{-1}(b)\big) =a+b.
\end{gather*}

In the limit, $\hbar\to 0$, Maslov `dequantises' $(\R_{\geq 0},+,\x)$ to obtain the tropical semiring $(\R\cup\{\infty\},\allowbreak\min,+)$, where its addition is the usual $\min$ operation and its multiplication operation is usual addition, hence the name ``min-plus semiring''.

Maslov views this construction as an analogue of the correspondence principle from quantum mechanics, with $(\R_{\geq 0},+,\x)$ as the quantum object and $(\R\cup\{\infty\},\min,+)$ as its classical counterpart.

\begin{rem}
The limiting process through $D_\hbar$ allows tropicalisation to be carried out precisely on equations by pulling the analogous tropical variables back, mapping the resulting equation back to~$\R\cup\{\infty\}$, and then taking the limit as $\hbar\to 0$. This can be seen in action in equations~\eqref{tropvarsforbbstoda}--\eqref{wnoughtnboundaries}, and again in Section~\ref{bbsphasesec}.
\end{rem}

\subsection{The ultradiscrete Toda lattice (udToda)}\label{sec:udtoda}
Following~\cite{bib:tokihiro}, we perform the process of ultradiscretisation on the discrete-time Toda lattice to see that this results in the box-ball system.

Recall the discrete-time Toda lattice:
\begin{equation*}
\begin{cases}
V_0^t=V_n^t=0,& \\
I_i^{t+1}=I_i^t+V_i^t-V_{i-1}^{t+1}, & i=1,\dots,n,\\
V_i^{t+1}I_i^{t+1}=I_{i+1}^tV_i^t, & i=1,\dots,n-1.
\end{cases}
\end{equation*}

It can be shown that the discrete-time Toda lattice is equivalent to the following system
\begin{equation*}
\begin{cases}
V_0^t=V_n^t=0,& \\
I_i^{t+1}=V_i^t+\dfrac{I_i^t\cdots I_{1}^t}{I_{i-1}^{t+1}\cdots I_{1}^{t+1}}, & i=1,\dots,n,\\
V_i^{t+1}I_i^{t+1}=I_{i+1}^tV_i^t, & i=1,\dots,n-1.
\end{cases}
\end{equation*}
Making the change of variables
\begin{equation}
I_i^t={\rm e}^{-\frac{1}{\hbar}Q_i^t(\hbar)},\qquad
V_i^t={\rm e}^{-\frac{1}{\hbar}W_i^t(\hbar)},\label{tropvarsforbbstoda}
\end{equation}
one obtains
\begin{gather}
W_i^{t+1}(\hbar) =Q_{i+1}^t(\hbar)+W_i^t(\hbar)-Q_i^{t+1}(\hbar), \qquad i=1,\dots,n-1,\\
Q_i^{t+1}(\hbar) =-\hbar\log\Big({\rm e}^{-\frac{1}{\hbar}W_i^{t}(\hbar)}+{\rm e}^{-\frac{1}{\hbar}\left(\sum_{j=1}^i Q_j^t(\hbar)-\sum_{j=1}^{i-1} Q_j^{t+1}(\hbar)\right)}\Big), \qquad i=1,\dots,n,\\
W_0^t(\hbar) =W_n^t(\hbar)=\infty.
\end{gather}
Finally, assuming the limits
\begin{alignat}{3}
&W_i^t:=\lim\limits_{\hbar\to 0^+} W_i^t(\hbar),\qquad&&
W_i^{t+1}:=\lim\limits_{\hbar\to 0^+} W_i^{t+1}(\hbar),&\nonumber\\
&Q_i^t:=\lim\limits_{\hbar\to 0^+} Q_i^t(\hbar),\qquad&&
Q_i^{t+1}:=\lim\limits_{\hbar\to 0^+} Q_i^{t+1}(\hbar),&
\end{alignat}
exist, one obtains
\begin{gather}
W_i^{t+1} =Q_{i+1}^t+W_i^t-Q_i^{t+1},\qquad i=1,\dots,n-1,\\
Q_i^{t+1} =\min\Bigg(W_i^{t},\sum_{j=1}^i Q_j^t-\sum_{j=1}^{i-1} Q_j^{t+1}\Bigg),\qquad i=1,\dots,n,\\
W_0^t =W_n^t=\infty,\label{wnoughtnboundaries}
\end{gather}
which are known as the ultradiscrete Toda (udToda) equations.

\subsection{Centre of mass}\label{sec:centreofmass}
By the isospectrality of dToda (\ref{symeseqn}), the trace is also a constant of motion for dToda, in fact, it is a \textit{Casimir} of the system. A quick computation of the diagonal entries of $L(t)R(t)$ in~\eqref{ltrtdtoda} reveals that the corresponding tropicalisation of
\[
\operatorname{tr} X(t)=\sum_{i=1}^n I_i^t+\sum_{i=1}^{n-1}V_i^t
\]
 is
\[
 \min\big(Q_1^t,\dots,Q_n^t,W_1^t,\dots,W_{n-1}^t\big).
\]
However, with the box-ball interpretation of the ultradiscretised dToda equations, the balls are seen to be moving right. In this interpretation, one can assign site numbers to the boxes (indexing them by~$\Z$) and the sum of the site numbers of the first ball in each block can be seen to grow linearly with rate $M=\sum_{i=1}^n Q_i^t$ which is another constant of motion (the tropical analogue of the determinant of $X(t)=L(t)R(t)$). This observation is key to the main result of this paper and is stated precisely and proved later in Lemma~\ref{lemma:sumsitesgrowsbyM}.

\section{The principal embedding} \label{crystal}
 \subsection{Kostant's theorem and the principal embedding}

The matrix reformulation of the Toda lattice in Flaschka's variables led to an elegant Lie-theoretic method for its solution that extended to the solving of {\it Lax equations} more generally \cite{bib:adler,bib:kostant,bib:symes78}). This is typically referred to as the Adler--Kostant--Symes (AKS) factorization theorem. Prior to stating that theorem we introduce some notation related to Lie theoretic factorizations.

We consider the Lie algebra decomposition of $n \times n$ matrices
\begin{gather*} 
 \frak{g} = \frak{gl}(n, \mathbb{R}) = \frak{n}_- \oplus \frak{b}_+,
 \end{gather*}
 where $\frak{n}_-$ is the lower triangular nilradical subalgebra and $\frak{b}_+$ is a maximal solvable subalgebra, referred to as a {\it Borel subalgebra}:
 \begin{gather*}
 \frak{n}_- = \left(\begin{matrix}
0 & & & &\\
*& 0 & & &\\
\vdots & \ddots & \ddots & \ddots &\\
\vdots & & \ddots & \ddots & \\
* & \dots & \dots & * & 0
\end{matrix} \right), \qquad
\frak{b}_+ = \left(\begin{matrix}
* & * &\dots &\dots &*\\
& * & * & &\vdots\\
& & \ddots & \ddots &\vdots\\
& & & \ddots & *\\
& & & & *
\end{matrix} \right).
\end{gather*}
We will also use $\frak{b}_-$, the transpose of $\frak{b}_+$. Employing the regular nilpotent element,
\begin{gather*}
\epsilon = \left(\begin{matrix}
0 & 1 & & &\\
& 0 & 1 & &\\
& & \ddots & \ddots &\\
& & & \ddots & 1\\
& & & & 0
\end{matrix} \right)
\end{gather*}
one defines an extended Toda phase space
\begin{gather*}
\epsilon + \frak{b}_- = \left(\begin{matrix}
* & 1 & & &\\
*& * & 1 & &\\
\vdots & \ddots & \ddots & \ddots &\\
\vdots & & \ddots & \ddots & 1\\
* & \dots & \dots & * & *
\end{matrix} \right),
\end{gather*}
 (which is the space of all lower Hessenberg matrices) on which the Toda Lax equation~(\ref{Lax}) as well the discrete time Toda equation~(\ref{symeseqn}) and their extensions will be defined. To describe explicit solutions one introduces the natural projections
\begin{gather*}
 \pi_- \colon \ \frak{g} \to \frak{n}_-, \qquad
 \Pi_- \colon \ G \to N_-, \\
 \pi_+ \colon \ \frak{g} \to \frak{b}_+, \qquad
 \Pi_+ \colon \ G \to B_+,
\end{gather*}
where $G = {\rm GL}(n, \mathbb{R})$, $N_-$ is the lower unipotent matrices, is the exponential group of the algebra~$\frak{n}_-$ and~$B_+$ is the invertible upper triangular matrices, is the exponential group of the algebra~$\frak{b}_+$. $\Pi_\pm$ are defined on the open dense subset of~$G$ where there is an~$LU$-factorization.

\begin{thm}[the factorization theorem, \cite{bib:adler,bib:kostant,bib:symes78}]\label{factorisationthmbkg}
To solve
\begin{equation} \label{Lax2}
\dfrac{{\rm d}}{{\rm d}t}X=[X,\pi_-(X)],\qquad X(0)=X_0,
\end{equation}
factor ${\rm e}^{X_0 t}=\Pi_-\big({\rm e}^{X_0t}\big)\Pi_+\big({\rm e}^{X_0t}\big)$, if possible $($locally it is$)$. Then, the solution is given by
\begin{equation}\label{fulltodasolnfactorsconj}
X(t)=\Pi_-^{-1}\big({\rm e}^{X_0t}\big)X_0\Pi_-\big({\rm e}^{X_0t}\big).
\end{equation}
\end{thm}

The dynamical system (\ref{Lax2}) on $\epsilon + \frak{b}_-$ is referred to as the {\it full Kostant--Toda lattice} and, per the factorization theorem, its solution is locally given by~(\ref{fulltodasolnfactorsconj}). The tridiagonal Hessenberg matrices are an invariant subspace for the full Toda flow and so basic facts about the original dynamical systems,~(\ref{Lax}) and~(\ref{symeseqn}), discussed in Section~\ref{TODA} are directly recovered from what we do in this section.

This factorization result can be extended to discrete-time Toda. To avoid
some technical complications, in what follows, we fix the spectrum $\lambda=(\lambda_1,\dots,\lambda_n)$ and assume that $\lambda_1 >\dots >\lambda_n > 0$. Under this assumption one has directly the following extension.

\begin{lem}[\cite{bib:dlt,bib:symes78, bib:watkins}]
To solve the discrete-time Toda lattice with initial condition \mbox{$X(0)\!=\!X_0$}, factor ${\rm e}^{t\log X_0}=X_0^t=\Pi_-\big(X_0^t\big)\Pi_+\big(X_0^t\big)$, if possible. Then, the solution is given by
\begin{equation*}
X(t)=\Pi_-^{-1}\big(X_0^t\big)X_0\Pi_-\big(X_0^t\big),
\end{equation*}
for all $t\in\N\cup\{0\}$.
\end{lem}

 These factorizations provide the basis for our analysis of the phase shift formulas in both continuous and discrete time Toda. This is facilitated by a key result due to Kostant that provides a natural embedding of the Toda dynamics into a flag manifold that plays a role analogous to action-angle variables in classical integrable systems theory.

\begin{thm}[\cite{bib:kostant}] \label{thm:kostant}For each $X \in \epsilon + \frak{b}_-$ there exists a~unique lower unipotent $L \in N_-$ such that
$
X = L \epsilon_\lambda L^{-1}$,
where $\epsilon_\lambda$ is as defined by \eqref{epslam}.
\end{thm}

Set
\[
\mathcal{F}_\lambda = \{ X \in \epsilon + \frak{b}_- \colon \sigma(X) = \lambda_1 > \cdots > \lambda_n\},
\]
the isospectral manifold, and define the {\it principal embedding} into a compact, homogeneous space, known as the {\it flag manifold}, to be
\begin{align*}
\kappa_\lambda \colon \ \mathcal{F}_\lambda &\to G/B_+,\\
X &\mapsto L^{-1} \mod B_+.
\end{align*}
This mapping simultaneously linearizes and completes the Toda flows~\cite{bib:efh,bib:efs}:

If $X_0 = L_0 \epsilon_\lambda L_0^{-1}$ then
\begin{gather*}
X(t) = \Pi_-^{-1}\big({\rm e}^{t X_0}\big)X_0 \Pi_-\big({\rm e}^{t X_0}\big) = \Pi_-^{-1}\big({\rm e}^{t X_0}\big)L_0\epsilon_\lambda L_0^{-1} \Pi_-\big({\rm e}^{t X_0}\big),\\
\kappa_\lambda(X(t)) = L_0^{-1} \Pi_-\big({\rm e}^{t X_0}\big) \mod B_+ \\
\hphantom{\kappa_\lambda(X(t))}{} = L_0^{-1} \big({\rm e}^{t X_0}\big) \mod B_+\\
\hphantom{\kappa_\lambda(X(t))}{}= {\rm e}^{\epsilon_\lambda t} L_0^{-1} \mod B_+.
\end{gather*}
Thus, one sees that under the principal embedding, the Toda flow maps to a linear semigroup action by ${\rm e}^{\epsilon_\lambda t}$ on $G/B_+$ through the principal image of the initial value,
$\kappa_\lambda(X_0) = L_0^{-1}$ which exists for all time.

In Remark~\ref{rem:symes} we pointed out that the continuous Toda flow commutes with the discrete time Toda evolution. In fact, there is a hierarchy of continuous flows commuting with dToda and with one another given by Lax equations of the form
\begin{gather} \label{hierarchy}
\frac{{\rm d}}{{\rm d}t_m} X = \big[ X, \pi_- X^m\big].
\end{gather}
On $\mathcal{F}_\lambda$ the first $n-1$ of these flows are locally independent and their respective images under the principal embedding has the form
\[
{\rm e}^{\epsilon^m_\lambda t_m} L_0^{-1} \mod B_+
\]
and together they generate a torus action on $G/B_+$~\cite{bib:efs} (see Section~\ref{subsec:eigenvectorstorus}).

\subsection[The fundamental element epsilon\_lambda]{The fundamental element $\boldsymbol{\epsilon_\lambda}$}

As mentioned earlier, we are going to see that the Toda flow limits to $\epsilon_\lambda$ (\ref{epslam}) in large forward time. This element is distinguished by the fact that it lies in the intersection
$(\epsilon + \frak{b}_-)_\lambda \cap \frak{b_+}$, and in that regard it is essentially unique. (There are $n!$ elements in this intersection corresponding to permutations of the distinct eigenvalues. In fact in large backwards time the Toda flow limits to the element corresponding to the longest permutation.) Since this element will play such a~central role in our initial calculations, we take a few moments to discuss some of its relevant features.

\subsubsection[Diagonalizing epsilon\_lambda]{Diagonalizing $\boldsymbol{\epsilon_\lambda}$}

 Because we assume $\lambda_i\neq \lambda_j$ for all $i\neq j$, we can of course diagonalise any matrix in $\mathcal{F}_\lambda$. The following result, which is an explicit form of Lemma~7 in~\cite{bib:efh}, describes a diagonalisation of $\epsilon_\lambda$.

\begin{lem}\label{uppertodiagefhexpl}
If $\lambda_1,\dots,\lambda_n$ are distinct, then one has $\epsilon_\lambda=U D_\lambda U^{-1}$, where $U=(u_{ij})$ is the upper triangular matrix given by
\[
u_{ij}=\prod_{k=1}^{i-1}(\lambda_j-\lambda_k),\qquad 1\leq i\leq j\leq n,
\]
and $D_\lambda=\epsilon_\lambda-\epsilon=\operatorname{diag}(\lambda_1,\dots,\lambda_n)$.
\end{lem}

\begin{proof}
The matrix $U$ is clearly invertible since
\[\det(U)=\prod_{i=1}^n \prod _{k=1}^{i-1}(\lambda_i-\lambda_k)=\prod_{1\leq k<i\leq n}(\lambda_i-\lambda_k)\]
and as stated above, we have assumed $\lambda_i\neq \lambda_j$ for all $i\neq j$.

It just remains to show that $\epsilon_\lambda U = U D_\lambda$. Let $u_j=(u_{ij})_{1\leq i\leq n}$ be the $j$-th column of $U$, then for $i<n$:
\begin{align*}
(\epsilon_\lambda u_j)_i&=\sum_{k=1}^n (\epsilon_\lambda)_{ik}u_{kj}=\lambda_i u_{ij}+u_{i+1,j}
=\lambda_i\prod_{k=1}^{i-1}(\lambda_j-\lambda_k)+\prod_{k=1}^{i}(\lambda_j-\lambda_k)\\
&=(\lambda_i+\lambda_j-\lambda_i)\prod_{k=1}^{i-1}(\lambda_j-\lambda_k)
=\lambda_j \prod_{k=1}^{i-1}(\lambda_j-\lambda_k)=\lambda_j (u_j)_i.
\end{align*}
For $i=n$, we simply have
\begin{gather*}
(\epsilon_\lambda u_j)_n =\lambda_n (u_j)_n
 = \begin{cases} 0, & j<n,\\ \lambda_j (u_j)_n, & j=n\end{cases}
 =\lambda_j (u_j)_n.
\end{gather*}
Thus, $\epsilon_\lambda u_j=\lambda_j u_j$ for each $j$.
\end{proof}

\begin{rem}\label{rem:diagonaliseepsilon}
In the remainder of this paper, it is convenient to work with a unipotent matrix of eigenvectors for $\epsilon_\lambda$. Since $U$ is triangular, we achieve this by scaling each column of the matrix~$U$ by the reciprocal of the diagonal entry for that column. The result is that we can choose\looseness=-1
\[(U)_{ij} = \begin{cases}
\displaystyle \prod\limits_{i\leq m<j} \dfrac{1}{\lambda_j-\lambda_m} & \text{for }i\leq j,\\
0 & \text{otherwise},
\end{cases}
\]
and still maintain $\epsilon_\lambda = UD_\lambda U^{-1}$. From now on, $U$ shall refer to this upper unipotent matrix.
\end{rem}

\subsubsection{Eigenvectors and the torus embedding}\label{subsec:eigenvectorstorus}

In this section we describe a slight modification of the principal embedding which displays more clearly the linearized character of the Toda flows that these embeddings reveal. This is the homogeneous space analogue of the linearization in terms of action-angle variables familiar from classical integrable systems theory~\cite{bib:arnold}.

Combining Lemma \ref{uppertodiagefhexpl} with Theorem~\ref{thm:kostant} provides a diagonalization of $X \in \epsilon + \frak{b}_-$:
\begin{gather}\label{LUDiag}
X = LU D_\lambda U^{-1}L^{-1}.
\end{gather}
Cross-multiplying in two different ways this is equivalent to each of two different representations
\begin{gather} \label{leftevec}
\big(U^{-1}L^{-1}\big) X = D_\lambda \big(U^{-1}L^{-1}\big),\\ \label{rightevec}
 X (LU) = (LU) D_\lambda .
\end{gather}
The representation in (\ref{leftevec}) presents $\big(U^{-1}L^{-1}\big)$ as a matrix whose rows are the left eigenvectors of~$X$ while~(\ref{rightevec}) presents $(LU)$ as a matrix whose columns are the right eigenvectors of~$X$. Based on this one is led to consider an alternative flag manifold embedding,
\begin{align*}
\operatorname{tor}_\lambda \colon \ \mathcal{F}_\lambda &\to G/B_+,\\
X &\mapsto U^{-1}L^{-1} \mod B_+
\end{align*}
from $X$ to its matrix of left eigenvectors. As before one can track the Toda dynamics through this embedding
and find that
\begin{gather*}
\operatorname{tor}_\lambda(X(t)) ={\rm e}^{t D_\lambda}U^{-1}L^{-1} \mod B_+.
\end{gather*}
Also, as before, this can be extended to the commuting hierarchy of Toda flows (\ref{hierarchy}),
\begin{gather*}
\operatorname{tor}_\lambda(X(t_1, \dots, t_n))  =
\exp\left({\sum_{m=1}^{n-1} t_m D^m_\lambda}\right)U^{-1}L^{-1} \mod B_+.
\end{gather*}
 This amounts to a toric action
(left multiplication by $(\mathbb{R}^*)^n$) on the flag manifold, explaining why $\operatorname{tor}_\lambda$ is referred to as the \textit{torus embedding}. This is a replacement of the nonlinear Toda dynamics by a linear semigroup action on $G/B_+$.
\subsubsection{Other Lie algebras}

The language we have used in this section to describe the Toda phase space and dynamics (Borel and nilradical subalgebras, regular nilpotent elements) have immediate extensions to the setting of general real semisimple Lie algebras to define the so-called {\it generalized Toda lattices}~\cite{bib:kostant2}. The structure and results about the flag manifold embeddings was similarly carried out in \cite{bib:efs}. Kostant also introduces the analogue of $\epsilon_\lambda$ in Theorem~1.5 of~\cite{bib:kostant2} and its diagonalization (Lemma~3.52).

\subsection{Painlev\'e balances}

It was pointed out in Theorem \ref{factorisationthmbkg} that lower-upper factorization of ${\rm e}^{t X_0}$ does not necessarily hold for all values of $t$. When this happens some entry of $X(t)$ will become infinite. In fact these entries are meromorphic (which fact is referred to as the Painlev\'e property). We refer to~\cite{bib:fh} for further details about this. But as we saw in the previous section, the evolution under the torus embedding is just a multi-scaling action in the flag manifold, which is compact, and so the flows exist for all time.
A geometric understanding of what is happening is provided by the {\it Birkhoff cell decomposition} of~$G/B_+$. The flag manifold may be written as a disjoint union,
\begin{gather*}
G/B_+ = \bigcup_{\widehat{w} \in W} N_- \widehat{w} B_+/B_+.
\end{gather*}
In this representation, $W$ is the group of $n \times n$
permutation matrices. The $n!$ sets, $N_- \widehat{w} B_+/B_+$, are cells (homeomorphic to a Euclidean space) called Birkhoff cells. The cell $N_-B_+/B_+$ is open and dense in~$G/B_+$ and referred to as the ``big cell''. The other cells are lower dimensional with co-dimension equal to the {\it length of $\widehat{w}$} (corresponding to the minimal number of pivots required to carry out Gaussian elimination of a representative group element for a point in that cell). By Kostant's theorem~\ref{thm:kostant}, all elements of $\epsilon + \frak{b}_-$ embed into the big cell. Under the Toda flow, ``blow-up'' of~$X(t)$ occurs precisely when ${\rm e}^{t X_0}$ cannot be factored. However, under $\kappa_\lambda$, this flow simply continues to a smaller Birkhoff cell corresponding to the sequence of pivots required to perform Gaussian elimination on
${\rm e}^{t X_0}$. Thus the singularity is ``resolved'' and the flow may be continued to arbitrary times. A similar statement applies to the flows of the full Toda hierarchy. A~consequence of this observation is that the meromorphic singularity structure of the Toda lattice solutions may be explicitly described as follows.

\begin{prop}[\cite{bib:fh}]
Let $\Theta\subseteq\{1,\dots,n-1\}$ be nonempty. There is a Laurent series solution of equations \eqref{flaschkatodaaeqn}--\eqref{flaschkatodabeqn} having the form
\begin{gather*}
b_j(t)=-\frac{\sigma_j^\Theta}{t^2}+\text{$($Taylor$)$ near} \ t=0,\qquad \text{if }j\in\Theta,\\
b_j(t)=\text{$($Taylor$)$ near} \ t=0,\qquad \text{if }j\not\in\Theta.
\end{gather*}
This solution depends on $2(n-1)-|\Theta|$ free parameters.
\end{prop}

Each type of series in the proposition is called a Painlev\'{e} balance $\Theta$ of dimension equal to the number of free parameters.

\begin{cor}[\cite{bib:fh}]\label{lowestbalancedim}
When $\Theta=\{1,\dots,n-1\}$, which is called the lowest balance, its dimension is $2(n-1)-(n-1)=n-1$.
\end{cor}

\begin{cor}[\cite{bib:efh}]\label{lowestbalance}
Restricted to the isospectral manifold $\mathcal{F}_\lambda$, the Toda flow, ${\rm e}^{\epsilon_\lambda t}\widehat{w}_0$,
passes through $\widehat{w}_0 \mod B_+$ as $t \to 0$, where $\widehat{w}_0$ denotes the permutation matrix for the reverse permutation, i.e., the longest permutation
\end{cor}

\subsection{Centering the tridiagonal flows}

In the remainder of this paper we will restrict attention to the tridiagonal flows we began with in Section~\ref{TODA} but largely view them in terms of their embedding in the flag manifold, $G/B_+$.

The phase shift formulas we describe do not depend on initial conditions. We have already seen this reflected in the fact (see Section \ref{Sort}) that the total momentum is a constant (in fact a Casimir) for the motion. Knowing this, it is helpful to choose an initial point that is appropriately ``centered" in the flag manifold. The centering we use is based on Corollary \ref{lowestbalance}. This centering, based on a different argument, was also used in \cite{bib:o}.

Motivated by the discussion at the end of the previous subsection, one begins with an upper triangular matrix of the form
\begin{equation} \label{torus}
b_0:=\exp\left(\sum_{i=1}^{n-1}\epsilon_\lambda^i t_i\right),
\end{equation}
where the values of the $t_i$ are fixed. Next consider
\begin{equation*}
b(t)={\rm e}^{\epsilon_\lambda t}b_0,
\end{equation*}
where $t$ is the Toda time parameter. Finally, we define from this a path in $N_-$ via
\begin{equation*}
L(t) = \Pi_- \big[b(t)\widehat{w}_0\big].
\end{equation*}
In \cite{bib:efh} it is shown that
\begin{equation*}
b(t)~\mapsto~L(t)^{-1}\epsilon_\lambda L(t) \in \mathcal{F}_\lambda.
\end{equation*}
$X(t) = L(t)^{-1}\epsilon_\lambda L(t)$ solves~(\ref{Lax}) and as the parameters $t_i$ are varied one sweeps out all isospectral solutions. The center of this invariant set, corresponding to all $t_i = 0$ is the point which is the smallest Birkhoff cell in~$G/B_+$.

\section[Representation theory and tau-functions]{Representation theory and $\boldsymbol{\tau}$-functions}\label{repntheorysec}

Let $G$ denote the group ${\rm GL}(\ell +1)$ and let $(\rho_n, V_n)$ denote the $n^{\rm th}$ fundamental representation of $G$. $\rho_1$ is the
{\it birth representation} with $V_1 = \mathbb{C}^{\ell + 1}$ defined by
\[
\rho_1(g) v = g v
\]
for $v \in V_1$. This induces a representation on the exterior algebra of $V_1$ that defines the remaining fundamental representations, respectively, on
$V_n = \bigwedge^n \mathbb{C}^{\ell + 1}$ given by
\[
\rho_n(g) v_1 \wedge \dots \wedge v_n = g v_1 \wedge \dots \wedge g v_n.
\]
With respect to the standard basis $e_1, \dots, e_{\ell + 1}$ of $\mathbb{C}^{\ell + 1}$
one defines a Hermitian inner product on $V_n$ by
\[
\langle e_{i_1} \wedge \dots \wedge e_{i_n}, e_{j_1} \wedge \dots \wedge e_{j_n}\rangle = \delta_{i_1, j_1} \cdots \delta_{i_n, j_n}.
\]
Set $v^{(n)} = e_1 \wedge \dots \wedge e_n$ and $v_{(n)} = e_{\ell - n + 2} \wedge \dots \wedge e_{\ell + 1}$. These are, respectively, the highest and lowest weight vectors, with respect to lexicographic order, of the representation~$\rho_n$~\cite{bib:fuha}.

We can now define the $n^{\rm th}$ {\it $\tau$-function} to be{\samepage
\[
\tau_n(t) = \big\langle \exp(t X_0) v^{(n)}, v^{(n)}\big\rangle,
\]
which is just the $n^{\rm th}$ principal minor of $\exp{t X_0}$.}

The solution of the Toda lattice equations in Hessenberg form is then given
explicitly by
\begin{gather} \label{btau}
b_n(t) = b_n(0) \frac{\tau_{n-1}(t) \tau_{n+1}(t)}{\tau_n^2(t)},\\ 
a_n(t) = \frac{{\rm d}}{{\rm d}t}\big( \log \tau_n(t)-\log \tau_{n-1}(t)\big).\nonumber
\end{gather}

\section{Toda phase shift in terms of weighted path counting}\label{todaphaseshiftsec}
We now come to our novel derivation of the phase shift formula for the continuous time Toda lattice. The original derivation was due to Moser~\cite{bib:moser} based on a continued fraction representation. Subsequently, Kostant~\cite{bib:kostant2} gave a fully Lie theoretic derivation. Here we present a direct combinatorial derivation based on the lemma of Gessel, Viennot and Lindstr{\"o}m for enumerating weighted paths in directed graphs. We feel this approach is less \textit{ad hoc} than Moser's and less technical than Kostant's while capturing the essential connection to the principal embedding through $\epsilon_\lambda$. This also makes it more amenable to ultradiscretization which is the main result of this paper and to future generalizations to other box ball systems.

The full details of our derivation are in the Appendices; this chapter just summarizes the essential results of that analysis, in Lemmas~\ref{toprightUminor} and~\ref{Uinverselemma}, Corollary~\ref{toprightUinverseminor}, and shows how to use these to derive the continuous time phase shift. It is then straightforward, based on what we have already shown, to modify this formula for discrete time Toda (which is carried out in Section~\ref{dTODA}).

Moser's phase shift, in terms of the original Hamiltonian variables (see (\ref{asymptotesf}) and (\ref{asymptotesb})), is given by
\[
\beta_{n-k+1}^+ = \beta_k^- + \sum_{j\neq k}\phi_{jk}(\alpha^-),
\]
where
\begin{gather*}
\phi_{jk}(\alpha^-)
 =
\begin{cases}
\log(\alpha_j^--\alpha_k^-)^2 & \text{for }j<k,\\
-\log(\alpha_j^--\alpha_k^-)^2 & \text{for }j>k
\end{cases}
 =
\begin{cases}
\log(\lambda_{n-k+1}-\lambda_{n-j+1})^2 & \text{for }j<k,\\
-\log(\lambda_{n-k+1}-\lambda_{n-j+1})^2 & \text{for }j>k.
\end{cases}
\end{gather*}

We reformulate this in terms of the Hessenberg matrix form of Flaschka's variables where it becomes equivalent to an elegant pairwise multiplicative formula for the sub-diagonal entries.

We have
\[
b_k(-t) = {\rm e}^{q_k(-t)-q_{k+1}(-t)} = {\rm e}^{(\lambda_{n-k+1}-\lambda_{n-k})t}{\rm e}^{\beta_k^--\beta_{k+1}^-+O({\rm e}^{-\delta t})},
\]
and
\[
b_{n-k}(t) = {\rm e}^{q_{n-k}(t)-q_{n-k+1}(t)} = {\rm e}^{(\lambda_{n-k+1}-\lambda_{n-k})t}{\rm e}^{\beta_{n-k}^+-\beta_{n-k+1}^++O({\rm e}^{-\delta t})}.
\]
Then consider
\begin{gather*}
\lim_{t\to +\infty} b_{n-k}(t)b_k(-t) {\rm e}^{2(\lambda_{n-k}-\lambda_{n-k+1})t}
= {\rm e}^{(\beta_{n-k}^+-\beta_{k+1}^-)-(\beta_{n-k+1}^+-\beta_k^-)}\\
\qquad{} =\exp\Bigg(\sum_{j\neq k+1}\phi_{j,k+1}(\alpha^-) - \sum_{j\neq k}\phi_{j,k}(\alpha^-) \Bigg) \\
\qquad{} =\exp \Bigg(\Bigg(
\sum_{j<k+1}\log (\lambda_{n-k}-\lambda_{n-j+1} )^2
-\sum_{j>k+1}\log (\lambda_{n-k}-\lambda_{n-j+1} )^2\Bigg)\\
 \qquad\quad{} -\sum_{j<k}\log (\lambda_{n-k+1}-\lambda_{n-j+1} )^2
+\sum_{j>k}\log (\lambda_{n-k+1}-\lambda_{n-j+1} )^2
\Bigg)\\
\qquad{} = \dfrac{\prod\limits_{j<k+1} (\lambda_{n-k}-\lambda_{n-j+1} )^2\prod\limits_{j>k} (\lambda_{n-k+1}-\lambda_{n-j+1} )^2}
{\prod\limits_{j>k+1} (\lambda_{n-k}-\lambda_{n-j+1} )^2\prod\limits_{j<k} (\lambda_{n-k+1}-\lambda_{n-j+1} )^2}.
\end{gather*}

This leads us to the following analogue of the Moser's phase shift formula for the off-diagonal entries, but in the Hessenberg setting:

\begin{thm}\label{classicaltodaphase}
In the Hessenberg setting, the following asymptotic formula holds for the subdiagonal entries of the Toda flow:
\begin{gather}
\lim\limits_{t\to +\infty}b_{n-k}(t)b_k(-t) {\rm e}^{2(\lambda_{n-k}-\lambda_{n-k+1})t}
= \dfrac{\prod\limits_{j<k+1}\! (\lambda_{n-k}-\lambda_{n-j+1} )^2\!\prod\limits_{j>k}\! (\lambda_{n-k+1}-\lambda_{n-j+1} )^2}
{\prod\limits_{j>k+1}\! (\lambda_{n-k}-\lambda_{n-j+1} )^2\!\prod\limits_{j<k}\! (\lambda_{n-k+1}-\lambda_{n-j+1} )^2}.\!\!\!
\label{multctstimephshift}
\end{gather}
\end{thm}

Although the analogy between this theorem and Moser's original result is apparent, we will provide a novel proof for this result. This analysis proceeds from a direct deconstruction of $\tau$-functions appearing in the representation~(\ref{btau}).

Recall that
\[ \epsilon_\lambda = \left[\begin{matrix}
\lambda_1 & 1\\
&\lambda_2 & \ddots\\
&&\ddots & 1\\
&&&\lambda_n
\end{matrix}\right] = \epsilon + D_\lambda,
\]
where $D_\lambda = \operatorname{diag}(\lambda_1,\dots,\lambda_n)$ and, from Remark \ref{rem:diagonaliseepsilon}, that it is diagonalised by the matrix
\[
(U)_{ij} = \begin{cases}
\displaystyle \prod\limits_{i\leq m<j} \dfrac{1}{\lambda_j-\lambda_m} & \text{for }i\leq j,\\
0 & \text{otherwise},
\end{cases}
\]
so that
\[
\epsilon_\lambda = U D_\lambda U^{-1}.
\]

We also recall here notation used by~\cite{bib:o} for describing certain minor determinants of matrices:

\begin{defn}\label{defnoconnelldelta}
For an $n\x n$ matrix $X$ and $1\leq k\leq m\leq n$, define
\begin{equation*}
 \Delta_k^m(X) = \det [x_{ij} ]_{1\leq i\leq k, m-k+1\leq j\leq m},
\end{equation*}
where, by convention, $\Delta_0^m(X)$ is taken to be $1$. In the particular case that $m=n$, $\Delta_k^n(X)$ is the top-right $k\x k$ minor determinant of~$X$.
\end{defn}

In proving Theorem~\ref{classicaltodaphase}, the following three results will be key:

\begin{lem}\label{toprightUminor}
\[ \Delta_k^n(U) = \prod_{i=n-k+1}^n \prod_{j=1}^{n-k}\dfrac{1}{\lambda_j-\lambda_i}.\]
\end{lem}

\begin{lem}\label{Uinverselemma}
\[ (U^{-1})_{ij} = \begin{cases}
0 & i>j,\\
\displaystyle \prod_{i<m\leq j}\dfrac{1}{\lambda_i-\lambda_m} & i\leq j.
\end{cases}
\]
\end{lem}

\begin{cor}\label{toprightUinverseminor}
\[ \Delta_k^n(U^{-1}) = \prod_{j=1}^k \prod_{i=k+1}^{n}\dfrac{1}{\lambda_j-\lambda_i}.\]
\end{cor}

Corollary~\ref{toprightUinverseminor} will later be shown to follow from Lemmas~\ref{toprightUminor} and~\ref{Uinverselemma}.

\subsection{Proof of multiplicative formula for the classical Toda phase shift}

We recall the level set parametrisation (\ref{torus})
\[ b_0:=\exp\left(\sum_{i=1}^{n-1}\epsilon_\lambda^i t_i\right),\]
with $\lambda_1>\lambda_2>\cdots >\lambda_n$,
define
$ L_0 := \Pi_- (b_0\widehat{w}_0)$,
and let
$ X_0:= L_0^{-1}\epsilon_\lambda L_0$
be the initial Hessenberg matrix.

\begin{lem}\label{bzerolemma}
For each $k$, one has
\[ b_k(0) = \dfrac{\Delta_{k+1}^n(b_0)\Delta_{k-1}^n(b_0)}{(\Delta_{k}^n(b_0))^2}.\]
\end{lem}
\begin{proof}
Define $B_0= \Pi_+(b_0\widehat{w}_0)$, so that $b_0\widehat{w}_0=L_0B_0$. Then
\begin{align*}
 B_0^{-1}X_0B_0&=(L_0B_0)^{-1}\epsilon_\lambda L_0B_0
 =\widehat{w}_0^{-1}b_0^{-1}\epsilon_\lambda b_0\widehat{w}_0\\
 &=\widehat{w}_0^{-1}U {\rm e}^{-\sum\limits_{i=1}^{n-1}t_iD_\lambda^i}U^{-1}UD_\lambda U^{-1} U {\rm e}^{\sum\limits_{i=1}^{n-1}t_iD_\lambda^i}U^{-1}\widehat{w}_0\\
 &=\widehat{w}_0^{-1}UD_\lambda U^{-1}\widehat{w}_0
 =\widehat{w}_0^{-1}\epsilon_\lambda \widehat{w}_0.
\end{align*}
From this, we see that
\begin{equation*}
 X_0 = B_0\widehat{w}_0^{-1}\epsilon_\lambda \widehat{w}_0 B_0^{-1}.
\end{equation*}
Splitting $\epsilon_\lambda = \epsilon+D_\lambda$, we study $\widehat{w}_0^{-1}\epsilon_\lambda \widehat{w}_0$ as follows
\begin{gather*}
 \widehat{w}_0^{-1}\epsilon_\lambda \widehat{w}_0
 =\widehat{w}_0^{-1}\epsilon \widehat{w}_0
 +
 \widehat{w}_0^{-1}D_\lambda \widehat{w}_0.
\end{gather*}
By recalling that
\[ \widehat{w}_0 = \left[\begin{matrix}
&&&&1\\
&&&1\\
&&1\\
&\iddots\\
1
\end{matrix}\right],
\]
one can check that $\widehat{w}_0^{-1}\epsilon_\lambda \widehat{w}_0=\hat{\epsilon}_\lambda$, where
\begin{equation*}
 \hat{\epsilon}_\lambda
 =\left[\begin{matrix}
\lambda_n \\
1& \lambda_{n-1}\\
&1 & \ddots\\
&&\ddots & \lambda_2\\
&&&1 & \lambda_1
\end{matrix}\right].
\end{equation*}
If one writes
\[ B_0 = \left[\begin{matrix}d_1 &&*\\ & d_2 \\
&&\cdots\\
&&&d_n\end{matrix}\right],
\]
then
\[
X_0 = B_0\hat{\epsilon}_\lambda B_0^{-1} =\renewcommand{\arraystretch}{1.5} \left[\begin{matrix}
* &&&&*\\
d_2d_1^{-1} & * \\
& d_3d_2^{-1} & \ddots\\
&&\ddots & *&\\
&&&d_{n}d_{n-1}^{-1} & *
\end{matrix}\right].
\]
To conclude, we recall that the diagonal part of $L_0B_0$, which is $\operatorname{diag}(B_0)$, is given by
\[d_j= \frac{\tau_j}{\tau_{j-1}},\]
where $\tau_j$ is the $j\times j$ principal minor of $b_0\widehat{w}_0$, hence $\tau_j=\Delta_j^n(b_0)$.
\end{proof}

\begin{rem}We note here that $b_0$ is a totally positive element of $B_+$. (A matrix $b \in B_+$ is defined to be {\it totally positive} with respect to $B_+$ if every minor that does not identically vanish on~$B_+$ has a positive value.) The evident relation between principal minors of $b_0\widehat{w}_0$ and upper-right minors of $b_0$ reflects a deeper combinatorial relation studied by Lusztig \cite{bib:bfz, bib:lu}. This enables one to associate a kind of notion of positivity to elements of $\epsilon + \frak{b}_-$ (Hessenberg matrices).
\end{rem}

Using the previous lemma and prior definitions we now proceed to expand the $\tau$-functions in terms of Pl\"ucker coordinates (minors) of $U$ and $U^{-1}$. In this derivation we employ the refinement of the $LU$-decomposition to the $LDU$-decomposition for generic $g \in G$:
\[
g = [g]_- [g]_0 [g]_+,
\]
where $[g]_-\in N_-, [g]_0$ is diagonal and $[g]_+\in N_+$, the subgroup of upper unipotent matrices.

We now proceed to the deconstruction of $\tau_k$ with a brief explanation of the steps given at the end
\begin{align} \label{D1}
\tau_k(t)
&= \big\langle {\rm e}^{tX_0} v^{(k)}, v^{(k)}\big\rangle\\ \label{D2}
&= \big\langle L_0^{-1} U {\rm e}^{t D_\lambda} U^{-1} L_0v^{(k)}, v^{(k)}\big\rangle\\ \label{D3}
&= \big\langle U {\rm e}^{t D_\lambda} U^{-1} [b_0\widehat{w}_0]_- v^{(k)}, v^{(k)}\big\rangle\\ \label{D4}
&= \big\langle U {\rm e}^{t D_\lambda} U^{-1} b_0\widehat{w}_0 [b_0\widehat{w}_0]_+^{-1} [b_0\widehat{w}_0]_0^{-1} v^{(k)}, v^{(k)}\big\rangle\\ \label{D5}
&= \left(\prod_{j=1}^k d_j^{-1}\right) \big\langle U {\rm e}^{t D_\lambda} U^{-1} b_0\widehat{w}_0v^{(k)}, v^{(k)}\big\rangle\\ \label{D6}
&= \left(\prod_{j=1}^k d_j^{-1}\right) \big\langle U {\rm e}^{t D_\lambda} U^{-1} b_0 v_{(k)}, v^{(k)}\big\rangle\\ \label{D7}
&= \left(\prod_{j=1}^k d_j^{-1}\right) \left\langle U {\rm e}^{t D_\lambda} U^{-1} \exp\left(\sum_{i=1}^{n-1}\epsilon_\lambda^i t_i\right) v_{(k)}, v^{(k)}\right\rangle\\ \label{D8}
&= \left(\prod_{j=1}^k d_j^{-1}\right) \left\langle U \exp\left(\sum_{i=1}^{n-1}D_\lambda^i (t_i+\delta_{i1}t)\right) U^{-1} v_{(k)}, v^{(k)}\right\rangle.
\end{align}

In \eqref{D2}, the diagonalization \eqref{LUDiag} was applied to $X_0$ to express it as $X_0 = L_0^{-1} U D_\lambda U^{-1} L_0$. In step~\eqref{D3}, $L_0^{-1}$ is transposed to the right-hand side of the bracket, where it is an upper unipotent matrix acting on the highest weight vector $v^{(k)}$ and therefore has no effect. Similarly, from its definition, $L_0$ may be replaced by $[b_0\widehat{w}_0]_-$ which is then rewritten in~\eqref{D4} using the $LDU$-decomposition. This then contracts to~\eqref{D5} where $d_j$ is the $j$-th entry along the diagonal of $[b_0\widehat{w}_0]_0$. The reverse permutation $\widehat{w}_0$ moves the highest weight vector to the lowest weight vector in~\eqref{D6}. Finally, the last two steps follow from the definition of~$b_0$ and the diagonalization of $\epsilon_\lambda$, respectively.

By a well-known result from linear algebra (see, for example, \cite{bib:strang}), $d_j$ is given by
\[ d_j = \dfrac{\Delta_j^n(b_0)}{\Delta_{j-1}^n(b_0)}.\]
Thus, the factor in front of the inner product is equal to $(\Delta_k^n(b_0))^{-1}$.

By the Cauchy--Binet formula, we get
\begin{align*}
\tau_k(t) &=
\dfrac{1}{\Delta_k^n(b_0)}\sum_{S\in {[n]\choose k}} \det\left(
U \exp\left(\sum_{i=1}^{n-1}D_\lambda^i (t_i+\delta_{i1}t)\right)
\right)_{[k],S}
\det\big(U^{-1}\big)_{S,]k[}\\
&=
\dfrac{1}{\Delta_k^n(b_0)}\sum_{S\in {[n]\choose k}}
\exp\left(\sum_{s\in S} \sum_{i=1}^n \lambda_s^i(t_i+\delta_{i1}t)\right)
\det (U_{[k],S} )
\det\big(U^{-1}_{S,]k[}\big),
\end{align*}
where $]k[ {} :=\{n-k+1,n-k+2,\dots,n\}$.

There are two limits we wish to consider for the $\tau$-functions in what follows. One of these is the limit in which $t$ tends to infinity while the eigenvalues, $\lambda_i$ remain fixed and the other is the tropical limit in the eigenvalues while $t$ is held fixed. In the former case, given the ordering $\lambda_1>\lambda_2>\cdots >\lambda_n$, the term for $S=[k]$ dominates as $t\to +\infty$, and the term for $S={} ]k[$ dominates as $t\to -\infty$. In the latter case, these same terms dominate for $t$ fixed sufficiently positive and, respectively, for $t$ fixed sufficiently negative. The only wrinkle in the latter case is that the minors $\det (U_{[k],S} )$ and
$\det\big(U^{-1}_{S,]k[}\big)$ depend on the eigenvalues. However, the growth of these terms is algebraic and may be crudely estimated to be no worse than $\lambda_1^{2k}$ in one direction and $\lambda_n^{2k}$ in the other. This algebraic growth is beaten by the exponential decay of the time-dependent factors in the tropical limit once the exponential factor associated to $S=[k]$, respectively for $S={}]k[$ has been factored out of the Cauchy--Binet sum. Since in both the $t$ and the tropical limits it is the same single term that dominates, the two limits manifestly commute.

Focusing now on the $t$ limit, we let $\tau_k^+$ and $\tau_k^-$ denote these terms, respectively. By Lemma~\ref{toprightUminor} and Corollary~\ref{toprightUinverseminor}, we have
\[ \tau_k^+ = \dfrac{1}{\Delta_k^n(b_0)} \exp\left(\sum_{s=1}^k \sum_{i=1}^n\lambda_s^i (t_i+\delta_{i1}t)\right)\prod_{j=1}^k \prod_{i=k+1}^{n}\dfrac{1}{\lambda_j-\lambda_i}
+ o(1)\]
and
\[
\tau_k^- = \dfrac{1}{\Delta_k^n(b_0)} \exp\left(\sum_{s=n-k+1}^n \sum_{i=1}^n\lambda_s^i (t_i-\delta_{i1}t)\right)\prod_{i=n-k+1}^n \prod_{j=1}^{n-k}\dfrac{1}{\lambda_j-\lambda_i} + o(1).
\]

As $t\to+\infty$, we therefore have
\[
b_k(-t) \sim b_k^-:=b_k(0)\dfrac{\tau_{k+1}^-\tau_{k-1}^-}{(\tau_k^-)^2},\qquad
b_{n-k}(t) \sim b_{n-k}^+:=b_{n-k}(0)\dfrac{\tau_{n-k+1}^+\tau_{n-k-1}^+}{(\tau_{n-k}^+)^2}.
\]

Plugging in, and writing $\Delta_j$ for $\Delta_j^n(b_0)$, we have
\begin{gather*}
b_k^-
 =b_k(0)\dfrac{(\Delta_k)^2}{\Delta_{k-1}\Delta_{k+1}}\exp\left(\sum_{i=1}^n (\lambda_{n-k}^i-\lambda_{n-k+1}^i)(t_i-\delta_{i1}t)\right)\times \Upsilon_k^-,
\end{gather*}
where
\begin{equation*}\Upsilon_k^-=\dfrac{\left(\prod\limits_{i=n-k}^n \prod\limits_{j=1}^{n-k-1}\dfrac{1}{\lambda_j-\lambda_i}\right)
\left(\prod\limits_{i=n-k+2}^n \prod\limits_{j=1}^{n-k+1}\dfrac{1}{\lambda_j-\lambda_i}\right)}{
\left(\prod\limits_{i=n-k+1}^n \prod\limits_{j=1}^{n-k}\dfrac{1}{\lambda_j-\lambda_i}\right)^2}\end{equation*}
and
\begin{gather*}
b_{n-k}^+
 =b_{n-k}(0)\dfrac{(\Delta_{n-k})^2}{\Delta_{n-k-1}\Delta_{n-k+1}}\exp\left(\sum_{i=1}^n (\lambda_{n-k+1}^i-\lambda_{n-k}^i)(t_i+\delta_{i1}t)\right)\times \Upsilon_{n-k}^+,
\end{gather*}
where
\begin{equation*}
 \Upsilon_{n-k}^+=\dfrac{
\left(\prod\limits_{j=1}^{n-k+1} \prod\limits_{i=n-k+2}^{n}\dfrac{1}{\lambda_j-\lambda_i}\right)
\left(\prod\limits_{j=1}^{n-k-1} \prod\limits_{i=n-k}^{n}\dfrac{1}{\lambda_j-\lambda_i}\right)
}{
\left(\prod\limits_{j=1}^{n-k} \prod\limits_{i=n-k+1}^{n}\dfrac{1}{\lambda_j-\lambda_i}\right)^2
}.
\end{equation*}

Simplifying $\Upsilon_k^-$ and $\Upsilon_{n-k}^+$, one has
\[
\Upsilon_k^-
=
\dfrac{
\left[
\prod\limits_{i<n-k+1}(\lambda_i-\lambda_{n-k+1})
\right]
\left[
\prod\limits_{i>n-k}(\lambda_{n-k}-\lambda_i)
\right]
}{
\left[
\prod\limits_{i<n-k}(\lambda_i-\lambda_{n-k})
\right]
\left[
\prod\limits_{i>n-k+1}(\lambda_{n-k+1}-\lambda_i)
\right]
},
\]
and
\[
\Upsilon_{n-k}^+
=
\dfrac{
\left[
\prod\limits_{i>n-k}(\lambda_{n-k}-\lambda_i)
\right]
\left[
\prod\limits_{i<n-k+1}(\lambda_i-\lambda_{n-k+1})
\right]
}{
\left[
\prod\limits_{i>n-k+1}(\lambda_{n-k+1}-\lambda_i)
\right]
\left[
\prod\limits_{i<n-k}(\lambda_i-\lambda_{n-k})
\right]}.
\]

We see that $\Upsilon_k^-=\Upsilon_{n-k}^+$. Thus, the product $\Upsilon_k^-\Upsilon_{n-k}^+$ simplifies to
\begin{align*}
\Upsilon_k^-\Upsilon_{n-k}^+
&=
\dfrac{
\left[
\prod\limits_{i>n-k}(\lambda_{n-k}-\lambda_i)^2
\right]
\left[
\prod\limits_{i<n-k+1}(\lambda_{n-k+1}-\lambda_i)^2
\right]
}{
\left[
\prod\limits_{i<n-k}(\lambda_{n-k}-\lambda_i)^2
\right]
\left[
\prod\limits_{i>n-k+1}(\lambda_{n-k+1}-\lambda_i)^2
\right]
}
\\
&=
\dfrac{
\left[
\prod\limits_{j<k+1}(\lambda_{n-k}-\lambda_{n-j+1})^2
\right]
\left[
\prod\limits_{j>k}(\lambda_{n-k+1}-\lambda_{n-j+1})^2
\right]
}{
\left[
\prod\limits_{j>k+1}(\lambda_{n-k}-\lambda_{n-j+1})^2
\right]
\left[
\prod\limits_{j<k}(\lambda_{n-k+1}-\lambda_{n-j+1})^2
\right]
}.
\end{align*}

We also observe that
\[ \exp\left(\sum_{i=1}^n \big(\lambda_{n-k}^i-\lambda_{n-k+1}^i\big)(t_i-\delta_{i1}t)\right)\exp\left(\sum_{i=1}^n \big(\lambda_{n-k+1}^i-\lambda_{n-k}^i\big)(t_i+\delta_{i1}t)\right)
\]
is equal to
\begin{align*}
\exp\left(\sum_{i=1}^n\big(\lambda_{n-k}^i-\lambda_{n-k+1}^i\big)([t_i-\delta_{1i}]-[t_i+\delta_{i1}t])\right)
&=\exp\left(\sum_{i=1}^n-2\delta_{i1}t\big(\lambda_{n-k}^i-\lambda_{n-k+1}^i\big)\right)\\
&= {\rm e}^{2(\lambda_{n-k+1}-\lambda_{n-k})t}.
\end{align*}

Putting all of the above together, one obtains
%
\begin{align*}
&b_k^-b_{n-k}^+{\rm e}^{2(\lambda_{n-k}-\lambda_{n-k+1})t}
= \dfrac{b_k(0)b_{n-k}(0)(\Delta_k\Delta_{n-k})^2}{\Delta_{k-1}\Delta_{k+1}\Delta_{n-k-1}\Delta_{n-k+1}}\\
& \qquad\quad{}\times \dfrac{
\left[
\prod\limits_{j<k+1}(\lambda_{n-k}-\lambda_{n-j+1})^2
\right]
\left[
\prod\limits_{j>k}(\lambda_{n-k+1}-\lambda_{n-j+1})^2
\right]
}{
\left[
\prod\limits_{j>k+1}(\lambda_{n-k}-\lambda_{n-j+1})^2
\right]
\left[
\prod\limits_{j<k}(\lambda_{n-k+1}-\lambda_{n-j+1})^2
\right]
}\\
&\qquad{}= \dfrac{
\left[
\prod\limits_{j<k+1}(\lambda_{n-k}-\lambda_{n-j+1})^2
\right]
\left[
\prod\limits_{j>k}(\lambda_{n-k+1}-\lambda_{n-j+1})^2
\right]
}{
\left[
\prod\limits_{j>k+1}(\lambda_{n-k}-\lambda_{n-j+1})^2
\right]
\left[
\prod\limits_{j<k}(\lambda_{n-k+1}-\lambda_{n-j+1})^2
\right]
},
\end{align*}
where the second equality holds by Lemma~\ref{bzerolemma}, which completes the proof of Theorem~\ref{classicaltodaphase}.

\subsection{The discrete-time Toda lattice phase shift (dToda)} \label{dTODA}
We now turn to proving the analogue of equation~\eqref{multctstimephshift} for the discrete-time Toda lattice. In doing this we again emphasize that $t$ will be discrete ($t \in \mathbb{Z}$) and that the eigenvalues of $X_0$ are taken to be positive and distinct.

\begin{thm}\label{discretetodaphase}
\begin{gather}
\lim_{t\to +\infty}b_{n-k}(t)b_k(-t) \left(\dfrac{\lambda_{n-k}}{\lambda_{n-k+1}}\right)^{2t}
= \dfrac{\prod\limits_{j<k+1} (\lambda_{n-k}-\lambda_{n-j+1} )^2\prod\limits_{j>k} (\lambda_{n-k+1}-\lambda_{n-j+1} )^2}
{\prod\limits_{j>k+1} (\lambda_{n-k}-\lambda_{n-j+1} )^2\prod\limits_{j<k} (\lambda_{n-k+1}-\lambda_{n-j+1} )^2}.\!\!\!\label{dTodaphaseshift}
\end{gather}
\end{thm}

This essentially comes down to the analogous description of the subdiagonal entries in this setting:
\[ b_k(t) = \dfrac{\Delta_{k+1}^{k+1}\big(X_0^t\big)\Delta_{k-1}^{k-1}\big(X_0^t\big)}{\big(\Delta_{k}^{k}\big(X_0^t\big)\big)^2}b_k(0).\]

In the above, we recall that $\Delta_k^k$ is the $k\x k$ principal minor determinant of a matrix. When one takes this into our proof of Moser's continuous-time setting, the only difference is that the~${\rm e}^{t\epsilon_\lambda}$ is replaced with~$\epsilon_\lambda^t$.

When we compute these principal minor determinants, we still use~$U$ to diagonalise, we still have $[b_0\widehat{w}_0]_-$, everything is the same, except that we have $D_\lambda^t$ instead of ${\rm e}^{D_\lambda t}$.

\section{BBS phase shift}\label{bbsphasesec}
We use the basic properties of tropicalisation to tropicalise equation~\eqref{dTodaphaseshift}:
\begin{equation*}
\lim_{t\to +\infty}b_{n-k}(t)b_k(-t) \left(\dfrac{\lambda_{n-k}}{\lambda_{n-k+1}}\right)^{2t}
= \dfrac{\prod\limits_{j<k+1} (\lambda_{n-k}-\lambda_{n-j+1} )^2\prod\limits_{j>k} (\lambda_{n-k+1}-\lambda_{n-j+1} )^2}
{\prod\limits_{j>k+1} (\lambda_{n-k}-\lambda_{n-j+1} )^2\prod\limits_{j<k} (\lambda_{n-k+1}-\lambda_{n-j+1} )^2}.
\end{equation*}

We first recall that the phase space of dToda consists of tridiagonal Hessenberg matrices and we write a solution as
\begin{equation}
X(t) = \left[\begin{matrix}
a_1(t) & 1\\
b_1(t) & a_2(t) & \ddots\\
&\ddots & \ddots & 1\\
&&b_{n-1}(t) & a_n(t)
\end{matrix}\right].\label{tridiaghessenb}
\end{equation}
Thus, the $b$'s are the subdiagonal entries of a solution at a given time.

It is important to recall that Tokihiro's \cite{bib:tokihiro} tropicalisation is performed on the variables in the lower-upper factorisation description of dToda which we briefly recall from Sections~\ref{subsecsymesdnt} and~\ref{sec:udtoda}:
\[ R(t) = \left[\begin{matrix} I_1^t & 1\\ &I_2^t & \ddots \\ &&\ddots & 1\\ &&&I_n^t\end{matrix}\right],
\qquad
L(t) = \left[\begin{matrix} 1 \\ V_1^t & 1\\ &\ddots&\ddots\\ &&V_{n-1}^t & 1 \end{matrix}\right].
\]
This evolves under the factorisation dynamics
\[L(t+1)R(t+1) = R(t) L(t).\]

By setting
$X(t) = L(t)R(t)$,
one obtains a matrix of the form in~\eqref{tridiaghessenb}. A quick comparison reveals
$b_k(t) = V_k^tI_k^t$
for all $k\in\{1,2,\dots,n-1\}$.

By Tokihiro's calculation, $V_k^t$ corresponds to the $k$-th finite gap and~$I_k^t$ corresponds to the $k$-th block size, each at time~$t$. Therefore, the tropical analogue of $b_k(t)$ is
$W_k^t+Q_k^t$.

The other components in equation~\eqref{dTodaphaseshift} to deal with are the eigenvalues. In the limit as $t\to+\infty$, we know that the matrix in~\eqref{tridiaghessenb} is
\begin{equation*}
X(t) = \left[\begin{matrix}
\lambda_1 & 1\\
&\lambda_2 & \ddots \\
&&\ddots & 1\\
&&&\lambda_n
\end{matrix}\right].
\end{equation*}

We also know that $L(t)$ limits to the identity matrix, so we know that $\lim\limits_{t \to +\infty}I_k^t=\lambda_k$ for each~$k$. Since the tropicalisation of $I_k^t$ is $Q_k^t$ (the $k$-th block size at time~$t$), the tropical analogue of~$\lambda_k$ is
\[\lim_{t\to+\infty}Q_k^t = \mu_k,\]
where
$\mu_1<\mu_2<\cdots <\mu_n$
are the sorted asymptotic block sizes of the BBS in forwards time.

Recalling that $\lambda_1>\lambda_2>\cdots>\lambda_n > 0$, we compute the tropicalisation of $(\lambda_i-\lambda_j)^2$ for each pair $i\neq j$. For the sake of this calculation, assume $i<j$ so that $\lambda_i-\lambda_j>0$. To detropicalise this quantity, we let $\lambda_i={\rm e}^{-\frac{1}{\hbar}\mu_i}$ and $\lambda_j={\rm e}^{-\frac{1}{\hbar}\mu_j}$ and compute
\begin{align*}-\lim_{\hbar\to 0^+} \hbar\log\big(
\big(
{\rm e}^{-\frac{1}{\hbar}\mu_i}-{\rm e}^{-\frac{1}{\hbar}\mu_j}
\big)^2
\big)
=
-2\lim_{\hbar\to 0^+} \hbar\log\big(
{\rm e}^{-\frac{1}{\hbar}\mu_i}-{\rm e}^{-\frac{1}{\hbar}\mu_j}
\big)
=
2\min(\mu_i,\mu_j)= 2\mu_i,
\end{align*}
where the last equality follows from a computation using l'H\^{o}pital's rule.

Using the property that tropicalisation sends products to sums and quotients to differences, the tropical analogue of equation~\eqref{dTodaphaseshift} is given by

\begin{thm}\label{maintheorembbsphaseshiftlim}
\begin{gather*}
 \lim_{t\to +\infty}\big(
W_{n-k}^t +Q_{n-k}^t +W_{k}^{-t} +Q_{k}^{-t} + 2t(\mu_{n-k}-\mu_{n-k+1})
\big)\nonumber\\
\qquad{} =\sum_{j<k+1}2\mu_{n-k}+\sum_{j>k}2\mu_{n-j+1}-\sum_{j>k+1}2\mu_{n-j+1} -\sum_{j<k}2\mu_{n-k+1}.
\end{gather*}
\end{thm}

\begin{rem}
Once the blocks of balls have become ordered according to their asymptotic lengths, which occurs in a finite time frame, all blocks travel freely with speeds given by their lengths. This means that $Q_k^t=\mu_k$ for each~$k$ and for all time after the sorting. Similarly, in backwards time, the blocks will become sorted in a finite time frame from largest to smallest, so that $Q_k^{-t}=\mu_{n-k+1}$ for each~$k$ and for all sufficiently large time. Taking $t$ large enough so that the time~$t$ and time $-t$ states are both ordered appropriately, we see that
\begin{gather*}
 W_{n-k}^{t+1} = W_{n-k}^t+(\mu_{n-k+1}-\mu_{n-k}),\\
 Q_{n-k}^{t+1} = Q_{n-k}^t=\mu_{n-k},\\
 W_{k}^{-(t+1)} = W_{k}^{-t}+(\mu_{n-k+1}-\mu_{n-k}),\\
 Q_{k}^{-(t+1)} = Q_{k}^{-t}=\mu_{n-k+1}.
\end{gather*}
Consequently,
\begin{gather*}
 W_{n-k}^{t+1} +Q_{n-k}^{t+1} +W_{k}^{-(t+1)} +Q_{k}^{-(t+1)} + 2(t+1)(\mu_{n-k}-\mu_{n-k+1})\\
\qquad{} = W_{n-k}^t+(\mu_{n-k+1}-\mu_{n-k}) +Q_{n-k}^t +W_{k}^{-t} +(\mu_{n-k+1}-\mu_{n-k})+Q_{k}^{-t} \\
\qquad\quad{} + 2(t+1)(\mu_{n-k}-\mu_{n-k+1})\\
\qquad{} =W_{n-k}^t +Q_{n-k}^t +W_{k}^{-t} +Q_{k}^{-t} + 2t(\mu_{n-k}-\mu_{n-k+1}).
\end{gather*}
Therefore, the quantity in the limit described in Theorem~\ref{maintheorembbsphaseshiftlim} becomes constant in finite time.
\end{rem}

Because the limit stabilises in finite time, Theorem~\ref{maintheorembbsphaseshiftlim} is equivalent to the statement that for all sufficiently large $t$, the following holds:
\begin{gather*}
 W_{n-k}^t + \mu_{n-k} +W_{k}^{-t} + \mu_{n-k+1} + 2t(\mu_{n-k}-\mu_{n-k+1})\nonumber\\
\qquad{} =\sum_{j<k+1}2\mu_{n-k}+\sum_{j>k}2\mu_{n-j+1}-\sum_{j>k+1}2\mu_{n-j+1} -\sum_{j<k}2\mu_{n-k+1}.
\end{gather*}

Or, equivalently
\begin{gather}
 W_{n-k}^t + W_{k}^{-t} + 2t(\mu_{n-k}-\mu_{n-k+1})\nonumber\\
 \quad{}=\sum_{j<k+1}2\mu_{n-k}+\sum_{j>k}2\mu_{n-j+1}-\sum_{j>k+1}2\mu_{n-j+1}-\sum_{j<k}2\mu_{n-k+1} - (\mu_{n-k}+\mu_{n-k+1}).\label{actualphaseshiftfromdtoda}
\end{gather}

\subsection{Free dynamics versus the phase shift}
To see the role phase shifts play in the tropical setting, we will initialise an $n$-soliton box-ball system sufficiently far back in time so that the blocks are arranged from largest to smallest.

We consider a special dynamical evolution of this state gotten by treating each block as a~single, immutable object, travelling with velocity given by its size. If we allow multiple balls to occupy the same box at any given time, the blocks will eventually pass through one another without experiencing a phase shift and order themselves from smallest to largest. The point here is that the intermediate stages may not correspond to box-ball states, but, in sufficiently large backwards and forwards time, boxes will contain at most one ball, resulting in a box-ball state.\looseness=-1

By definition of the phase-shift, the resulting far enough forward time state for this special dynamics will differ from the corresponding (same number of time-steps) box-ball system by the phase shift of each block.

Let
$\mu_1<\mu_2<\cdots<\mu_n$
be the ordered block sizes after sufficiently many time evolutions. In the far past the block sizes are initially ordered as
\[
Q_1^{-t} = \mu_n<Q_2^{-t} = \mu_{n-1} < \cdots <Q_n^{-t} =\mu_1.
\]

We temporarily introduce some variables by indexing all of the boxes in the system by the integers~$\Z$ and by letting~$s_k$ be equal to the index of the first ball of the $\mu_k$ block at time~$-t$.

From the initial first ball positions and the blocks sizes, one can determine the values of $W_k^{-t}$ for each $k$:
\[ W_k^{-t} = s_{n-k}-s_{n-k+1}-\mu_{n-k+1}.\]

In the figure below, we depict how the $s_k$'s are set up for this initialised backward time configuration:
\begin{figure}[h!]
\centering
\tikz[scale=0.52]{
\foreach \x in {0,4,6,10,11,13,14,18}
{\draw[fill=white] (\x,3) -- (\x+1,3) -- (\x+1,4) -- (\x,4) -- cycle;			
}
\foreach \x in {1,3,7,9,15,17}
{\draw[fill=white] (\x,3) -- (\x+1,3) -- (\x+1,4) -- (\x,4) -- cycle;			
\fill[cyan] (\x+0.5,3.5) circle (0.25);
}
\foreach \x in {2,5,8,12,16}
{\draw[fill=white] (\x,3) -- (\x+1,3) -- (\x+1,4) -- (\x,4) -- cycle;			
\node at (\x+0.5,3.5) {$\cdots$};
}
\foreach \x in {19}
{\draw[fill=white,white] (\x,3) -- (\x+2,3) -- (\x+2,4) -- (\x,4) -- cycle;
\draw[-] (\x,3) -- (\x,4);
\draw[-] (\x,3) -- (\x+2,3);
\draw[-] (\x,4) -- (\x+2,4);
\draw[-] (\x+1,3) -- (\x+1,4);
\node at (\x+1.5,3.5) {$\cdots$};
}
\foreach \x in {0}
{\draw[fill=white,white] (\x,3) -- (\x-2,3) -- (\x-2,4) -- (\x,4) -- cycle;
\draw[-] (\x,3) -- (\x,4);
\draw[-] (\x,3) -- (\x-2,3);
\draw[-] (\x,4) -- (\x-2,4);
\draw[-] (\x-1,3) -- (\x-1,4);
\node at (\x-1.5,3.5) {$\cdots$};
}
\draw [decorate,decoration={brace,amplitude=4pt}] (0.9,2.85) -- (-1.9,2.85) node [black,midway,yshift=-0.4cm] {\scriptsize{$W_0^{-t}=\infty$}};
\draw [decorate,decoration={brace,amplitude=4pt}] (3.9,2.85) -- (1.1,2.85) node [black,midway,yshift=-0.4cm] {\scriptsize{$Q_1^{-t}=\mu_n$}};
\draw [decorate,decoration={brace,amplitude=4pt}] (6.9,2.85) -- (4.1,2.85) node [black,midway,yshift=-0.4cm] {\scriptsize{$W_1^{-t}$}};
\draw [decorate,decoration={brace,amplitude=4pt}] (9.9,2.85) -- (7.1,2.85) node [black,midway,yshift=-0.4cm] {\scriptsize{$Q_2^{-t}=\mu_{n-1}$}};
\draw [decorate,decoration={brace,amplitude=4pt}] (17.9,2.85) -- (15.1,2.85) node [black,midway,yshift=-0.4cm] {\scriptsize{$Q_n^{-t}=\mu_{1}$}};
\draw [decorate,decoration={brace,amplitude=4pt}] (20.9,2.85) -- (18.1,2.85) node [black,midway,yshift=-0.4cm] {\scriptsize{$W_n^{-t}=\infty$}};
\node at (1.5,4.3) {\small{$s_n$}};
\node at (7.5,4.3) {\small{$s_{n-1}$}};
\node at (15.5,4.3) {\small{$s_1$}};
}
\end{figure}

Recall how this compares to the large positive $t$ picture, in which the blocks become arranged in the opposite ordering:

\begin{figure}[h!]
\centering
\tikz[scale=0.52]{
\foreach \x in {0,4,6,10,11,13,14,18}
{\draw[fill=white] (\x,3) -- (\x+1,3) -- (\x+1,4) -- (\x,4) -- cycle;			
}
\foreach \x in {1,3,7,9,15,17}
{\draw[fill=white] (\x,3) -- (\x+1,3) -- (\x+1,4) -- (\x,4) -- cycle;			
\fill[cyan] (\x+0.5,3.5) circle (0.25);
}
\foreach \x in {2,5,8,12,16}
{\draw[fill=white] (\x,3) -- (\x+1,3) -- (\x+1,4) -- (\x,4) -- cycle;			
\node at (\x+0.5,3.5) {$\cdots$};
}
\foreach \x in {19}
{\draw[fill=white,white] (\x,3) -- (\x+2,3) -- (\x+2,4) -- (\x,4) -- cycle;
\draw[-] (\x,3) -- (\x,4);
\draw[-] (\x,3) -- (\x+2,3);
\draw[-] (\x,4) -- (\x+2,4);
\draw[-] (\x+1,3) -- (\x+1,4);
\node at (\x+1.5,3.5) {$\cdots$};
}
\foreach \x in {0}
{\draw[fill=white,white] (\x,3) -- (\x-2,3) -- (\x-2,4) -- (\x,4) -- cycle;
\draw[-] (\x,3) -- (\x,4);
\draw[-] (\x,3) -- (\x-2,3);
\draw[-] (\x,4) -- (\x-2,4);
\draw[-] (\x-1,3) -- (\x-1,4);
\node at (\x-1.5,3.5) {$\cdots$};
}
\draw [decorate,decoration={brace,amplitude=4pt}] (0.9,2.85) -- (-1.9,2.85) node [black,midway,yshift=-0.4cm] {\scriptsize{$W_0^{t}=\infty$}};
\draw [decorate,decoration={brace,amplitude=4pt}] (3.9,2.85) -- (1.1,2.85) node [black,midway,yshift=-0.4cm] {\scriptsize{$Q_1^{t}=\mu_1$}};
\draw [decorate,decoration={brace,amplitude=4pt}] (6.9,2.85) -- (4.1,2.85) node [black,midway,yshift=-0.4cm] {\scriptsize{$W_1^{t}$}};
\draw [decorate,decoration={brace,amplitude=4pt}] (9.9,2.85) -- (7.1,2.85) node [black,midway,yshift=-0.4cm] {\scriptsize{$Q_2^{t}=\mu_{2}$}};
\draw [decorate,decoration={brace,amplitude=4pt}] (17.9,2.85) -- (15.1,2.85) node [black,midway,yshift=-0.4cm] {\scriptsize{$Q_n^{t}=\mu_{n}$}};
\draw [decorate,decoration={brace,amplitude=4pt}] (20.9,2.85) -- (18.1,2.85) node [black,midway,yshift=-0.4cm] {\scriptsize{$W_n^{t}=\infty$}};
}
\end{figure}

The following can be done without this introduction, but it does make the calculations a~little simpler to follow. Under this free dynamics, the $\mu_k$-block travels with velocity~$\mu_k$. Thus, if~$j$ is sufficiently large, the blocks will be ordered with the first ball of the $\mu_k$-block occupying position
$ s_k+j\mu_k$.
If we let $\Phi_k$ denote the phase shift experienced in the true BBS dynamics by the $k$-block, one has that the position of the first ball of the $\mu_k$ block is given by
$s_k+j\mu_k + \Phi_k$.
As before, and for $j$ sufficiently large, we can use these positions and the block sizes to determine the values $W_k^{-t+j}$ for each $k$:
\begin{align*}
 W_k^{-t+j} &= ( s_{k+1}+j\mu_{k+1} + \Phi_{k+1}) - (s_k+j\mu_k + \Phi_k) - \mu_k\\
 &= s_{k+1}-s_k +j(\mu_{k+1}-\mu_k) + \Phi_{k+1}-\Phi_k -\mu_k.
\end{align*}
Assuming that $t$ was large enough for both the time $-t$ and time $t$ states to be ordered appropriately, substituting in $j=2t$ and looking at the above for $n-k$, one has
\[ W_{n-k}^{t} = s_{n-k+1}-s_{n-k} +2t(\mu_{n-k+1}-\mu_{n-k}) + \Phi_{n-k+1}-\Phi_{n-k} -\mu_{n-k}.\]
Thus, for $t$ sufficiently large, one has
\begin{equation}
W_{n-k}^t + W_k^{-t} + 2t(\mu_{n-k}+\mu_{n-k+1}) = \Phi_{n-k+1}-\Phi_{n-k}-(\mu_{n-k}+\mu_{n-k+1})\label{phaseshiftinws}
\end{equation}

Comparing equations~\eqref{actualphaseshiftfromdtoda} and~\eqref{phaseshiftinws}, we now see that we have
\begin{equation}
\Phi_{n-k+1}-\Phi_{n-k} = \sum_{j<k+1}2\mu_{n-k}+\sum_{j>k}2\mu_{n-j+1}-\sum_{j>k+1}2\mu_{n-j+1} -\sum_{j<k}2\mu_{n-k+1}\label{telescopingphaseshifts}
\end{equation}
for $k=1, 2, \dots, n-1$. This gives us $n-1$ equations for the $n$ phase shifts, $\Phi_i$ for $i=1,\dots,n$.

\subsection{Conservation laws}
\begin{lem}\label{lemma:sumsitesgrowsbyM}
Suppose one initialises an $n$-soliton state at time $t=0$ with coordinates
\[ \big(\infty,Q_1^0,W_1^0,Q_2^0,W_2^0,\dots,Q_n^0,\infty\big).\]
Define $\ell_i^0$ for $i\in\{1,\dots,n\}$ as follows
\[ \ell_i^0 = Q_i^0 +\sum_{j=1}^{i-1}\big(Q_j^0+W_j^0\big).\]
If at time $t=1$, the system has coordinates
\[\big(\infty,Q_1^1,W_1^1,Q_2^1,W_2^1,\dots,Q_n^1,\infty\big),\]
define $\ell_i^1$ for $i\in\{1,\dots,n\}$ as follows
\[ \ell_i^1 = Q_1^0 + Q_i^1 +\sum_{j=1}^{i-1}\big(Q_j^1+W_j^1\big).\]
Then one has
\begin{equation}
\sum_{i=1}^n \ell_i^1 = M + \sum_{i=1}^n \ell_i^0,\label{sumlastpositionsM}
\end{equation}
where $M$ is the total mass of the system, i.e.,
\[M = \sum_{i=1}^n \mu_i = \sum_{i=1}^n Q_i^t,\]
where the latter equality is independent of the choice of $t$.
\end{lem}
\begin{proof}
Recall the evolution rule for the $W$ coordinates from the box-ball coordinate dy\-namics~(\ref{Witplusoneeqnbbs}):
\[ W_j^{t+1} = Q_{j+1}^t + W_j^t - Q_j^{t+1}\]
for each $i\in\{1,\dots,n\}$. Rearranging this and taking $t=0$, one has
\[W_j^1+Q_j^1 = Q_{j+1}^0+W_j^0.\]
Therefore,
\begin{align*}
\sum_{i=1}^n \ell_i^1
&= \sum_{i=1}^n \left(Q_1^0 + Q_i^1 +\sum_{j=1}^{i-1}\big(Q_j^1+W_j^1\big)\right)
 = \sum_{i=1}^n Q_i^1 + \sum_{i=1}^n \left(Q_1^0 + \sum_{j=1}^{i-1}\big(Q_{j+1}^0+W_j^0\big)\right)\\
&= M + \sum_{i=1}^n \left(Q_i^0 + \sum_{j=1}^{i-1}\big(Q_j^0+W_j^0\big)\right)
 = M + \sum_{i=1}^n \ell_i^0.\tag*{\qed}
\end{align*}\renewcommand{\qed}{}
\end{proof}

\begin{rem}\label{constgrowsthtotoalmasssum}
The significance of the $\ell_i$ quantities is realised by labelling all boxes of the initial box-ball state by $\Z$ such that the left-most ball in the time $t=0$ state is in box 1. In this configuration, the right-most ball of the first block occupies box $\ell_1^0 = Q_1^0$. The right-most ball of the second block occupies box $\ell_2^0 = Q_1^0+W_1^0 + Q_2^0$. In general, the right-most ball of the $i$-th block occupies the $\ell_i^0$-th box. So, the quantity $\sum\limits_{i=1}^n \ell_i^0$ records the sum of the labels of the right-most balls of the block.

We provide the following two-time step illustration of the $\ell_i^t$ assignment for an particular BBS configuration:

\begin{figure}[h!]
\centering
\tikz[scale=0.57]{
\foreach \y in {0}{
\node at (-2.7,\y+0.5) {$t=0:$};
\foreach \x in {1,2,3,4,5,6,7,8,9,10,11,12,13,14,15}{
\node at (\x+0.5,\y+1.25) {\small{\x}};
}
\node at (3+0.5,\y-0.4){\small{$\ell_1^0$}};
\node at (7+0.5,\y-0.4){\small{$\ell_2^0$}};
\node at (10+0.5,\y-0.4){\small{$\ell_3^0$}};
\foreach \x in {4,5,8,11,12,13,14}
{\draw[fill=white] (\x,\y) -- (\x+1,\y) -- (\x+1,\y+1) -- (\x,\y+1) -- cycle;			
}
\foreach \x in {1,2,3,6,7,9,10}
{\draw[fill=white] (\x,\y) -- (\x+1,\y) -- (\x+1,\y+1) -- (\x,\y+1) -- cycle;			
\fill[cyan] (\x+0.5,\y+0.5) circle (0.25);
}
\foreach \x in {15}
{\draw[fill=white,white] (\x,\y) -- (\x+2,\y) -- (\x+2,\y+1) -- (\x,\y+1) -- cycle;
\draw[-] (\x,\y) -- (\x,\y+1);
\draw[-] (\x,\y) -- (\x+2,\y);
\draw[-] (\x,\y+1) -- (\x+2,\y+1);
\draw[-] (\x+1,\y) -- (\x+1,\y+1);
\node at (\x+1.5,\y+0.5) {$\cdots$};
}
\foreach \x in {1}
{\draw[fill=white,white] (\x,\y) -- (\x-2,\y) -- (\x-2,\y+1) -- (\x,\y+1) -- cycle;
\draw[-] (\x,\y) -- (\x,\y+1);
\draw[-] (\x,\y) -- (\x-2,\y);
\draw[-] (\x,\y+1) -- (\x-2,\y+1);
\draw[-] (\x-1,\y) -- (\x-1,\y+1);
\node at (\x-1.5,\y+0.5) {$\cdots$};
}
}
\foreach \y in {-2.7}{
\node at (-2.7,\y+0.5) {$t=1:$};
\foreach \x in {1,2,3,4,5,6,7,8,9,10,11,12,13,14,15}{
\node at (\x+0.5,\y+1.25) {\small{\x}};
}
\foreach \x in {1,2,3,6,7,9,10}
{\draw[fill=white] (\x,\y) -- (\x+1,\y) -- (\x+1,\y+1) -- (\x,\y+1) -- cycle;			
}
\foreach \x in {4,5,8,11,12,13,14}
{\draw[fill=white] (\x,\y) -- (\x+1,\y) -- (\x+1,\y+1) -- (\x,\y+1) -- cycle;			
\fill[cyan] (\x+0.5,\y+0.5) circle (0.25);
}
\foreach \x in {15}
{\draw[fill=white,white] (\x,\y) -- (\x+2,\y) -- (\x+2,\y+1) -- (\x,\y+1) -- cycle;
\draw[-] (\x,\y) -- (\x,\y+1);
\draw[-] (\x,\y) -- (\x+2,\y);
\draw[-] (\x,\y+1) -- (\x+2,\y+1);
\draw[-] (\x+1,\y) -- (\x+1,\y+1);
\node at (\x+1.5,\y+0.5) {$\cdots$};
}
\foreach \x in {1}
{\draw[fill=white,white] (\x,\y) -- (\x-2,\y) -- (\x-2,\y+1) -- (\x,\y+1) -- cycle;
\draw[-] (\x,\y) -- (\x,\y+1);
\draw[-] (\x,\y) -- (\x-2,\y);
\draw[-] (\x,\y+1) -- (\x-2,\y+1);
\draw[-] (\x-1,\y) -- (\x-1,\y+1);
\node at (\x-1.5,\y+0.5) {$\cdots$};
}
\node at (5+0.5,\y-0.4){\small{$\ell_1^1$}};
\node at (8+0.5,\y-0.4){\small{$\ell_2^1$}};
\node at (14+0.5,\y-0.4){\small{$\ell_3^1$}};
}
}
\end{figure}

Observe not only what the $\ell_i$'s are, but also how the sum of the $\ell_i$'s increases by $M=7$ from $t=0$ to $t=1$.

By keeping the above labelling and evolving the dynamics to time $t=1$, the same argument for the labels of the right-most balls of the blocks carries through, but everything is shifted to the right by~$Q_1^0$ since the left-most ball of the system now occupies box~$Q_1^0$. Therefore, equation~\eqref{sumlastpositionsM} says that the sum of the positions of the last ball of each block increases by the total mass of the system from time $t=0$ to time $t=1$. Of course, one could initialise at any time and the total mass is independent of time, so this extends to the statement that the sum of the positions of the last ball of each block at any time increases by $M\cdot j$ after~$j$ time steps, given an initial indexing of the boxes.
\end{rem}

\begin{rem}
We note that this ultradiscrete conservation law corresponds to the conservation at both the continuous and discrete time levels described in Section \ref{sec:centreofmass}.
\end{rem}

\begin{lem}\label{sumphaseshiftsiszeo}
The sum of all of the phase shifts is zero $\sum_{i=1}^n \Phi_i = 0$.
\end{lem}
\begin{proof}
Suppose $t$ is taken large enough so that at time $-t$ the BBS has achieved its asymptotic block structure with blocks ordered from largest to smallest, and at time $t$ the BBS has achieved its asymptotic block structure with blocks ordered from smallest to largest.

At time $-t$, suppose the boxes of the BBS are indexed by $\Z$, with $\ell_i$ denoting the label of the last ball of the $\mu_i$ block. Until the collision, the label of the last ball in the $\mu_i$-block increases by $\mu_i$ at each stage. Therefore, the sum of the labels of each last ball of the blocks increases by
$\sum_{i=1}^n \mu_i = M$
at each time step. If we let $\hat{\ell}_i$ denote the label of the last ball in the $\mu_i$-block at time $t$ in this phase shift-free situation, then it would follow that
$\hat{\ell}_i = \ell_i + 2t\mu_i$.
Thus,
\begin{equation}\sum_{i=1}^n \hat{\ell}_i = 2tM + \sum_{i=1}^n\ell_i.\label{phasefreesumposition}\end{equation}
By definition of the phase shifts, the true label for the last ball of the $\mu_i$ block is given by
$\hat{\ell}_i +\Phi_i$.
By Remark \ref{constgrowsthtotoalmasssum}, these labels satisfy
\begin{equation}\sum_{i=1}^n \big(\hat{\ell}_i +\Phi_i\big) = 2tM + \sum_{i=1}^n \ell_i.\label{phaseshiftpositionsmass}\end{equation}
Comparing \eqref{phasefreesumposition} and \eqref{phaseshiftpositionsmass}, we see that
$\sum_{i=1}^n \Phi_i = 0$.
\end{proof}

As a result of equation \eqref{telescopingphaseshifts} and Lemma~\ref{sumphaseshiftsiszeo}, we now have that the $n$ phase shifts satisfy the following system of $n$ linear equations:
\begin{equation}
\begin{cases}
\displaystyle \Phi_{n-k+1}-\Phi_{n-k} = \sum_{j<k+1}2\mu_{n-k}+\sum_{j>k}2\mu_{n-j+1}-\sum_{j>k+1}2\mu_{n-j+1} -\sum_{j<k}2\mu_{n-k+1},\\
\displaystyle \sum\limits_{i=1}^n \Phi_i = 0,
\end{cases}\label{systemofequationsforphaseshifts}
\end{equation}
where the first equation is for each $k\in \{1,\dots,n-1\}$.

\begin{lem}
Take the $n\x n$ matrix
\[\left[\begin{matrix}
-1 & 1 \\
&-1&1\\
&&\ddots & \ddots \\
&&&-1&1\\
1&1&\cdots & 1 & 1
\end{matrix}\right],
\]
which nest from the bottom-right $($so that for $n=1$, this is just~$[1])$. This matrix is non-singular.
\end{lem}
This is a simple exercise in linear algebra (one can show that this matrix has determinant $n(-1)^{n+1}$. The key point is that it now follows that the system of equations in~\eqref{systemofequationsforphaseshifts} has a unique solution, which is given by the phase shifts. We can now state our main theorem:

\begin{thm}
For $i\in\{1,\dots,n\}$, define $\hat{\Phi}_i$ via
\begin{equation}\label{equation:phaseshiftforbbs}
 \hat{\Phi}_i = \sum_{j<i} 2\mu_j - \sum_{j>i}2\mu_i.
\end{equation}For a box-ball system with asymptotic block sizes given by
$\mu_1<\mu_2<\cdots <\mu_n$,
the phase shift, $\Phi_i$, of the $\mu_i$-block experienced from $-t$ to $t$ for $t$ sufficiently large is equal to $\hat{\Phi}_i$.
\end{thm}
\begin{proof}
For each $k=1,\dots,n-1$, one has
\begin{align*}
\hat{\Phi}_{n-k+1} - \hat{\Phi}_{n-k}
&= \sum_{j<n-k+1} 2\mu_j - \sum_{j>n-k+1}2\mu_{n-k+1} - \sum_{j<n-k} 2\mu_j + \sum_{j>n-k}2\mu_{n-k}\\
& = \sum_{j>k}2\mu_{n-j+1}-\sum_{j<k}2\mu_{n-k+1}-\sum_{j>k+1}2\mu_{n-j+1} +\sum_{j<k+1}2\mu_{n-k}.
\end{align*}
Thus, $(\hat{\Phi}_i)_{i=1}^n$ satisfies the first $n-1$ equations of \eqref{systemofequationsforphaseshifts}.

Next, let us compute $\sum\limits_{k=1}^n \Phi_k$ by direct calculation:
\begin{align*}
\sum_{k=1}^n \hat{\Phi}_k
&= \sum_{k=1}^n \left(\sum_{j<k} 2\mu_j - \sum_{j>k}2\mu_k\right)= \sum_{k=1}^n\sum_{j<k} 2\mu_j - \sum_{k=1}^n\sum_{j>k}2\mu_k.
\end{align*}
These are all finite sums, so we can reverse the order of summation. For the second double sum, summing from $k=1$ to $n$ with $j>k$ is equivalent to summing from $j=1$ to $n$ with $k<j$. Therefore, the above is equal to{\samepage
\[\sum_{k=1}^n\sum_{j<k} 2\mu_j - \sum_{j=1}^n\sum_{k<j}2\mu_k,\]
which is clearly zero.}

We have now established that $(\hat{\Phi}_i)_{i=1}^n$ the full system of equations in~\eqref{systemofequationsforphaseshifts}. Since this system has a unique solution, and the phase shifts $({\Phi}_i)_{i=1}^n$ satisfy this system, it follows that $\hat{\Phi}_i=\Phi_i$ for $i=1,\dots,n$.
\end{proof}

\section{Comparison with the literature}\label{subsection:tokihiro}
Tokihiro and collaborators \cite{bib:tns} have studied a related advanced version of the box-ball system in which balls are labelled and given priority in their movements based on the labels. Like the classic box-ball system, the advanced box-ball system also exhibits the solitonic scattering phenomenon, and they describe a soliton scattering rule that shows how the ball labels are permuted within the solitonic blocks in both forwards and backwards asymptotic time. Instead of the usual discrete-time Toda lattice with tridiagonal Hessenberg phase space, the advanced box-ball system is obtained via ultradiscretisation of the so-called \textit{hungry Toda molecule equation}. This tropical correspondence allowed the dynamics and scattering rules of the advanced box-ball system to be understood by detropicalising, passing through the hungry Toda dynamics, and then re-tropicalising back to an advanced box-ball configuration. This, when restricted to just one box label, reduces to classical dToda and the classical box-ball system.

For the restriction to dToda and classical BBS, \cite{bib:tns} \textit{observed} that the phase shift phenomenon decomposes into pairwise individual 2-soliton phase interactions and used this observation to justify restricting their attention to just 2-soliton phase shifts. However, this observation does not seem to be proved. With our approach, however, we do not make this assumption or observation \textit{a priori}. Instead, we take an $n$-soliton configuration and prove an overall phase shift formula (equation~\eqref{equation:phaseshiftforbbs}) which we recall below:
\[ \hat{\Phi}_i = \sum_{j<i} 2\mu_j - \sum_{j>i}2\mu_i\]
and, just as Moser did in \cite{bib:moser}, we can interpret the formula as saying that one can think of the overall phase shift phenomenon as decomposing into pairwise phase shift interactions, even though we do not assume that only two blocks collide at a time. However, a major point to be made here is that one can easily construct a box-ball configuration in which multiple collisions occur simultaneously, and our proof here does not depend on assuming that such a multi-collision does not occur; the proof applies to such configurations.

Subsequent to the completion of this paper, we were made aware of another more recent approach to scattering properties of box-ball systems based on a purely combinatorial method known as the Kerov--Kirillov--Reshetikhin bijection. The works of \cite{bib:kosty,bib:takagi} use this bijection to construct action-angle coordinates on (generalised) BBS, involving Young tableaux with \textit{riggings}. This yields a linearisation of the BBS dynamics in terms of which the asymptotic behavior of the soliton blocks can be determined. It may be of interest in future studies to compare this approach with the diagrammatic phase shift analysis performed in this paper.

\appendix

\section{The Gessel--Viennot--Lindstr\"om lemma}
The following lemma of Gessel, Viennot and Lindstr\"om can be found in \cite{bib:aigner}:

\begin{lem}[Gessel--Viennot--Lindstr\"om]\label{gvllemma}
Let $G=(V,E)$ be a directed acyclic graph with weights on $E$, $\mathcal{A}=\{A_1,\dots,A_n\}$ and $\mathcal{B}=\{B_1,\dots,B_n\}$ be subsets of $V$, and $M$ the path-matrix from $\mathcal{A}$ to $\mathcal{B}$. Then,
\begin{equation}\label{eqngvl}
\det(M)=\sum_{\mathcal{P}\in VD} (\operatorname{sign} \mathcal{P})w(\mathcal{P}),
\end{equation}
where $VD$ is the set of all \textit{vertex-disjoint} path systems from $\mathcal{A}\to \mathcal{B}$, i.e., each $\mathcal{P}\in VD$ is an $n$-tuple of paths $P_i\colon A_i\to B_{\sigma(i)}$ for some $\sigma\in\mathfrak{S}_n$ such that $P_i$ and $P_j$ have no vertices in common for $i\neq j$. The \textit{sign} of a vertex-disjoint path system is taken to be the sign of the permutation $\sigma$, and its weight is the product of all weights of all edges used in the path system.
\end{lem}

Of particular interest to us will be path graphs of the following type:

\begin{figure}[H]
\centering
\tikz[scale=2.4]{
\foreach \inc in {0.2}{
\foreach \n in {4}{
\foreach \x in {1,...,\n}{
\foreach \y in {1,...,\n}{
\node at (\y,\n-\x+1) {\color{black!70}{(\x,\y)}};
\ifnum \x=\n
\node at (\y,\n-\x+1) {(\x,\y)};
\fi
\ifnum \y=\n
\node at (\y,\n-\x+1) {(\x,\y)};
\fi
}
}
\foreach \x in {1,...,\n}{
\foreach \y in {1,...,\n}{
\ifnum \y<\n
\ifnum \x>\y+1
\draw [->, thick] (\y,\n-\x+1+\inc) -- (\y,\n-\x+2-\inc);
\fi
\fi
}
}
\foreach \x in {1,...,\n}{
\foreach \y in {1,...,\n}{
\ifnum \y>1
\ifnum \x<\y+1
\draw[->,thick] (\y,\n-\x+\inc) -- (\y,\n-\x+1-\inc);
\fi
\fi
}
}
\foreach \x in {1,...,\n}{
\foreach \y in {1,...,\n}{
\ifnum \x<\n
\ifnum \y>\x+1
\draw[->,thick] (\y-1+\inc,\n-\x+1) -- (\y-\inc,\n-\x+1);
\fi
\fi
}
}
}
}
}
\end{figure}

In general, these will be graphs whose vertices form an $n\x n$ array with the following properties:
\begin{enumerate}\itemsep=0pt
 \item[1)] the vertices are labelled $(i,j)$ for $1\leq i,j\leq n$, with $i$ describing the row, and $j$ the column,
 \item[2)] the bottom row vertices $\{(n,j)\colon 1\leq j\leq n\}$ will be the set of sources,
 \item[3)] the last column vertices $\{(i,n)\colon 1\leq i\leq n\}$ will be the set of sinks, sharing $(n,n)$ with the sources,
 \item[4)] all edges are directed either up or right,
 \item[5)] any path from a source $(n,j)$ to a sink $(i,n)$ must pass through $(j,j)$; therefore, there are only paths from $(n,j)$ to $(i,n)$ for $i\leq j$.
\end{enumerate}

These graphs will be described in more detail (specifying the edge weights), but, for now, these properties are all that are needed to make the following key observation: \textit{a~path from $(n,j)$ to $(i,n)$ separates sources $(n,k)$ from sinks $(l,n)$ for all $k>j$ and $l<i$.}

The following lemma is a consequence of this last observation above.

\begin{lem}\label{lemvertexdisjidentity}
If one has a $k$-set of sources $\mathcal{A}=\{(n,a_i)\colon 1\leq i\leq k\}$ and a $k$-set of sinks $\mathcal{B}=\{(b_i,n)\colon 1\leq i\leq k\}$, each ordered so that $a_1<a_2<\cdots <a_k$ and $b_1<b_2<\cdots <b_k$, then, if there are any vertex-disjoint path systems from $\mathcal{A}$ to $\mathcal{B}$, then $\sigma\in\mathfrak{S}_k$ must be the identity, i.e., the path system must consist of disjoint paths $P_i\colon (n,a_i)\to (b_i,n)$ for $1\leq i\leq k$.
\end{lem}
\begin{proof}
We assume $k>1$, otherwise the statement holds true trivially. Suppose then that $P_1\colon (n,a_1)\to (b_i,n)$ for some $i>1$. By the observation, this path cuts off all sources $(n,a_i)$ for $i>1$ from the sink $(b_1,n)$. This means the path system cannot be completed without violating the vertex-disjointedness. Thus, one must have $P_1\colon (n,a_1)\to (b_1,n)$.

To complete the proof, we show that if $\mathcal{P}=\{P_i\colon (n,a_i)\to (b_{\sigma(i)},n),\, 1\leq i\leq k\}$ is a vertex-disjoint path system, for some $\sigma\in\mathfrak{S}_k$, and $\sigma(i)=i$ for $1\leq i\leq j$, then $\sigma(j+1)=j+1$.

The argument is the analogous to the base case: suppose $\sigma(j+1)\neq j+1$, then, $\sigma(j+1)>j+1$ since $\sigma(i)=i$ for all $i\leq j$. This means that path $P_{j+1}\colon (n,a_{j+1})\to (b_{\sigma(j+1)},n)$ cuts off all sources $(n,a_i)$ for $i>j+1$ from all sinks $(b_l,n)$ for $l<\sigma(j+1)$. In particular, since $j+1<\sigma(j+1)$, this means that no source can be paired with $(j+1,n)$ without violating the vertex-disjoint requirement. Thus, if such a~vertex-disjoint path system were to exist, one would have to have $\sigma(j+1)=j+1$.
Thus, by induction, we have shown that $\sigma(i)=i$ for each~$i$.
\end{proof}

By Lemma \ref{lemvertexdisjidentity}, the $\text{sign}~\mathcal{P}$ piece is always equal to~1 in equation~\eqref{eqngvl} for our graphs of interest, and so we are permitted use of the following simplified statement of the Gessel--Viennot--Lindstr\"om lemma for our graphs:
$\det(M)=\sum_{\mathcal{P}\in VD} w(\mathcal{P})$.

\section{Proofs of key lemmas}

\subsection{Proof of Lemma \ref{toprightUminor}}
For $n\in \N$ with $n>1$, we construct a directed, weighted graph $Q_n$ as follows:
\begin{enumerate}\itemsep=0pt
\item The vertices of the graph will be pairs $(x,y)$ with $1\leq x,y\leq n$ arranged in a rectangular grid with $x$ and $y$ denoting the row and column, respectively, of the vertex.
\item For every $x\in [n-1]$ and for each $y\in[n]$, there will be a vertical arrow up from $(x+1,y)$ to $(x,y)$. If $x\geq y$, this arrow will be assigned a weight of~1. Otherwise, the weight $w_{x,y}^{(n)}$ from $(x+1,y)\to (x,y)$ will be defined by
\[ w_{x,y}^{(n)} = \dfrac{\prod\limits_{m=y+1}^n (\lambda_m-\lambda_{x+1})}{\prod\limits_{p=y}^n (\lambda_p-\lambda_{x})}.\]
\item For each $x$, $y$ with $1\leq x<y \leq n$, there will be an arrow from $(x,y-1)$ to $(x,y)$ with weight~$1$.
\end{enumerate}

For example, below is $Q_4$:

\begin{figure}[h!]
\centering
\tikz[scale=2.4]{
\foreach \inc in {0.2}{
\foreach \n in {4}{
\foreach \x in {1,...,\n}{
\foreach \y in {1,...,\n}{
\node at (\y,\n-\x+1) {\color{black!70}{(\x,\y)}};
\ifnum \x=\n
\node at (\y,\n-\x+1) {\color{red}{(\x,\y)}};
\fi
\ifnum \y=\n
\node at (\y,\n-\x+1) {\color{red}{(\x,\y)}};
\fi
}
}
\foreach \x in {1,...,\n}{
\foreach \y in {1,...,\n}{
\ifnum \y<\n
\ifnum \x>\y+1
\draw [->, cyan, thick] (\y,\n-\x+1+\inc) -- (\y,\n-\x+2-\inc) node[midway,left] {$1$};
\fi
\fi
}
}
\foreach \x in {1,...,\n}{
\foreach \y in {1,...,\n}{
\ifnum \y>1
\ifnum \x<\y+1
\draw[->,magenta,thick] (\y,\n-\x+\inc) -- (\y,\n-\x+1-\inc) node[midway,left] {$w_{\x,\y}^{(\n)}$};;
\fi
\fi
}
}
\foreach \x in {1,...,\n}{
\foreach \y in {1,...,\n}{
\ifnum \x<\n
\ifnum \y>\x+1
\draw[->,thick] (\y-1+\inc,\n-\x+1) -- (\y-\inc,\n-\x+1) node[midway,above] {$1$};
\fi
\fi
}
}
}
}
}
\end{figure}

We associate to $Q_n$ its weighted path matrix $P^{(n)}$ defined by taking $P^{(n)}_{ij}$ to be the sum over all paths to $(i,n)$ from $(n,j)$, where a path's contribution to the sum is taken by multiplying all weights of edges taken by the path. Immediately from the fact that all edges are either up or right, with the only place to go (via unit weighted paths) from $(n,j)$ is to $(j,j)$, and the fact that all horizontal paths have unit weight, it is easy to see that each $P^{(n)}$ is a upper unitary matrix. The first step in proving Lemma \ref{toprightUminor} is to show that we in fact have
$P^{(n)} = U(\lambda_1,\dots,\lambda_n)$
for each $n>1$.

To aid in some of the arguments to come, we include $Q_2$, $Q_3$ and $Q_4$ below with the weight written out in full:

Here is $Q_2$:

\begin{figure}[h!]
\centering
\tikz[yscale=2.8,xscale=4.7]{
\foreach \inc in {0.2}{
\foreach \n in {2}{
\foreach \x in {1,...,\n}{
\foreach \y in {1,...,\n}{
\node at (\y,\n-\x+1) {\color{black!70}{(\x,\y)}};
\ifnum \x=\n
\node at (\y,\n-\x+1) {\color{red}{(\x,\y)}};
\fi
\ifnum \y=\n
\node at (\y,\n-\x+1) {\color{red}{(\x,\y)}};
\fi
}
}
\foreach \x in {1,...,\n}{
\foreach \y in {1,...,\n}{
\ifnum \y<\n
\ifnum \x>\y+1
\draw [->, cyan, thick] (\y,\n-\x+1+\inc) -- (\y,\n-\x+2-\inc) node[midway,left] {$1$};
\fi
\fi
}
}
\foreach \x in {1,...,\n}{
\foreach \y in {1,...,\n}{
\ifnum \y>1
\ifnum \x<\y+1
\draw[->,magenta,thick] (\y,\n-\x+\inc) -- (\y,\n-\x+1-\inc);
\fi
\fi
}
}
\foreach \x in {1,...,\n}{
\foreach \y in {1,...,\n}{
\ifnum \x<\n
\ifnum \y>\x+1
\draw[->,thick] (\y-1+\inc,\n-\x+1) -- (\y-\inc,\n-\x+1) node[midway,above] {$1$};;
\fi
\fi
}
}
}
\node at (2-0.15,1.5) {\color{magenta}{$\frac{1}{\lambda_2-\lambda_1}$}};
}
}
\end{figure}

Here is $Q_3$:

\begin{figure}[h!]
\centering
\tikz[yscale=2.8,xscale=4.7]{
\foreach \inc in {0.2}{
\foreach \n in {3}{
\foreach \x in {1,...,\n}{
\foreach \y in {1,...,\n}{
\node at (\y,\n-\x+1) {\color{black!70}{(\x,\y)}};
\ifnum \x=\n
\node at (\y,\n-\x+1) {\color{red}{(\x,\y)}};
\fi
\ifnum \y=\n
\node at (\y,\n-\x+1) {\color{red}{(\x,\y)}};
\fi
}
}
\foreach \x in {1,...,\n}{
\foreach \y in {1,...,\n}{
\ifnum \y<\n
\ifnum \x>\y+1
\draw [->, cyan, thick] (\y,\n-\x+1+\inc) -- (\y,\n-\x+2-\inc) node[midway,left] {$1$};
\fi
\fi
}
}
\foreach \x in {1,...,\n}{
\foreach \y in {1,...,\n}{
\ifnum \y>1
\ifnum \x<\y+1
\draw[->,magenta,thick] (\y,\n-\x+\inc) -- (\y,\n-\x+1-\inc);
\fi
\fi
}
}
\foreach \x in {1,...,\n}{
\foreach \y in {1,...,\n}{
\ifnum \x<\n
\ifnum \y>\x+1
\draw[->,thick] (\y-1+\inc,\n-\x+1) -- (\y-\inc,\n-\x+1) node[midway,above] {$1$};;
\fi
\fi
}
}
}
\node at (2-0.27,2.5) {\color{magenta}{$\frac{\lambda_3-\lambda_2}{(\lambda_3-\lambda_1)(\lambda_2-\lambda_1)}$}};
\node at (3-0.13,2.5) {\color{magenta}{$\frac{1}{\lambda_3-\lambda_1}$}};
\node at (3-0.13,1.5) {\color{magenta}{$\frac{1}{\lambda_3-\lambda_2}$}};
}
}
\end{figure}

Here is $Q_4$:

\begin{figure}[h!]
\centering
\tikz[yscale=2.3,xscale=3.5]{
\foreach \inc in {0.2}{
\foreach \n in {4}{
\foreach \x in {1,...,\n}{
\foreach \y in {1,...,\n}{
\node at (\y,\n-\x+1) {\color{black!70}{(\x,\y)}};
\ifnum \x=\n
\node at (\y,\n-\x+1) {\color{red}{(\x,\y)}};
\fi
\ifnum \y=\n
\node at (\y,\n-\x+1) {\color{red}{(\x,\y)}};
\fi
}
}
\foreach \x in {1,...,\n}{
\foreach \y in {1,...,\n}{
\ifnum \y<\n
\ifnum \x>\y+1
\draw [->, cyan, thick] (\y,\n-\x+1+\inc) -- (\y,\n-\x+2-\inc) node[midway,left] {$1$};
\fi
\fi
}
}
\foreach \x in {1,...,\n}{
\foreach \y in {1,...,\n}{
\ifnum \y>1
\ifnum \x<\y+1
\draw[->,magenta,thick] (\y,\n-\x+\inc) -- (\y,\n-\x+1-\inc);
\fi
\fi
}
}
\foreach \x in {1,...,\n}{
\foreach \y in {1,...,\n}{
\ifnum \x<\n
\ifnum \y>\x+1
\draw[->,thick] (\y-1+\inc,\n-\x+1) -- (\y-\inc,\n-\x+1) node[midway,above] {$1$};;
\fi
\fi
}
}
}
\node at (2-0.45,3.5) {\color{magenta}{$\frac{(\lambda_4-\lambda_2)(\lambda_3-\lambda_2)}{(\lambda_4-\lambda_1)(\lambda_3-\lambda_1)(\lambda_2-\lambda_1)}$}};
\node at (3-0.31,3.5) {\color{magenta}{$\frac{\lambda_4-\lambda_2}{(\lambda_4-\lambda_1)(\lambda_3-\lambda_1)}$}};
\node at (3-0.31,2.5) {\color{magenta}{$\frac{\lambda_4-\lambda_3}{(\lambda_4-\lambda_2)(\lambda_3-\lambda_2)}$}};
\node at (4-0.13,3.5) {\color{magenta}{$\frac{1}{\lambda_4-\lambda_1}$}};
\node at (4-0.13,2.5) {\color{magenta}{$\frac{1}{\lambda_4-\lambda_2}$}};
\node at (4-0.13,1.5) {\color{magenta}{$\frac{1}{\lambda_4-\lambda_3}$}};
}
}
\end{figure}

We show that $P^{(n)}=U(\lambda_1,\dots,\lambda_n)$ by induction on $n$.
One can immediately see from $Q_2$ what $P^{(2)}$ is
\[\left[\begin{matrix} 1 & \dfrac{1}{\lambda_2-\lambda_1} \\ 0 & 1\end{matrix}\right],\]
where we take the convention that there is a path from $(2,2)$ to $(2,2)$ by which uses no arrows, and the product over the weights is the empty one, hence is equal to $1$. This is $U(\lambda_1,\lambda_2)$, so we have the base case of $n=2$.

For our induction hypothesis, assume $P^{(n)}=U(\lambda_1,\dots,\lambda_n)$ for some $n>1$ and consider the weighted graph $Q_{n+1}$. We make two observations on how to translate paths from $Q_n$ to $Q_{n+1}$. The first is that the lower-right subgraph of $Q_{n+1}$ obtained by deleting all vertices $(x,y)$ with $x=1$ or $y=1$ is precisely $Q_n$, but with all vertices $(x,y)$ replaced by $(x-1,y-1)$ and the indices of all eigenvalues reduced by~$1$. To see this, suppose $1<x<y\leq n+1$ and consider~$w_{x,y}^{(n+1)}$:
\[w_{x,y}^{(n+1)} = \dfrac{\prod\limits_{m=y+1}^{n+1} (\lambda_m-\lambda_{x+1})}{\prod\limits_{p=y}^{n+1} (\lambda_p-\lambda_{x})}\]
and compare this to $w_{x-1,y-1}^{(n)}$:
\[ w_{x-1,y-1}^{(n)} = \dfrac{\prod\limits_{m=y}^n (\lambda_m-\lambda_{x})}{\prod\limits_{p=y-1}^n (\lambda_p-\lambda_{x-1})}.\]
If we add $1$ to each index on the eigenvalues in $w_{x-1,y-1}^{(n)}$ and reindex, we obtain
\[\dfrac{\prod\limits_{m=y}^n (\lambda_{m+1}-\lambda_{x+1})}{\prod\limits_{p=y-1}^n (\lambda_{p+1}-\lambda_{x})}
=
\dfrac{\prod\limits_{m=y+1}^{n+1} (\lambda_{m}-\lambda_{x+1})}{\prod\limits_{p=y}^{n+1} (\lambda_{p}-\lambda_{x})}
=
w_{x,y}^{(n+1)}.\]
This proves the claim. Therefore, to see what $P_{ij}^{(n+1)}$ is for $i,j>1$, we note that this observation gives us{\samepage
\[P_{ij}^{(n+1)} = \widehat{P_{i-1,j-1}^{(n)}},\]
where the right-hand side's hat tells us to modify all eigenvalues by increasing their indices by~$1$.}

By the induction hypothesis, we have
\begin{align*}
P_{ij}^{(n+1)} &= (U(\lambda_2,\dots,\lambda_n))_{i-1,j-1}
= \prod_{i-1\leq m<j-1} \frac{1}{\lambda_{(j-1)+1}-\lambda_{m+1}}\\
&= \prod_{i\leq m' <j}\frac{1}{\lambda_j-\lambda_{m'}}\
= (U(\lambda_1,\dots,\lambda_{n+1}))_{ij}.
\end{align*}
This proves the induction hypothesis for the entries $P^{(n+1)}_{ij}$ with $i,j>1$.

For $j=1$ or $n+1$, the structure of $Q_{n+1}$ makes things simple:
\begin{enumerate}\itemsep=0pt
\item There is only one vertex of the form $(i,n+1)$ that can be reached from $(n+1,1)$, and that is $(1,n+1)$, which is via $n$ up arrows (each with unit weight) followed by $n$ right arrows (also with unit weight). Therefore, $P_{1j}^{(n+1)} = \delta_{1j}$, as required.
\item Each of $(i,n+1)$ can be reached by $(n+1,n+1)$, just by vertical arrows. If $i=n+1$, this is via the empty path, with weight $1$. Otherwise, the path from $(n+1,n+1)$ to $(i,n+1)$ has weight
\[\prod_{i\leq m< n+1} w_{i,n+1}^{(n+1)}.\]
By calculating
\[ w_{i,n+1}^{(n+1)} = \dfrac{\prod\limits_{m=n+2}^{n+1} (\lambda_m-\lambda_{i+1})}{\prod\limits_{p=n+1}^{n+1} (\lambda_p-\lambda_{i})}
= \dfrac{1}{\lambda_{n+1}-\lambda_i},\]
we see that
\[\prod_{i\leq m< n+1} w_{i,n+1}^{(n+1)} = \prod_{i\leq m< n+1} \dfrac{1}{\lambda_{n+1}-\lambda_i}= (U(\lambda_1,\dots,\lambda_{n+1}))_{i,n+1}.\]
\end{enumerate}

All that remains of the induction step is to show that
\[P_{1,j}^{(n+1)}= (U(\lambda_1,\dots,\lambda_{n+1}))_{1j}\] for $1<j<n+1$. To see, this, we now compare the subgraph of $Q_{n+1}$, obtained by removing the last column and last row, to the entirety of $Q_n$. We note that they have the same form (with unit weights in the same locations) and compare $w_{x,y}^{(n)}$ and $w_{x,y}^{(n+1)}$ directly by division for $1\leq x<y\leq n$:
\[\dfrac{w_{x,y}^{(n+1)}}{w_{x,y}^{(n)}}
\dfrac{\prod\limits_{m=y+1}^{n+1} (\lambda_m-\lambda_{x+1})}{\prod\limits_{p=y}^{n+1} (\lambda_p-\lambda_{x})}
\dfrac{\prod\limits_{p=y}^n (\lambda_p-\lambda_{x})}{\prod\limits_{m=y+1}^n (\lambda_m-\lambda_{x+1})}=\dfrac{\lambda_{n+1}-\lambda_{x+1}}{\lambda_{n+1}-\lambda_x}.\]

We see that this does not depend on $y$, so, as long as $x<y$, the (multiplicative) difference in the weights between $Q_n$ and the subgraph of $Q_{n+1}$ is a fixed quantity depending only on the row of the weight. Furthermore, since each path from a node $(n,j)$ to $(i,n)$ use precisely the $j-i$ non-unit weight belonging to the consecutive rows terminating in row $i$, we immediately see that the paths in $Q_{n+1}$ from $(j,n+1)$ to $(n,i)$ must be given by
\[\dfrac{\lambda_{n+1}-\lambda_j}{\lambda_{n+1}-\lambda_i}\times (U(\lambda_1,\dots,\lambda_n))_{ij},\]
where the first term comes from multiplying together
\[\dfrac{\lambda_{n+1}-\lambda_{j}}{\lambda_{n+1}-\lambda_{j-1}}\times \dfrac{\lambda_{n+1}-\lambda_{j-1}}{\lambda_{n+1}-\lambda_{j-2}}\times \cdots \times
\dfrac{\lambda_{n+1}-\lambda_{i+2}}{\lambda_{n+1}-\lambda_{i+1}}\times \dfrac{\lambda_{n+1}-\lambda_{i+1}}{\lambda_{n+1}-\lambda_{i}}.\]
We use this information to obtain $P_{1j}^{(n+1)}$ for $1<j<n+1$ by partitioning the paths from $(n+1,j)$ to $(1,n+1)$ based on the last vertex in column $n$ of $Q_{n+1}$ that the path hits before moving on to column $n+1$ (i.e., the next arrow must be a right one). These paths then have only one way of being completed to a path to $(1,n+1)$, which is by a upwards arrows the rest of the way. Such paths contribute a weight of
\[\prod_{1\leq p<i} w_{p,n+1}^{(n+1)}= \prod_{1\leq p<i} \dfrac{1}{\lambda_{n+1}-\lambda_p}.\]
Putting this together with the information from the partition, we have
\[P_{1j}^{n+1} = \sum_{m=1}^j \dfrac{\lambda_{n+1}-\lambda_j}{\lambda_{n+1}-\lambda_m}\left( \prod_{1\leq p<m} \dfrac{1}{\lambda_{n+1}-\lambda_p}\right)((U(\lambda_1,\dots,\lambda_n))_{mj}.\]
Therefore, to complete the final step of the induction step amounts to proving the following identity:
\[\prod_{1\leq m<j} \dfrac{1}{\lambda_j-\lambda_m} =
\sum_{m=1}^j \dfrac{\lambda_{n+1}-\lambda_j}{\lambda_{n+1}-\lambda_m}\left( \prod_{1\leq p<m} \dfrac{1}{\lambda_{n+1}-\lambda_p}\right)\prod_{m\leq l<j}\dfrac{1}{\lambda_j-\lambda_l}.\]

\begin{lem}
For all $n\in\N$ and $1<j<n$,
\[\prod_{1\leq m<j} \dfrac{1}{\lambda_j-\lambda_m} =
\sum_{m=1}^j \dfrac{\lambda_{n}-\lambda_j}{\lambda_{n}-\lambda_m}\left( \prod_{1\leq p<m} \dfrac{1}{\lambda_{n}-\lambda_p}\right)\prod_{m\leq l<j}\dfrac{1}{\lambda_j-\lambda_l}.
\]
\end{lem}
\begin{proof}
By multiplying both sides by
\[\prod_{1\leq i <j}(\lambda_n-\lambda_i)(\lambda_j-\lambda_i),\]
the above is equivalent to showing
\begin{align*}
 \prod_{i=1}^{j-1}(\lambda_n-\lambda_i)
 &= \sum_{m=1}^j \dfrac{\lambda_{n}-\lambda_j}{\lambda_{n}-\lambda_m}\left(\prod_{p=m}^{j-1}(\lambda_n-\lambda_p)\right)\left(\prod_{l=1}^{m-1}(\lambda_j-\lambda_l)\right)\\
&= \sum_{m=1}^j \left(\prod_{p=m+1}^{j}(\lambda_n-\lambda_p)\right)\left(\prod_{l=1}^{m-1}(\lambda_j-\lambda_l)\right).
\end{align*}

This holds trivially for $j=1$ (by convention for the empty product). For $j=2$, we want
\[\lambda_n-\lambda_1 = (\lambda_n-\lambda_2)\times 1+1\times (\lambda_2-\lambda_1).\]
We proceed by induction: suppose the above identity holds for each $j-1$ and consider the identity for~$j$. The right-hand side is equal to
\begin{align*}
\text{RHS}
&=
\sum\limits_{m=1}^{j-2} \left(\prod\limits_{p=m+1}^{j}(\lambda_n-\lambda_p)\right)\left(\prod\limits_{l=1}^{m-1}(\lambda_j-\lambda_l)\right)+ \prod\limits_{l=1}^{j-1} (\lambda_j-\lambda_l)
+(\lambda_n-\lambda_j)\prod\limits_{l=1}^{j-2}(\lambda_j-\lambda_l)
\\
&= \sum\limits_{m=1}^{j-2} \left(\prod\limits_{p=m+1}^{j}(\lambda_n-\lambda_p)\right)\left(\prod\limits_{l=1}^{m-1}(\lambda_j-\lambda_l)\right)
+ (\lambda_j-\lambda_{j-1} +\lambda_n-\lambda_j)\prod\limits_{l=1}^{j-2}(\lambda_j-\lambda_l)\\
&= (\lambda_n-\lambda_{j-1})\left[
\sum_{m=1}^{j-2} \left(\prod_{\substack{p=m+1\\p\neq j-1}}^{j}(\lambda_n-\lambda_p)\right)\left(\prod_{l=1}^{m-1}(\lambda_j-\lambda_l)\right)
+ \prod_{l=1}^{j-2}(\lambda_j-\lambda_l)
\right]\\
&= (\lambda_n-\lambda_{j-1})\prod_{i=1}^{j-2}(\lambda_n-\lambda_i)= \prod_{i=1}^{j-1}(\lambda_n-\lambda_i),
\end{align*}
where the penultimate equality was the induction hypothesis applied to the eigenvalues
\[\lambda_1,\dots,\lambda_{j-2},\lambda_j,\dots,\lambda_n.\tag*{\qed}\]\renewcommand{\qed}{}
\end{proof}

With the equality $U(\lambda_1,\dots,\lambda_n)=P^{(n)}$ for all $n>1$ now established, we have that
\[\Delta_k^n U(\lambda_1,\dots,\lambda_n) = \Delta_k^n\big(P^{(n)}\big).\]
By the Gessel--Viennot--Lindstr\"om lemma, this amounts to considering $k$-tuples of non-in\-ter\-sec\-ting paths from the set $\{(n,j)\colon n-k+1\leq j\leq n\}$ to the set $\{(i,n)\colon 1\leq i\leq k\}$. It is clear that the tuple must have one of its paths being $(n,n)\to (k,n)$ since one can only go straight up from $(n,n)$ and one must avoid any intersections. Similarly, $(n,n-1)$ must be linked with $(k-1,n)$, and this path consists of going up to the $(k-1)$-st row, then across to $(k-1,n)$. Continuing in a similar vein, we see that the only $k$-tuple of non-intersecting paths available to us is the collection of up-right hooks from $(n,n-m+1)$ to $(m,n)$ for $1\leq m\leq k$.

For $j\in\{n-k+1,n-k+2,\dots,n\}$, we must traverse $n-j$ unit weight arrows to get to $(j,j)$, after which we must traverse vertical non-unit arrows to get to row $j+k-n$. The product along these weighted arrows is then
\[\prod_{p=j+k-n}^{j-1} w_{p,j}^{(n)}.\]
By noting that
\[w_{p,j}^{(n)} = \dfrac{\prod\limits_{m=j+1}^n (\lambda_j -\lambda_{p+1})}{\prod\limits_{m=j}^n (\lambda_m-\lambda_p)} = \dfrac{1}{\lambda_j-\lambda_p} \prod_{m=j+1}^n \dfrac{\lambda_m-\lambda_{p+1}}{\lambda_m-\lambda_p},\]
the latter part of which telescopes over multiplying for consecutive values of~$p$, we see that
\[\prod_{p=j+k-n}^{j-1} w_{p,j}^{(n)}
=
\left(\prod_{i=l+j-n}^{j-1} \dfrac{1}{\lambda_j-\lambda_i}\right)\prod_{m=j+1}^n \dfrac{\lambda_m-\lambda_j}{\lambda_m-\lambda_{j+k-n}}.
\]

We proceed to prove Lemma \ref{toprightUminor} showing that the result for $k-1$ and $n-1$ implies the result for $k,n$, and by proving that the result holds true for all $n$ when $k=1$. Since the lemma only concerns $k<n$, this will prove all cases of interest.

For $k=1$, we see this directly since
\[(U(\lambda_1,\dots,\lambda_n))_{1n}= \prod_{1\leq m<n}\dfrac{1}{\lambda_n-\lambda_m} = \prod_{j=n-1+1}^n \prod_{i=1}^{n-1}\dfrac{1}{\lambda_j-\lambda_i}.\]

For the induction hypothesis, suppose it is known that{\samepage
\[\Delta_k^n (U(\lambda_1,\dots,\lambda_n)) = \prod_{j=n-k+1}^n \prod_{i=1}^{n-k}\dfrac{1}{\lambda_j-\lambda_i}\]
for some $1\leq k<n$.}

Consider the graph $Q_{n+1}$ along with the $(k+1)$-tuple of non-intersecting paths from the set $\{(n+1,j)\colon n-k+1\leq j\leq n+1\}$ to the set $\{(i,n+1)\colon 1\leq i\leq k+1\}$. This $(k+1)$-tuple decomposes into a ``large hook'' from $(n+1,n-k+1)$ to $(1,n+1)$, and a $k$-tuple of paths from $\{(n+1,j)\colon n-k\leq j\leq n+1\}$ to the set $\{(i,n+1)\colon 2\leq i\leq k+1\}$.

By our earlier observation of how $Q_n$ sits in the lower-right part of $Q_{n+1}$ (we reduce all indices by one, both in the coordinates and in the eigenvalues), the induction hypothesis (applied to~$Q_n$, then translated over to~$Q_{n+1}$) tells us that this $k$-tuple of paths has product
\[ \Delta_k^n(U(\lambda_2,\dots,\lambda_{n+1})) = \prod_{j=n-k+1}^n \prod_{i=1}^{n-k} \dfrac{1}{\lambda_{j+1}-\lambda_{i+1}}.\]

The ``large hook'' has weight
\[\prod_{p=1}^{n-k}w_{p,n-k+1}^{(n+1)} = \left(\prod_{i=1}^{n-k}\dfrac{1}{\lambda_{n-k+1}-\lambda_i}\right)\prod_{m=n-k+2}^{n+1}\dfrac{\lambda_m-\lambda_{n-k+1}}{\lambda_m-\lambda_1}.\]

Thus, by the Gessel--Viennot--Lindstr\"om lemma and the induction hypothesis, we have that $\Delta_{k+1}^{n+1}(U(\lambda_1,\dots,\lambda_{n+1}))$ is equal to the following:
\begin{align*}
& \left( \prod_{j=n-k+1}^n \prod_{i=1}^{n-k} \dfrac{1}{\lambda_{j+1}-\lambda_{i+1}}\right)\left(\prod_{i=1}^{n-k}\dfrac{1}{\lambda_{n-k+1}-\lambda_i}\right)\prod_{m=n-k+2}^{n+1}\dfrac{\lambda_m-\lambda_{n-k+1}}{\lambda_m-\lambda_1}\\
&\qquad{}= \left( \prod_{j=n-k+2}^{n+1} \prod_{i=1}^{n-k} \dfrac{1}{\lambda_{j}-\lambda_{i+1}}\right)\left(\prod_{i=1}^{n-k}\dfrac{1}{\lambda_{n-k+1}-\lambda_i}\right)\prod_{m=n-k+2}^{n+1}\dfrac{\lambda_m-\lambda_{n-k+1}}{\lambda_m-\lambda_1}\\
&\qquad{}= \left( \prod_{j=n-k+2}^{n+1} \prod_{i=1}^{n-k-1} \dfrac{1}{\lambda_{j}-\lambda_{i+1}}\right)\left(\prod_{i=1}^{n-k}\dfrac{1}{\lambda_{n-k+1}-\lambda_i}\right)\prod_{m=n-k+2}^{n+1}\dfrac{1}{\lambda_m-\lambda_1}\\
&\qquad{}= \left( \prod_{j=n-k+2}^{n+1} \prod_{i=0}^{n-k-1} \dfrac{1}{\lambda_{j}-\lambda_{i+1}}\right)\left(\prod_{i=1}^{n-k}\dfrac{1}{\lambda_{n-k+1}-\lambda_i}\right)\\
&\qquad{}= \left( \prod_{j=n-k+2}^{n+1} \prod_{i=1}^{n-k} \dfrac{1}{\lambda_{j}-\lambda_{i}}\right)\left(\prod_{i=1}^{n-k}\dfrac{1}{\lambda_{n-k+1}-\lambda_i}\right)\\
&\qquad{}= \prod_{j=n-k+1}^{n+1} \prod_{i=1}^{n-k} \dfrac{1}{\lambda_{j}-\lambda_{i}},
\end{align*}
which is precisely the formula we wished to prove for $\Delta_{k+1}^{n+1} U(\lambda_1,\dots,\lambda_{n+1})$. This completes the proof of Lemma \ref{toprightUminor}.

\subsection{Proof of Lemma \ref{Uinverselemma}}
\begin{proof}
Let $U$ and $V$ be the matrices defined by
\[(U)_{ij} = \begin{cases}
\displaystyle \prod\limits_{i\leq m<j} \dfrac{1}{\lambda_j-\lambda_m} & \text{for }i\leq j,\\
0 & \text{otherwise},
\end{cases}
\]
and
\[(V)_{ij} = \begin{cases}
0, & i>j,\\
\displaystyle \prod_{i<m\leq j}\dfrac{1}{\lambda_i-\lambda_m}, & i\leq j.
\end{cases}
\]
We first remark that the entries of $U$ and $V$ only depend on the indices $(i,j)$ and not on~$n$. Therefore, we see that the matrices have the following block structure:
\[
U(\lambda_1,\dots,\lambda_n) = \left[\begin{array}{c|c}
U(\lambda_1,\dots,\lambda_{n-1}) & \vec{u}_n\\
\hline 0 & 1
\end{array}\right],\qquad
V = \left[\begin{array}{c|c}
V_{[n-1],[n-1]} & \vec{v}_n\\
\hline 0 & 1
\end{array}\right],
\]
where
\[(\vec{u}_n)_i = \prod_{i\leq m<n}\dfrac{1}{\lambda_n-\lambda_m},\qquad i=1,2,\dots,n-1,\]
and
\[(\vec{v}_n)_i = \prod_{i< m\leq n}\dfrac{1}{\lambda_i-\lambda_m},\qquad i=1,2,\dots,n-1.\]
Multiplying these matrices in these block forms yields:
\[
U(\lambda_1,\dots,\lambda_n)\times V = \left[\begin{array}{c|c}
U(\lambda_1,\dots,\lambda_{n-1})V_{[n-1],[n-1]} & U(\lambda_1,\dots,\lambda_{n-1})\vec{v}_n + \vec{u}_n\\
\hline 0 & 1
\end{array}\right].
\]
Thus, we see that an induction would be appropriate here: if we have shown the result for $n-1$, then the top-left block is the $(n-1)\x (n-1)$ identity matrix, so we need only show that
\[U(\lambda_1,\dots,\lambda_{n-1})\vec{v}_n + \vec{u}_n=\vec{0}.\]
The $n=1$ (base) case trivially holds, so we can proceed with the induction step for $n>1$. We then need only establish the following to complete the proof:
\[0 = \prod_{i\leq m<n}\dfrac{1}{\lambda_n-\lambda_m} + \sum_{j=i}^{n-1}\left(\prod\limits_{i\leq p<j} \dfrac{1}{\lambda_j-\lambda_p} \right)\left(\prod_{j< q\leq n}\dfrac{1}{\lambda_j-\lambda_q}\right),\]
which needs to be shown for each $i=1,2,\dots, n-1$.

Viewing the right-hand side as a rational function of $\lambda_n$, we see that one should consider the rational function
\[\prod_{i\leq m<n} \dfrac{1}{x-\lambda_m}.\]
We set up a partial fraction decomposition:
\[\prod_{i\leq m<n} \dfrac{1}{x-\lambda_m} = \sum_{j=i}^{n-1} \dfrac{c_j}{x-\lambda_j}.\]
We multiply both sides by $\prod_{i\leq m<n} (x-\lambda_m)$:
\[1 = \sum_{j=i}^{n-1} c_j\prod_{\substack{i\leq m<n\\m\neq j}} (x-\lambda_m).\]
Plugging in $x=\lambda_j$, we see that
\[c_j\prod_{\substack{i\leq m<n\\m\neq j}} (\lambda_j-\lambda_m) = 1\]
for each $j$. Hence,
\[c_j = \prod_{\substack{i\leq m<n\\m\neq j}} \dfrac{1}{\lambda_j-\lambda_m} = \left(\prod\limits_{i\leq p<j} \dfrac{1}{\lambda_j-\lambda_p} \right)\left(\prod_{j< q< n}\dfrac{1}{\lambda_j-\lambda_q}\right).\]
To conclude, we have
\begin{align*}
0 & = 1 - \sum_{j=i}^{n-1} \left(\prod\limits_{i\leq p<j} \dfrac{1}{\lambda_j-\lambda_p} \right)\left(\prod_{j< q< n}\dfrac{1}{\lambda_j-\lambda_q}\right)\dfrac{1}{\lambda_n-\lambda_j}\\
& = 1 + \sum_{j=i}^{n-1} \left(\prod\limits_{i\leq p<j} \dfrac{1}{\lambda_j-\lambda_p} \right)\left(\prod_{j< q< n}\dfrac{1}{\lambda_j-\lambda_q}\right)\dfrac{1}{\lambda_j-\lambda_n}\\
& = 1 + \sum_{j=i}^{n-1} \left(\prod\limits_{i\leq p<j} \dfrac{1}{\lambda_j-\lambda_p} \right)\left(\prod_{j< q\leq n}\dfrac{1}{\lambda_j-\lambda_q}\right).
\end{align*}
This completes the induction, and therefore proves Lemma \ref{Uinverselemma}.
\end{proof}

\subsection{Proving Corollary \ref{toprightUinverseminor}}

\begin{lem}
Let $V$ be the matrix as in Lemma {\rm \ref{Uinverselemma}}:
\[(V)_{ij} = \begin{cases}
0, & i>j,\\
\displaystyle \prod_{i<m\leq j}\dfrac{1}{\lambda_i-\lambda_m}, & i\leq j.
\end{cases}
\]
Then
\[ V_{ij} = \widehat{w}_0 U(\lambda_n,\dots,\lambda_1)^{\rm T} \widehat{w}_0.\]
\end{lem}
\begin{proof}
Recall how we defined $U(\lambda_1,\dots,\lambda_n)$:
\[(U)_{ij} = \begin{cases}
\displaystyle \prod\limits_{i\leq m<j} \dfrac{1}{\lambda_j-\lambda_m} & \text{for }i\leq j,\\
0 & \text{otherwise}.
\end{cases}
\]
Therefore,
\[(U(\lambda_n,\dots,\lambda_1))_{ij} = \begin{cases}
\displaystyle \prod\limits_{i\leq m<j} \dfrac{1}{\lambda_{n-j+1}-\lambda_{n-m+1}} & \text{for }i\leq j,\\
0 & \text{otherwise}.
\end{cases}
\]

We now prove that $V = \widehat{w}_0 U(\lambda_n,\dots,\lambda_1)^{\rm T} \widehat{w}_0$ by direct calculation:
\begin{align*}
\big(\widehat{w}_0 U(\lambda_n,\dots,\lambda_1)^{\rm T} \widehat{w}_0\big)_{ij}
& = \sum_{l=1}^n (\widehat{w}_0)_{il}\big(U(\lambda_n,\dots,\lambda_1)^{\rm T} \widehat{w}_0\big)_{lj}
 =\big(U(\lambda_n,\dots,\lambda_1)^{\rm T} \widehat{w}_0\big)_{n-i+1,j}\\
&=\sum_{m=1}^n \big(U(\lambda_n,\dots,\lambda_1)^{\rm T}\big)_{n-i+1,m} (\widehat{w}_0)_{m,j}\\
&=\big(U(\lambda_n,\dots,\lambda_1)^{\rm T}\big)_{n-i+1,n-j+1}
 =(U(\lambda_n,\dots,\lambda_1))_{n-j+1,n-i+1}\\
&= \begin{cases}
\displaystyle \prod\limits_{n-j+1\leq m<n-i+1} \dfrac{1}{\lambda_{i}-\lambda_{n-m+1}} & \text{if }i\leq j,\\
0 & \text{otherwise}
\end{cases}\\
&= \begin{cases}
\displaystyle \prod\limits_{i<m\leq j} \dfrac{1}{\lambda_{i}-\lambda_{m}} & \text{for }i\leq j,\\
0 & \text{otherwise},
\end{cases}
\end{align*}
which is indeed $(V)_{ij}$.
\end{proof}

Since $\widehat{w}_0^{\rm T}=\widehat{w}_0$, the above result can be equivalently expressed as
\begin{equation*}
 U(\lambda_1,\dots,\lambda_n)^{-1} = (\widehat{w}_0U(\lambda_n,\dots,\lambda_1)\widehat{w}_0)^{\rm T}.
\end{equation*}

We now use this lemma to prove Corollary \ref{toprightUinverseminor}:

\begin{proof}
Since $V=U(\lambda_1,\dots,\lambda_n)^{-1}$, $\Delta_k^n(V)= \Delta_k^n\big(U^{-1}\big)$ can be computed as follows
\begin{align*}
\Delta_k^n(V)
&= \Delta_k^n\big(\widehat{w}_0 U(\lambda_n,\dots,\lambda_1)^{\rm T} \widehat{w}_0\big)\\
&= \sum_{S\in {[n]\choose k}} \det((\widehat{w}_0)_{[k],S})\det\big((U(\lambda_n,\dots,\lambda_1)^{\rm T}\widehat{w}_0)_{S,]k[}\big)\\
&= (-1)^{\lfloor k/2\rfloor}\det\big((U(\lambda_n,\dots,\lambda_1)^{\rm T}\widehat{w}_0)_{]k[,]k[}\big)\\
&= (-1)^{\lfloor k/2\rfloor} \sum_{S\in {[n]\choose k}} \det\big(U(\lambda_n,\dots,\lambda_1)^{\rm T}\big)_{]k[,S})\det((\widehat{w}_0)_{S,]k[})\\
&= (-1)^{2\lfloor k/2\rfloor} \det\big(U(\lambda_n,\dots,\lambda_1)^{\rm T}\big)_{]k[,[k]})\\
&= \det((U(\lambda_n,\dots,\lambda_1))_{[k],]k[}).
\end{align*}

Therefore, by Lemma \ref{toprightUminor}, we can deduce
$$\Delta_k^n(U^{-1}) = \Delta_k^n(V) = \prod_{j=n-k+1}^n \prod_{i=1}^{n-k}\dfrac{1}{\lambda_{n-j+1}-\lambda_{n-i+1}} = \prod_{j=1}^k \prod_{i=k+1}^{n}\dfrac{1}{\lambda_j-\lambda_i},$$
thus, completing the proof of Corollary \ref{toprightUinverseminor}.
\end{proof}

\subsection*{Acknowledgements}
This work was supported by NSF grant DMS-1615921. We thank the referees for their very careful reading of the manuscript.

\pdfbookmark[1]{References}{ref}
\LastPageEnding

\end{document}